\documentclass[preprint,12pt]{elsarticle}
\usepackage{graphicx}
\usepackage{amssymb,mathtools}
\usepackage{listings}
\usepackage{natbib}
\usepackage{float}
\usepackage{lineno}
\usepackage{graphicx,color,colordvi}
\usepackage{latexsym,amsbsy,epsfig,fancybox}
\usepackage{amstext}
\usepackage{amsfonts,textgreek}
\usepackage{mathrsfs}
\usepackage{mathtools}
\usepackage{stmaryrd,dsfont}
\usepackage{bbm}
\usepackage{url}
\usepackage{wrapfig,bm}
\usepackage{tikz}
\usepackage{pstricks}
\usepackage{caption}
\usepackage{psfrag}
\usepackage{upgreek}
\usepackage{isomath}
\usepackage{latexsym}
\usepackage{array}
\usepackage{multirow}
\usepackage{booktabs}
\usepackage{verbatim}
\usepackage{natbib}
\biboptions{sort&compress}
\usepackage{color}
\usepackage{colortbl}
\usepackage[colorlinks=true,citecolor=blue]{hyperref}
\usepackage{overpic}
\usepackage{enumerate}
\usepackage{amsmath,amssymb,xspace}
\usetikzlibrary{plotmarks}
\usetikzlibrary{fit,positioning,arrows}
\usetikzlibrary{shapes.geometric, plotmarks}
\usepackage{appendix}

\usepackage[margin=3cm]{geometry}
\usepackage{soul}

\tikzstyle{nicebox}=[draw=black!100, fill=white!10, rectangle, inner sep=4pt, inner ysep=16pt]
\tikzstyle{niceboxtitle}=[draw=black!100, fill=white, text=black, rectangle]

\usetikzlibrary{arrows,decorations.markings}
\newcommand{\tb}[1]{\textcolor{black}{#1}}

\newcommand{\ac}[1]{\textcolor{red}{add citation}}
\journal{Journal of the Mechanics and Physics of Solids}
\usepackage{enumerate}
\usepackage{color}
\usepackage{multirow}\usepackage{graphicx}
\usepackage{amsfonts}
\usepackage{amsmath}
\usepackage{amssymb}

\usepackage{color}
\usepackage{dsfont} 
\usepackage{pgfplots}
\usepackage{subfigure}
\usetikzlibrary{patterns}
\usepackage{amssymb}
\usepackage{graphicx}
\usepackage{amsmath}
\usepackage{bm}
\usepackage{algorithm} 
\usepackage{algorithmic}  
\usepackage[algo2e]{algorithm2e} 
\usepackage{caption}
  
\usetikzlibrary{decorations.pathreplacing}

\usetikzlibrary{calc,patterns}
\usepackage{url}
\usepackage{fancyhdr}
\usepackage{indentfirst}

\usepackage{enumerate}
\usepackage{dsfont}
\usepackage[british]{babel}
\usepackage[colorlinks=true,citecolor=blue]{hyperref}
\usepackage{amsthm}
\usepackage{color}
\usepackage{pifont}
\usepackage{comment}
\usepackage{tikz-cd}
\usepackage{graphicx}
\usepackage{amsfonts}
\usepackage{amsmath}
\usepackage{amssymb}
\usepackage{url}
\usepackage{bbm}
\usepackage{fancyhdr}
\usepackage{indentfirst}
\usepackage{enumerate}
\usepackage{dsfont}
\usepackage[british]{babel}
\pgfplotsset{compat=1.7}
\usepackage[colorlinks=true,citecolor=blue]{hyperref}
\usepackage{cleveref}
\usepackage{amsthm}
\usepackage{color}
\renewcommand{\cite}{\citet}
\usepackage{comment}
\usepackage{algorithm}

\renewcommand{\d}{\,\mathrm{d}}
\newcommand{\dd}{\overset{\mathrm{law}}{=}}
\newcommand{\p}{\mathbb{P}}

\usepackage{xcolor}
\newcommand{\E}{\mathbb{E}}    
\newcommand{\R}{\mathbb{R}}    
\newcommand{\N}{\mathbb{N}}    

\theoremstyle{plain}
\newtheorem{theorem}{Theorem}
\newtheorem{corollary}[theorem]{Corollary}
\newtheorem{lemma}[theorem]{Lemma}
\newtheorem{proposition}[theorem]{Proposition}
\theoremstyle{definition}

\newtheorem{conjecture}{Conjecture}

\theoremstyle{remark}
\newtheorem{remark}{Remark}

\usepackage{tikz}
\usetikzlibrary{arrows.meta,positioning}

\newcommand{\ee}{\varepsilon}

\newcommand{\n}[1]{\left\lVert#1\right\rVert}

\newcommand{\bxi}{{\boldsymbol{\xi}}}

\newcommand{\bP}{\mathbb{P}}

\newcommand{\bS}{\mathbb{S}}

\newcommand{\bx}{\mathbf{x}}

\newcommand{\bl}{c_2}
\newcommand{\tl}{\widetilde{\kappa}}

\newcommand{\tnu}{\widetilde{\nu}}

\def\d{\mathrm{d}}

\newcommand{\bone}{ {\mathbbm{1}} }
\renewcommand{\S}{\mathbb{S}}

\usepackage{mathrsfs}

\newcommand{\z}{\mathbf{0}}

\renewcommand{\le}{\leqslant}
\renewcommand{\geq}{\geqslant}
\renewcommand{\leq}{\leqslant}
\renewcommand{\epsilon}{\varepsilon}

\newcommand{\revision}[1]{\tb{#1}}

\newcommand{\resubfigs}{resubfigs}

 \pgfdeclarepatternformonly{south east lines}{\pgfqpoint{-0pt}{-0pt}}{\pgfqpoint{3pt}{3pt}}{\pgfqpoint{3pt}{3pt}}{
    \pgfsetlinewidth{0.4pt}
    \pgfpathmoveto{\pgfqpoint{0pt}{3pt}}
    \pgfpathlineto{\pgfqpoint{3pt}{0pt}}
    \pgfpathmoveto{\pgfqpoint{.2pt}{-.2pt}}
    \pgfpathlineto{\pgfqpoint{-.2pt}{.2pt}}
    \pgfpathmoveto{\pgfqpoint{3.2pt}{2.8pt}}
    \pgfpathlineto{\pgfqpoint{2.8pt}{3.2pt}}
    \pgfusepath{stroke}}

\begin{document}
\begin{frontmatter}
\title{Modeling Shortest Paths in Polymeric Networks Using Spatial Branching Processes}
\author[1]{Zhenyuan Zhang \fnref{label1}}
\author[2]{Shaswat Mohanty\fnref{label1}}

\author[3]{Jose Blanchet\corref{cor1}}
\author[2]{Wei Cai\corref{cor1}}
\cortext[cor1]{Corresponding author}
\address[1]{Department of Mathematics, Stanford University, CA 94305-4040, USA}
\address[2]{Department of Mechanical Engineering, Stanford University, CA 94305-4040, USA}
\address[3]{Department of Management Science and Engineering, Stanford University, CA 94305-4040, USA}
\fntext[label1]{Equal contribution}

\begin{abstract}
Recent studies have established a connection between the macroscopic mechanical response of polymeric materials and the statistics of the shortest path (SP) length between distant nodes in the polymer network. Since these statistics can be costly to compute and difficult to study theoretically, we introduce a branching random walk (BRW) model to describe the SP statistics from the coarse-grained molecular dynamics (CGMD) simulations of polymer networks. We postulate that the first passage time (FPT) of the BRW to a given termination site can be used to approximate the statistics of the SP between distant nodes in the polymer network. 
We develop a theoretical framework for studying the FPT of spatial branching processes and obtain an analytical expression for estimating the FPT distribution as a function of the cross-link density. 
We demonstrate by extensive numerical calculations that the distribution of the FPT of the BRW model agrees well with the SP distribution from the CGMD simulations. The theoretical estimate and the corresponding numerical implementations of BRW provide an efficient way of approximating the SP distribution in a polymer network.  
Our results have the physical meaning that by accounting for the realistic topology of polymer networks, extensive bond-breaking is expected to occur at a much smaller stretch than that expected from idealized models assuming periodic network structures.
Our work presents the first analysis of polymer networks as a BRW and sets the framework for developing a generalizable spatial branching model for studying the macroscopic evolution of polymeric systems. 
\end{abstract}


\end{frontmatter}


\section{Introduction} 
\label{sec:Intro}

Polymers play a vital role in technology and products of daily use due to their wide-ranging utility. 
An important class of polymers is elastomers~\citep{carrillo2005nanoindentation,erman1989rubber,beatty1987topics}, which are rubber-like materials that exhibit an elastic response even under high strains. 
The elastic and hyperelastic behavior of these materials has been modeled across various length scales from the continuum~\citep{schiel2016finite,heydari20163d} down to the microstructural scale \revision{by using coarse-grained molecular dynamics (CGMD)}~\citep{arruda1993three,kremer1990dynamics,volgin2018coarse,shi2023coarse,shen2021molecular,joshi2021review,rottach2006permanent}. However, their inelastic response due to the bond-breaking events is not well understood.
Recently, we have shown that network analysis can be applied to \revision{CGMD} models to explain the experimentally observed stress-strain hysteresis~\citep{ducrot2014toughening} in elastomers by identifying the shortest path length between distant \revision{nodes} as the governing microstructural parameter~\citep{yin2020topological}. \revision{In particular, the experimentally observed stress-strain hysteresis of tough multi-network elastomers is well captured by CGMD simulations.  Furthermore, the evolution of the shortest path length between distant nodes with strain explains not only the stress-strain hysteresis but also the anisotropic nature of the damage caused by the strain-induced bond-breaking events~\citep{yin2020topological}.}

The CGMD model comprises a large number of beads (to represent the polymer chains and cross-linking ligands), and the preparation of these systems until they reach equilibrium is a time-consuming process. Given that the macroscopic response is shown to strongly depend on the network statistics of the shortest paths, there is an incentive to replace expensive CGMD simulations with a probabilistic model that represents the polymer network. Representing polymer networks as random walks and evaluating their response under load has a long tradition in polymer theory \revision{since}~\citep{flory1960elasticity}. \revision{More recent statistical analyses of the polymer network include~\citep{svaneborg2005disorder,lang2003length}, in which the authors analyzed the properties of subchains (defined as a section of the polymer chain between neighboring cross-links) in randomly cross-linked polymer networks, and showed that the lengths of subchains follow an exponential distribution. The work of \citep{wu2012langevin} discussed the effect of the choice of the interatomic potential on the network topology, quantified by the mean squared end-to-end distance of the subchains.}
However, the focus of these studies has been on the conformation of individual polymer \revision{subchains}, \revision{which is a very local feature of the network}, and not so much on the statistics involving the \revision{global} network topology. \revision{In our analysis, the distribution of shortest paths is computed between far away nodes to quantify the global network topology}. 

In this paper, we discuss the statistical properties of the shortest path length (SPL, abbreviated as SP) distribution between far-away nodes with polymer networks modeled as branching random walks (BRWs). The SP in the polymer network corresponds to the first passage time (FPT) of the BRW, i.e., the time taken for the BRW to reach a prescribed sink or halt criterion. Roughly speaking, in a BRW we start from one particle at the origin and each existing particle independently performs a random walk until it (randomly) branches into more particles,\footnote{The BRW describes the positions of  \emph{particles} evolving in time. We do not intend this evolution to represent the time evolution of the CGMD network. Following the tradition in polymer theory, we represent typical polymer chains as random walks, but we explicitly model the cross-link topology by the branching mechanism.} 
according to a certain branching rate. Here, the branching rate in the BRW corresponds to the cross-link density in the polymeric system, and the random walk paths correspond to the polymer chains. 
We also establish theoretical predictions of the FPT of the BRW. 
%
We show that the SP distributions predicted from the CGMD simulations are consistent with the numerical implementation of the BRW model, as well as with theoretical BRW estimates.
In particular, our theoretical results on the FPT of the BRW provide explicit formulas that predict the mean SP from the CGMD network given the cross-link density. 
In addition, we show that in the long-distance limit, the fluctuation of the FPT is asymptotically much lower than its expected value, a property validated by CGMD simulations.

Analyzing the extremal behavior of spatial branching processes is an active area of recent research in probability theory.
Such processes include BRW and its continuous-time sibling, the branching Brownian motion (BBM), as well as various extensions.
The literature has focused mostly on the asymptotics of the maximum and the limit behavior near the maximum; see the works \citep{addario2009minima,aidekon2013convergence,bramson2016convergence,madaule2017convergence} for the case of one-dimensional BRW and \citep{bramson1978maximal,lalley1987conditional,roberts2013simple,aidekon2013branching,arguin2013extremal} for one-dimensional BBM, among others. 
Here we are interested in branching structures in higher (e.g.~3) dimensions (that represent real physical systems), where we mention the recent studies of \citep{mallein2015maximal,berestycki2021extremal,kim2023maximum,bezborodov2023maximal} on the location of the maximum norm. Nevertheless, we are not aware of the results regarding the FPT of spatial branching structures in general dimensions.\footnote{By continuity of the trajectories, it is not hard to show that the FPT in dimension one is precisely the inversion of the maximum.} Our work provides a necessary first step towards understanding the limit behavior of the FPT beyond dimension one. 

Our theoretical results on the FPT of spatial branching processes provide explicit formulas that predict the mean SP from the CGMD network given the cross-link density. The mean SP of the polymeric system denotes how stretched the average load-bearing polymer chain in the system is, thereby serving as an important microstructural parameter that can describe the macroscopic response of the material. 
The mean SP determines the maximum stretch that can be applied to the polymer before significant bond-breaking events occur, thus providing a measure of the stretchability before appreciable strain-induced damage.
Our theoretical estimates from the spatial branching processes show that the stretchability of polymer before bond breaking is much smaller than estimates based on idealized, periodically repeating, network topologies, such as in the 8-chain model~\citep{arruda1993three}.
Our findings demonstrate that spatial branching processes are remarkably successful in capturing the SP statistics of polymer networks and are highly promising in revealing the microstructural origin of the mechanical properties of polymeric materials. \revision{The goal of this paper is to establish the SP statistics of randomly cross-linked polymer networks at equilibrium (i.e.~before strain-induced bond-breaking occurs) by considering spatial branching processes. The evolution of SP statistics due to strain-induced bond-breaking events will be the focus of a later study.}

The rest of the paper is structured as follows. In Section~\ref{sec:num_methods}, we discuss the CGMD model that is used to obtain the equilibrated network in which the reference SP calculations are carried out. In addition, we present a high-level description of our numerical BRW model and some of its extensions. In Section~\ref{sec:results}, we numerically analyze the FPT of our spatial branching models and present the consistency with the SP distribution of the CGMD network. In Section \ref{sec: analytic}, we present analytic expressions for the FPT of the BRW models.
Section \ref{sec:conclusion} concludes the article with discussions on several future research directions.

\section{Numerical methods} \label{sec:num_methods}
 
This section provides
a high-level overview of the connections between the CGMD model of polymers and BRW processes. In particular, we compare the basic features of the CGMD and their analogs in the BRW model, which motivate extensions of the classical BRW model. We also discuss an efficient numerical algorithm for computing the FPT of the BRW and its extensions. 

\subsection{Coarse-grained molecular dynamics}
We use the bead-spring (Kremer-Grest) model~\citep{kremer1990dynamics} for a CGMD representation of the polymeric system. The simulations were carried out using LAMMPS~\citep{LAMMPS} where the initial configuration of $N_c = 500$ chains comprising $l_c = 500$ beads each (250,000 beads in total) is generated as a self-avoiding random walk. The simulation cell is a cube of length $\approx 98.2$\,nm and is subjected to periodic boundary conditions in all directions. 
The non-bonded pair interactions between beads are modeled by a Lennard-Jones (LJ) potential, $U_{\rm LJ}$, with
\begin{equation*}
    U_{\rm LJ}(r) = 4\epsilon\left[ \left(\frac{\sigma}{r}\right)^{12}-\left(\frac{\sigma}{r}\right)^6\right],
\end{equation*}
where $\sigma=15$\,\AA ~is the size of the bead, $\epsilon=2.5$ kJ/mol is depth of the energy well and $r$ is the inter-bead distance. All distances in this paper will be discussed in terms of $\sigma$. As a result, the non-dimensional bead size will be $\sigma=1$.
The neighboring interactions between beads on the polymer backbone are modeled using the finite extensible (FENE) bond potential~\citep{kremer1990dynamics}, $U_{\rm FENE}$, allowing for the incorporation of the non-linear elastic response of the polymer backbone, with
\begin{equation*}
    U_{\rm FENE}(r) = -k\left(\frac{R_0^2}{2}\right)\ln\left[1-\left(\frac{r}{R_0}\right)^2\right],
\end{equation*}
where $k=30\,\epsilon/\sigma^2$ is the bond stiffness, and $R_0=1.5\,\sigma$ is the extensibility limit.
%
The initial configuration is equilibrated using the two-step procedure~\citep{sliozberg2012fast} to obtain the equilibrated \textit{baseline configuration} (i.e.~a polymer melt), which is then modified to incorporate cross-links. 

Once the {baseline configuration} is well equilibrated, the candidate bonding beads are identified by choosing distinct bead pairs that lie within a cutoff distance of $r_c<1.15\,\sigma$. We then randomly choose from the set of candidate bead pairs and assign an irreversible cross-link between the chosen pairs.
To model the irreversible cross-links, we use a quartic bond potential $U_{\rm Q}$~\citep{ge2013molecular}, defined by
\begin{equation*}
    U_{\rm Q}(r) = 
    \begin{cases}
    K(r-R_c)^2(r-R_c-B_1)(r-R_c-B_2)+U_0 &\text{  for }r\leq R_c,\\
    0 &\text{  for }r>R_c.
    \end{cases}
\end{equation*}
Here $K=1200\,\epsilon/\sigma^4$ is the bond stiffness, $B_1=-0.55\,\sigma$, $B_2=0.95\,\sigma$, $U_0=34.6878\,\epsilon$, and $R_c=1.3\,\sigma$~is the cutoff distance beyond which the quartic bond is considered broken (and cannot be formed again).  This procedure is identical to the preparation of a single-network (SN) elastomer model in~\citep{yin2020topological}.
\revision{Having $R_c$ greater than $\sigma$ allows for the enthalpic stretching of the cross-link. However, we found that in our CGMD simulations, the enthalpic stretching of shortest paths before breaking is not significant (less than 10\%).}
 
\subsection{Network analysis}
We describe the polymer network from the CGMD simulation cell by \revision{focusing on} the cross-linking beads in the system. Cross-linking beads that are connected to each other along the backbone of the polymer chain are denoted by a graph edge with a weight equal to the number of bonds between them. Furthermore, we also define an edge between the pair of cross-linking nodes that are involved in forming a cross-link, with an edge weight of 1. Further details regarding the network representation and SP calculation can be found in~\citep{yin2020topological}.
For each node $i$, we find a destination node $j$ that is closest to the point offset from node $i$ by $q_x$ in the $x$-direction and compute the shortest path length (SP) connecting nodes $i$ and $j$. \revision{To ensure that SP measures a non-local feature of the network, $q_x$ is chosen to be large enough such the SP is computed between distant cross-link beads.} The non-local microstructural measure of the SP has the advantage of being independent of chain conformational fluctuations and they only depend on the network connectivity, as a result evolving only as a function of the evolution of the polymer connectivity network. 
We compute the SPs starting from all nodes in the polymer network using  Dijkstra's algorithm~\citep{dijkstra2022note}, giving rise to a distribution. We are interested in how this distribution depends on the choice of distance $q_x$, where $0 < q_x \le L_x$ and $L_x$ is the periodic length of the simulation cell in the $x$-direction.
Note that when $q_x = L_x$, node $j$ is the periodic image of node $i$ and the simulation box is replicated in the $x$-direction for the SP calculation, as shown in Figure~\ref{subfig:qx_schematic}.
\begin{figure}[ht!]
    \centering  
    \subfigure[]{\includegraphics[trim={0cm 7.5cm 0cm 7.5cm},clip,width=0.56\textwidth]{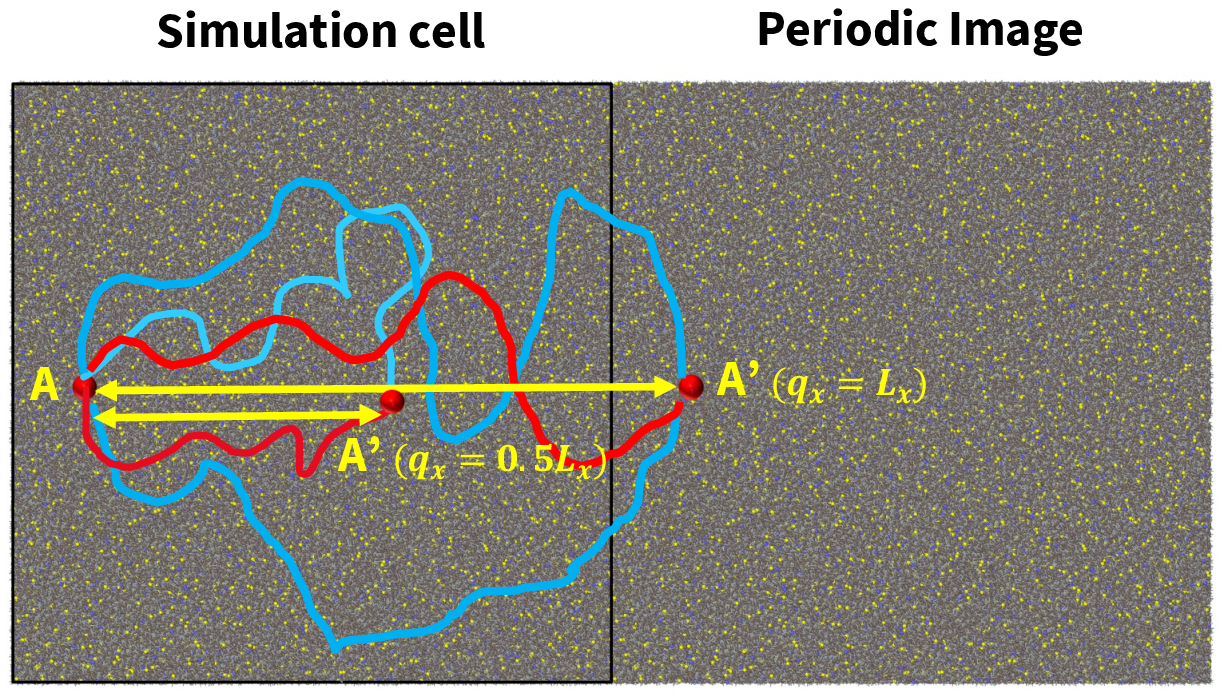}\label{subfig:qx_schematic}}
    \subfigure[]{\includegraphics[width=0.41\textwidth]{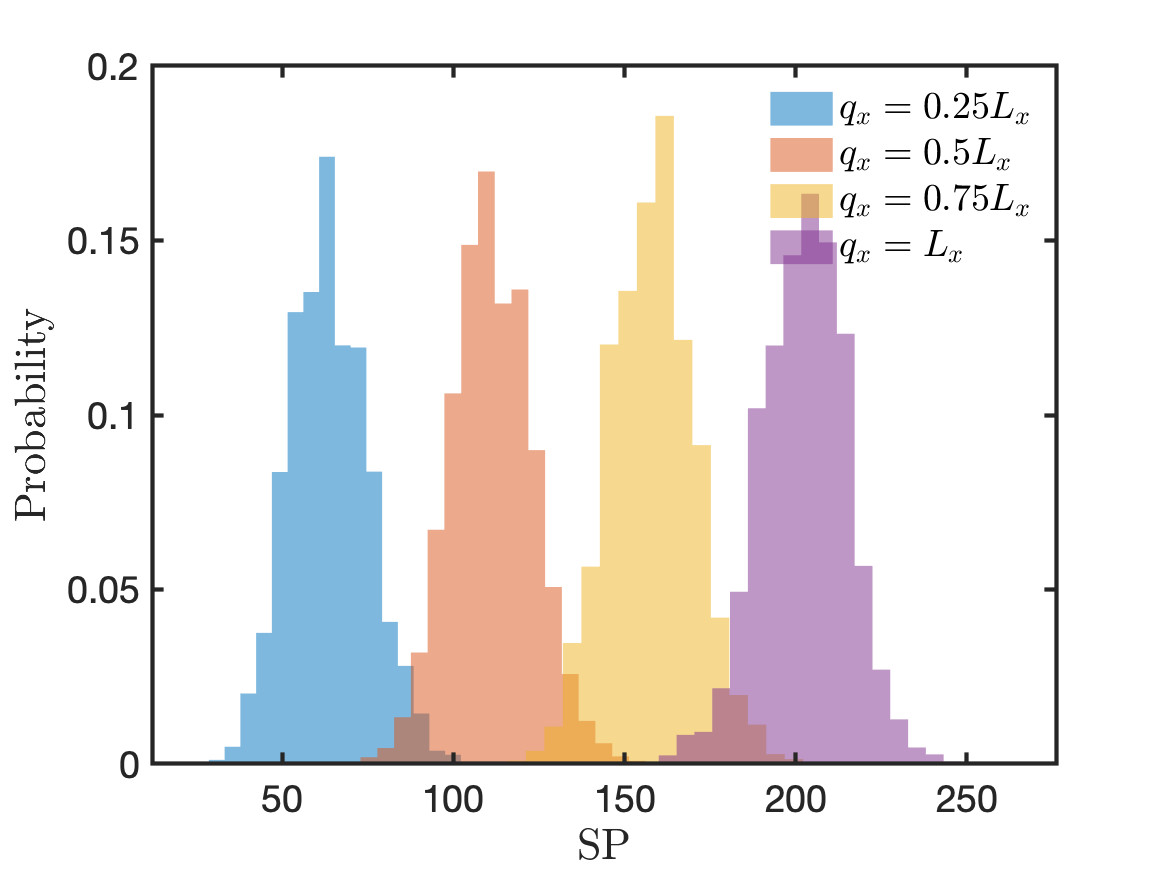}\label{subfig:cgmd_hist}}
    \caption{ (a) The SP between $A$ and the destination node $A'$ (shown in red) separated by different distances $q_x$ in the CGMD simulation cell (the longer paths are shown in blue).\protect\footnotemark ~(b) The SP distributions for all nodes at different values of $q_x$, where $L_x = 65.5\,\sigma$.}
    \label{fig:cgmd_ref}
\end{figure}
\footnotetext{Note that if the offset distance $q_x$ is equal to the simulation cell size $L_x$, we are guaranteed to find the destination node at the exact distance since it is the periodic image of the source node. However, for any other value of $q_x$, we search for the nearest node as the destination node.}

As an example, Figure~\ref{subfig:cgmd_hist} plots the SP histogram from the CGMD simulation cell at offset distances of $q_x = 0.25L_x,\,0.5L_x,\,0.75L_x,\,L_x$, respectively.
The SP distribution appears to have a mean that increases linearly with $q_x$, and a standard deviation that is insensitive to $q_x$.
There does not exist any theory or analysis in the literature that explains this behavior. 
Furthermore, brute-force computation of the SP distribution is time-consuming, requiring lengthy equilibration of the CGMD network followed by SP calculation for all network nodes.
These limitations provide the motivation to obtain theoretical estimates of the SP statistics based on analytically tractable models.

\subsection{Branching random walk}
\label{sec:brw}

In the classical setting, a (discrete-time) branching random walk (BRW) in $\R^d$ describes a stochastic process indexed by $n=0,1,2,\dots$, where starting from a single particle at $\z\in\R^d$ at time $n=0$, each particle performs a standard random walk (e.g., with independent increments that are uniformly distributed on the sphere $\S^{d-1}=\{\bx\in\R^d\mid\n{\bx}=1\}$),\footnote{Throughout this paper, we use the Euclidean ($\ell_2$) norm.} but randomly reproduces particles at each time step,  independently from each other and from their common ancestor(s). The trajectories of the particles form a (possibly infinite) tree in $\R^d$. We will be mostly interested in the physical space of dimension $d=3$.

At first sight, there are a few reasons that motivate why the BRW may serve as a good proxy for the SP statistics of polymer networks:  
\begin{itemize}
    \item It is well known that the equilibrium configuration of single chains in a polymer melt can be well described by a random walk~\citep{doi1988theory,flory1960elasticity,graessley2003polymeric}.  In our CGMD model, the cross-links are added into well-equilibrated polymer melt, and hence we expect the path obtained by traversing along the backbone of the polymer network to be also well described by a random walk.
    While a cross-link in the polymer network joins two polymer chains together, this can be modeled by extending the random walk to allow a branching process where one walker becomes three walkers.
    
    \item The BRW has a mathematically rich and relatively tractable structure. The extremal behavior of the BRW is a well-studied subject in probability theory. For example, in dimension $d=1$, a precise asymptotic of the largest displacement of the BRW particles is well known~\citep{addario2009minima}, and the limit behavior near the frontier has been fully characterized~\citep{aidekon2013branching,madaule2017convergence}. 
    \item The BRW allows enough freedom in the choice of parameters, such as the random walk step-size distribution and the branching rate. These parameters may be chosen appropriately so that the statistics of the FPT in the BRW model agree well with that of the SP in the CGMD model of the polymer network.

\end{itemize}

The BRW provides an alternative point of view to the polymer network. \revision{Although the entire collection of paths (forming a tree) generated by BRW does not represent the polymer network as a whole, we will show that the FPT of the BRW serves as an accurate representation of the SP between distant nodes in the polymer network.} As shown in Figure~\ref{fig:schematic2}(a),  the traversal along the backbone of the polymer (starting from node $A$) is the equivalent of a random walk path and the point of cross-linking can be thought of as a branching event, where the path splits into \emph{three} paths: (i) continuation of the polymer backbone ($A\rightarrow A''$), and (ii) two parts of another polymer chain ($B\rightarrow B''$). 
\begin{figure}[ht!]
    \centering
    \includegraphics[trim={0.7cm 10.8cm 0.7cm 10.8cm},clip,width=0.7\textwidth]{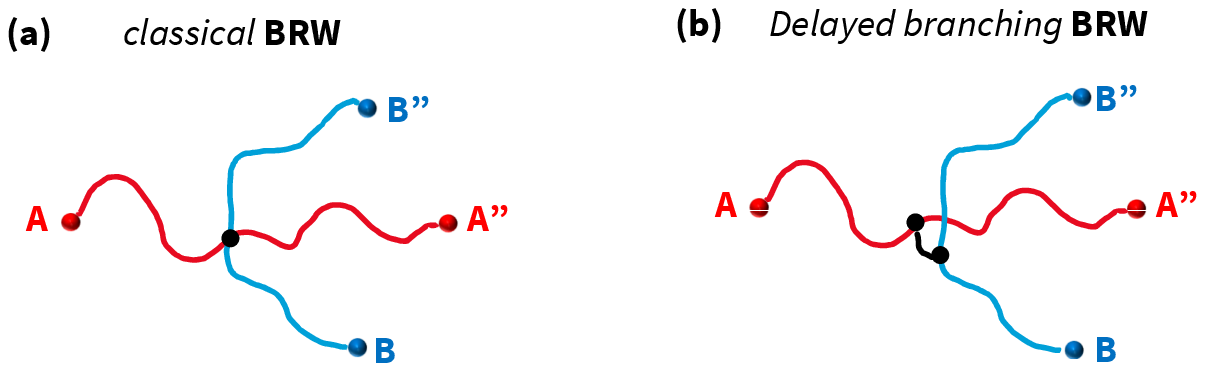}
    \caption{BRW tree representations in the (a) classical and the (b) delayed branching regimes.}
    \label{fig:schematic2}
\end{figure}
We would like to point out that what the branching random walk model aims to mimic is the part of the polymer network that emanates from a single node, rather than the entire network.
Cross-links in polymers typically bind monomers (on chains) that are physically close to each other, 
and the occurrence of the cross-linking events in space can often be described by a Poisson process.
In the BRW model, we incorporate the cross-linking effect by allowing an effective branching rate of $\tl$. 
However, in certain polymer systems, cross-links can be controlled to occur only at specific sites, e.g.~those evenly distributed along the polymer backbone~\citep{cooper2020multivalent}. This could be modeled using deterministic branching events, that occur once every $1/\tl$ steps. In this paper, we restrict our discussion to the polymer systems where the cross-linking follows a Poisson process, where $\tl$ is obtained from the CGMD simulation cell (see below). 

\paragraph{BRW algorithm}
The BRW is initialized at the origin as a single path with the branching probability $\tl$ at every time step into three paths.
Every subsequent branched child has the same branching likelihood at every time step.
The jumps are of a fixed constant length $\sigma$, corresponding to the bead size in the CGMD simulation cell.
Similar to the SP calculations in the CGMD model, we terminate the BRW (and record FPT) as soon as any of the random walkers hits within a sphere of radius $R_c = \sigma$ centered at $(q_x,0,0)$. 
Given that each step corresponds to a unit of time, we can calculate the FPT (equivalent to the SP) for a BRW. Here, the radius for the termination region specifies the maximum deviation from the destination at an offset distance of $q_x$ from the origin, similar to the SP analysis in the CGMD model.\footnote{The quality of our models are not sensitive to the choice of radius $R_c$; any value within the same order of magnitude would give qualitatively the same result.}

\paragraph{Branching rate estimation}
In the case of the polymer network modeled in the CGMD simulation, 
the cross-link formation between proximal inter-polymer sites mimics a Poisson point process. We can track all the inter-cross-link distances along the polymer backbone, denoted by the variable $x_b$. The survival function ($1-$CDF$(x_b)$) of the histogram of $x_b$ approximates the exponential decay $x\mapsto \exp(-\kappa_a x)$. Here, the exponent $\kappa_a$ is the arrival rate for a Poisson process. In our estimation of the branching rate $\tl = \kappa_a$, we will obtain the $\kappa_a$ from the cotangent of the logarithm of the survival function of the $x_b$ distribution, given by $1 - $CDF($x_b$). 

\paragraph{Numerical implementation} 
Tracking all the branched paths (children) of the BRW tree is memory intensive. To circumvent this difficulty, in the numerical representation of the BRW model, we purge the tracks that are farthest away from the termination point.  In our implementation, we fix a large number $p_c=9000$ and include a \textit{path purging} step, where once the number $n_p$ of tracked paths exceeds $p_c$ we trim down the number of paths to $p_c/3$. 
Our numerical results (in \ref{sec:validation}) support the expectation that path purging does not noticeably affect FPT calculations. \revision{Intuitively, this makes sense because a fast particle is likely to maintain its speed at all times}. 


\subsection{Extensions of the classical branching random walk model}\label{sec:newbrw}

A closer examination of the classical BRW model reveals the following features of the CGMD model of the polymer network that have not yet been taken into account.
\begin{itemize}
    \item The length of the cross-link is neglected. For the classical BRW model, the cross-links are assumed to have zero length and are not taken into account when computing the SP statistics. On the other hand, each cross-link in the CGMD model has a non-zero length and is counted towards the SP.
    \item The polymer chains (before cross-linking) have a finite length of $l_c = 500$ beads each in the CGMD model, and hence we would expect a termination rate for the corresponding BRW model for consistency.
    \item There is a nontrivial correlation between neighboring link vectors on a polymer chain in the CGMD network. The classical BRW model fails to incorporate such correlation effects since its increments are independent.
\end{itemize}

For these reasons, we introduce the following important extensions of the classical BRW model that more accurately represent the network in our CGMD configurations. 

\paragraph{Delayed branching regime} We introduce a two-step branching regime. 
Instead of each random walker immediately branching into three walkers, each branching event now consists of two sub-events, as shown in Figure \ref{fig:schematic2}(b).
First, the walker branches into two walkers; the trajectory of the first descendant corresponds to the original backbone towards $A''$ and the next step of the second descendant corresponds to the cross-link.
Then, after one step, the second descendant branches into two walkers again; their trajectories correspond to the backbone of the chain $B\rightarrow B''$.
The delayed branching regime is analytically tractable when we consider the FPT, as explained in Section \ref{sec:analytic:BRW}.

\paragraph{Termination regime} We introduce a termination rate that controls the length of the paths. Observe that in our CGMD model, when tracing through the cross-link towards a new polymer chain, the total lengths of the two directions along $B\rightarrow B''$ sum up to $l_c = 500$. 
Therefore, to be fully consistent with our CGMD model, one expects that the lengths of both branches are uniformly distributed among possible choices (i.e., pairs of non-negative integers that sum up to $l_c$), which is technically difficult for both theoretical analysis and numerical implementation of our BRW model. Instead, we propose a (random) termination criterion 
that leads to similar chain lengths on average and is both analytically tractable and memory-efficient. Since each branching event corresponds to cross-linking to a new polymer chain of length $l_c$ represented by two new paths, each path on average takes a length of $l_c/2$. Therefore, the (random) termination rate can be approximated as $\tnu=2/l_c$.
At each time step, each existent path terminates independently with probability $\widetilde{\nu}$, 
in the sense that it stops performing random walks and no new path will be produced from it.
Here, the termination event and branching event are exclusive, that is, we must have $\tl+\tnu\leq 1$.
Since the length of a polymer chain of $l_c = 500$ is much larger than the possible offset distance $q_x$ for the shortest path analysis considered here, this termination effect will be almost negligible for large branching rates.

We name the resulting integrated BRW model (with the delayed branching regime and termination) the $(\tl,\tnu)$-BRW, where $\tl$ refers to the branching rate and $\tnu$ is the termination rate. 
Our theoretical analysis of the $(\tl,\tnu)$-BRW model is applicable to any termination rate $\tnu$.
We remark that with a positive termination rate, there is a nontrivial extinction probability for the spatial branching process. For this reason, in our analysis, we will always condition on the \emph{survival event}, that there exist alive particles at all times.

\paragraph{Correlated jumps and mean-squared internal distance} 
So far, the BRW model we have discussed incorporates the jumps (i.e., link vectors between nodes) as independent and identically distributed (i.i.d.) random variables. In practice, one may expect correlations between the jumps due to the volume exclusion of the nodes in a CGMD network. To this end, we introduce the \emph{branching correlated random walk} (BCRW) as a generalization of the BRW, 
where the law of a jump vector depends on the jump vector at the previous step.
In the case of a branching event, for simplicity, we assume that the first step of the children depends on the jump vector of the parent walker at the previous step.\footnote{This is a simplification of the cross-linking event where the orientation of the cross-link and subsequent polymer chain (children) are not expected to be highly correlated with the original polymer chain (parent). Such an approximation affects very little the results.}

To define precisely the BCRW model we need to introduce the concept of \emph{mean-squared internal distance} (MSID). To this end, we summarize some of the key results that are well-known in polymer physics~\citep{rubinstein2003polymer}. Let us consider a polymer chain of $N$ links connecting a series of nodes whose positions are specified by the vectors $\mathbf{R}_i,~i=0,\dots,N$.\footnote{Typically, in this paper, a bold symbol refers to a vector.} An idealized model using classical BRW to represent the CGMD  indicates that the link vectors $\mathbf{r}_i=\mathbf{R}_i-\mathbf{R}_{i-1}$ are i.i.d., resulting in $\E[\mathbf{r}_i \cdot \mathbf{r}_{i+1}]=0$. Recalling that each jump length $\n{\mathbf{r}_i}=\sigma$, we have in this idealized case that $\E[\n{\mathbf{R}_N-\mathbf{R}_0}^2]=N\sigma^2$. However, we can introduce a measure of correlation, $\alpha$, between consecutive links such that the average angle $\theta$ between consecutive links is given by $\alpha = \cos \theta$. This results in $\E[\mathbf{r}_i \cdot \mathbf{r}_{i+1}]=\alpha\,\sigma^2$. In general, it can be shown that $\E[\mathbf{r}_i \cdot \mathbf{r}_{j}]=\alpha^{\vert i-j \vert}\,\sigma^2$. As a result, in the case of correlated jumps, we have $\E[\n{\mathbf{R}_N-\mathbf{R}_0}^2]=\sum_{j=-\revision{(N-1)}}^{N-1} \alpha^{|j|}\,\sigma^2 ,$ which in the limit $N\rightarrow\infty$ simplifies to 
$\E[\n{\mathbf{R}_N-\mathbf{R}_0}^2]
= C_{\infty}\,N,$
where $C_\infty = (1+\alpha)/(1-\alpha)$ is called the \emph{characteristic ratio}.
The mean-squared internal distance can now be defined as
\begin{equation}
{\rm MSID}(n) = \frac{\E[\n{\mathbf{R}_n-\mathbf{R}_0}^2]}{n}\,.
\end{equation}
We have $\lim_{n\to\infty} {\rm MSID}(n) = C_{\infty}$.
Figure~\ref{fig:msid} (red line) shows how the MSID of a correlated random walk approaches this limit with increasing $n$.

\begin{figure}
    \centering
    \includegraphics[width=0.41\textwidth]{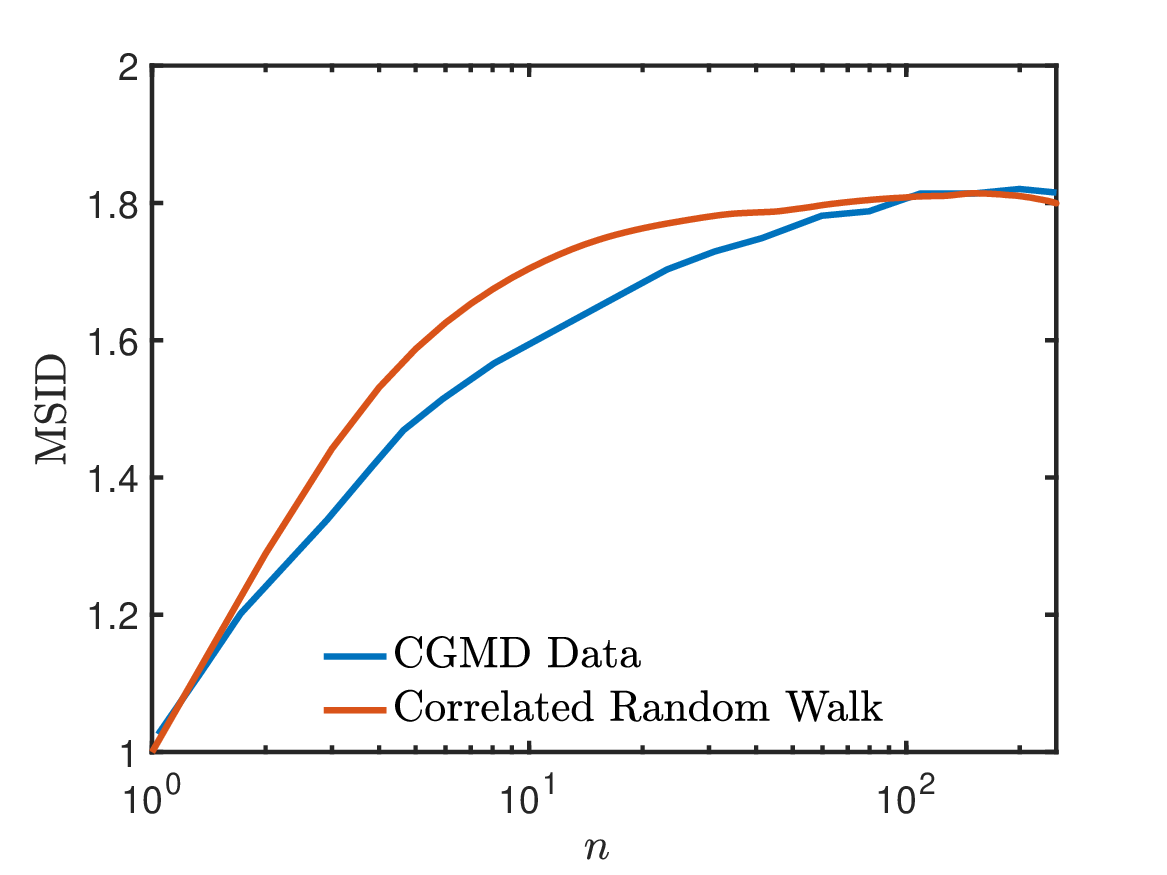}
    \caption{The MSID (in units of $\sigma^2$) computed from the CGMD model compared against the MSID of an ensemble of single chains generated using a 
    correlated random walk. }
    \label{fig:msid}
\end{figure}

We now specify the BCRW model adopted in this work. Given $\{\mathbf{r}_k\}_{1\leq k\leq i-1}$,  the link vector $\mathbf{r}_i$ is 
specified by the following expressions
\begin{align}
    \mathbf{j}_i&=\beta \, \mathbf{r}_{i-1} + \sqrt{1-\beta^2} \, \boldsymbol{\delta}_i, \ 
    \mathbf{r}_i= \frac{\mathbf{j}_i}{\n{ \mathbf{j}_i}},\label{eq:markov}
\end{align}
where $\beta\in(0,1)$ and  $\{\boldsymbol{\delta}_i\}_{1\leq i\leq N}$ are i.i.d.~samples from the unit sphere $\S^2$.
In particular, the sequence of jump vectors $\mathbf{r}_i$ 
is Markovian and
as a consequence, there is a computable constant $\alpha=\alpha(\beta)$ such that~\citep{auhl2003equilibration}
$$\lim_{n\to \infty} {\rm MSID}(n) = \frac{\sigma^2(1+\alpha)}{1-\alpha}.$$
The MSID statistics in the CGMD model (see Figure~\ref{fig:msid} blue line) can be used to determine the constant $\beta$ (by inverting $\alpha(\beta)$) in the correlated random walk model. 
The MSID obtained from the CGMD network asymptotes to $C_\infty\approx 1.83 \, \sigma^2$.
We pick $\beta\approx 0.4$ in \eqref{eq:markov} so that the MSID obtained from the numerical random walk approached $C_\infty$ in the large $n$ limit (see Figure \ref{fig:msid}). 

\revision{We note that even after accounting for the correlation between consecutive jumps, the description of a single polymer chain as a correlated random walk is still an idealization that ignores the volume-exclusion interactions between non-consecutive monomers on the same chain.  Such volume-exclusion effects can lead to swelling of the polymer chain in a dilute polymer solution (e.g.~a gel), resulting in a deviation of the chain conformation from (correlated) random walk statistics. However, in a polymer melt (which is in the concentrated solution regime), the self-repulsion effect of the polymer chains is canceled by the mutual competition of many polymer molecules for the same space.  As a result, the chain statistics in a polymer melt is well described by the (correlated) random walk model~\citep{flory1960elasticity}.}
\revision{In this work, the cross-links are randomly distributed in space, so that the cross-linking density has heterogeneity at the microscopic scale.  However, the effects considered in~\citep{bastide1990enhancement,bastide1990scattering} related to local swelling caused by the microscopic heterogeneity are unlikely to arise here because we have a polymer melt instead of a gel. We do not consider large-scale variations of cross-linking density that may result from the detailed dynamics of the cross-linking process~\citep{vilgis1992rubber}.}

Another approach besides introducing correlation between jumps is to modify the jump distance $\sigma$ in the BRW model so that its MSID (which is constant in $n$) matches that of the CGMD model at a certain value of $n$.  A reasonable choice is $n = 1 / \tl$, so that the average distance between consecutive branching points along a path in the BRW model matches the average distance between neighboring cross-linking nodes on the same chain in the CGMD model.
We use the \emph{scaled $(\tl,\tnu)$-BRW} to represent the BRW model whose jump distance is scaled to match the MSID at $n = 1/\tl$.
\revision{It has been well recognized that, with proper rescaling, the statistics of polymer chains at large length scales are rather insensitive to the details of the inter-molecular interactions~\citep{deGennes1979scaling}.  As a result, the random walk has been a very useful model for understanding the behavior of polymer networks.}

\paragraph{Gaussian jumps} 
The BCRW model with correlated jumps is not analytically convenient to deal with. However, when the effective branching rate is low, we expect from the central limit theorem (for sums of i.i.d.~increments and for sums of Markov increments) that our BRW and BCRW models both can be well approximated by a branching random walk model with independent Gaussian increments.\footnote{Meanwhile, it is well known that normal approximations sometimes behave unsatisfactorily in the large deviation regime, which applies to our case since we will be interested in the extremal behavior of the spatial branching models. It is therefore risky to apply the aforementioned heuristic of the central limit theorem. Fortunately, when the branching rate is sufficiently small (which is the case for the polymeric systems in which we are interested) and the offset distance $q_x$ is not too large, we are in a relatively moderate deviation regime, and the normal approximation shows quite good precision.} For this reason, we introduce the \emph{Gaussian branching random walk} (GBRW) model, where by definition the jumps are i.i.d.~and Gaussian distributed.\footnote{The increments of a GBRW model are independent, meaning that it fails to capture the correlation between the jumps on distinct cross-linked chains when compared to the BCRW model. Nevertheless, this does not hamper the FPT statistics significantly; see Figure \ref{fig:bbm_c1_rho_cgmd_b} below. } We incorporate the termination and delayed branching schemes, as well as proper scaling with MSID, into the GBRW model in order to stay close to the BCRW model defined above.

Our goal in introducing the GBRW model is 
to show that little precision is lost even if we use a simpler (and more universal) model to describe the SP statistics of a polymer network. 
Here, the simplicity refers to the fact that the increments are independent, and there exists a precise formula  (up to $O_{\bP}(\log\log x)$) for the FPT just as scaled BRW.
The universality comes from the central limit theorem.
We wish to convey the crucial idea that the correlated random walk defined in \eqref{eq:markov} should be considered as a prototype, and other similarly defined correlated random walks (e.g., using the angular fan method) that exhibit a central limit behavior (with first and second moments matching MSID calculations)  should provide equally good descriptions of SP.

As a further simplification to the GBRW, we will present in Section \ref{sec:analytic BBM}  the theoretical treatment of the branching Brownian motion (BBM) where we prove a precise formula (up to $O_{\bP}(1)$) for the FPT.

In the majority of our paper, we will focus on the numerical implementation of the scaled $(\tl,\tnu)$-BRW,  $(\tl,\tnu)$-BCRW, and scaled $(\tl,\tnu)$-GBRW 
 models. When there is no confusion, we simply write BRW, BCRW, and GBRW respectively. The numerical implementation of the BBM and the unscaled BRW and GBRW models are validated in \ref{sec:validation}, though we do not present their performance against the CGMD results owing to their relatively poorer performance in comparison to the BRW, BCRW, and GBRW models.
A summary of the four models is given in Table \ref{table}.

\begin{table}[ht]
\centering
\begin{tabular}{|c|c|c|c|c|}
\hline
\textbf{Model}                  & \textbf{BCRW}                     & \textbf{BRW}                         & \textbf{GBRW}        &     \textbf{BBM}    \\ \hline
Time  variable           & discrete                & discrete                 & discrete    & continuous    \\ \hline
Jumps             & dependent                & independent                 & independent    & independent    \\ \hline
Jump distribution & uniform on $\mathbb S^2$ & uniform on $\mathbb S^2$    & Gaussian     & Gaussian      \\ \hline
Jumps $\propto \sqrt{ \rm MSID}$                 &   N/A                       & $\checkmark$                & $\checkmark$   & $\checkmark$    \\ \hline
Termination            & $\checkmark$             & $\checkmark$                &  $\checkmark$   & N/A  \\ \hline
Delayed branching      & $\checkmark$             & $\checkmark$                &  $\checkmark$   & N/A  \\ \hline
FPT predictions        & not available            &  $\pm O(\log\log x)$ & $\pm O(\log\log x)$ &  $\pm O(1)$ \\ \hline
\end{tabular}
\caption{A summary of the characteristics of the four spatial branching models considered in this work}
\label{table}
\end{table}

\section{Results}
\label{sec:results}

The key results of the SP statistics obtained from various types of spatial branching models are discussed in this section. 
We start from the BCRW model which is specifically constructed to carry all three features of the CGMD network: termination, delayed branching, and correlation, following  Section \ref{sec:newbrw}.\footnote{We slightly abuse notation here by including the delayed branching property in BRW, where we also assume that jumps are uniform on $\S^2$; instead, the traditional BRW (i.i.d.~offsprings) will be referred to the \emph{classical} BRW.} We then demonstrate the scaled BRW and GBRW models,\footnote{By \emph{scaled} we mean that the jumps (positions of subsequent particles in the tree) are scaled by a common factor.} which are more tractable and have an equally good description of the SP statistics.
 In particular, for the BRW and GBRW models, we will be able to develop asymptotic formulas for the FPT, which leads to useful theoretical predictions in Section \ref{sec:analytic:BRW}.

\subsection{SP statistics predicted by BCRW}
\label{sec:cgmd as bmrw}

\paragraph{Numerical results}
Using the method described in Section~\ref{sec:brw}, the effective branching rate corresponding to the CGMD model containing 250,000 beads with 6,000 cross-links is $\tl\approx 0.0398$ and that with 13,600 cross-links is $\tl\approx 0.0856$.
Figure~\ref{fig:corr_cgmd_dist_lower} shows the SP distribution (for $q_x \approx 65.5 \, \sigma$) from the CGMD model and that predicted by the BCRW model at these two cross-linking densities.
Good agreement is observed between the CGMD results and BCRW predictions, both in terms of the mean and the standard deviation.
\begin{figure}[h!]
    \centering
        \subfigure[]{\includegraphics[width=0.38\textwidth]{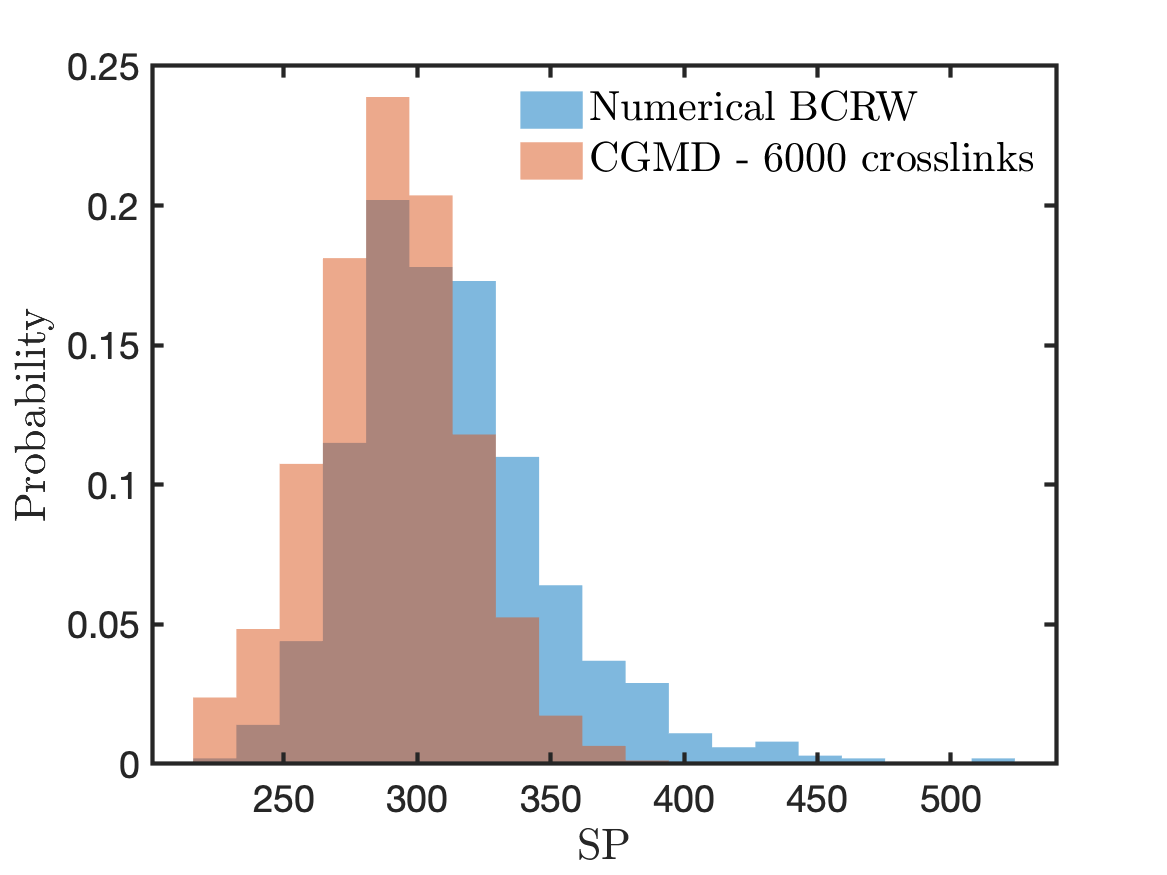}\label{fig:corr_cgmd_dist_lower_c}}
         \subfigure[]{\includegraphics[width=0.38\textwidth]{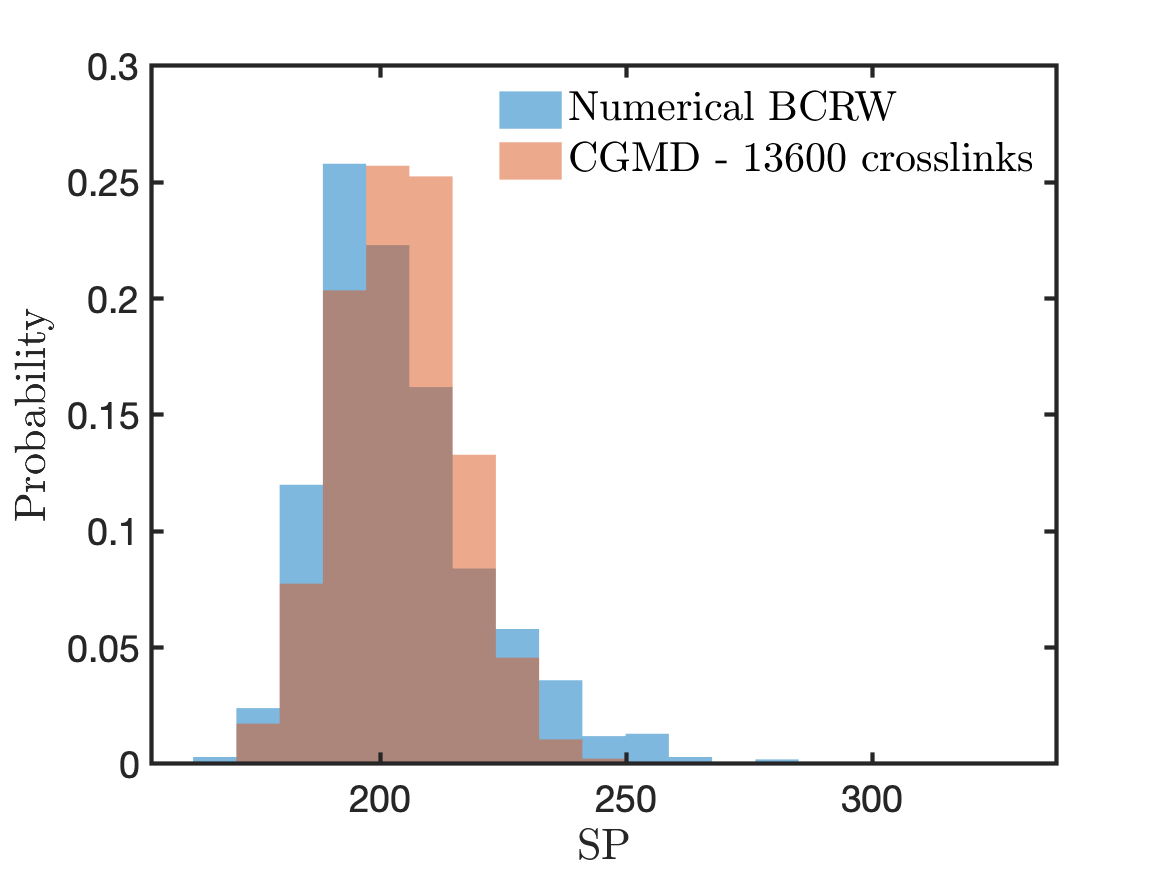}\label{fig:corr_cgmd_dist_ref_c}}
    \caption{The SP distributions from the BCRW compared against a single CGMD configuration at different cross-link densities: (a) 6000 cross-links ($\tl \approx 0.0398$), and (b) 13600 cross-links ($\tl \approx 0.0856$) for $q_x = 65.5\,\sigma$ ($\approx 98.2$ nm). 
    }
    \label{fig:corr_cgmd_dist_lower}
\end{figure}
Figure~\ref{fig:corr_sig_valid_a} shows how the SP distribution predicted by the BCRW model depends on the offset distance $q_x$.
The distribution shifts to the right with its shape nearly unchanged with increasing $q_x$, in good agreement with the CGMD result shown in Figure~\ref{subfig:cgmd_hist}.
Figure~\ref{fig:cgmd linear} shows that the mean SP grows linearly with $q_x$, with excellent agreement between BCRW and CGMD models.
Figure~\ref{fig:corr_sig_valid_b} shows that the standard deviation of SP remains almost unchanged in the range of $q_x$ considered here, with the BCRW prediction a little higher than the CGMD results.
In summary, the BCRW model successfully captures the linear dependence of the SP mean with $q_x$ and the CGMD observation that the SP standard deviation stays near a constant value much smaller than the mean.\footnote{That the FPT of BCRW is concentrated can be supported intuitively by the following famous quote from \citep{talagrand1996new}: \emph{A random variable that depends (in a ``smooth" way) on the influence of many independent variables (but not too much on any of them) is essentially constant.}}
We note that the correlation between successive jumps is important for reaching the level of agreement between BCRW and CGMD predictions shown here.
If the correlation were simply removed (not shown), then the SP mean would become larger (the FPT increases since uncorrelated jumps make the trajectory more tortuous), while the standard deviation would remain largely unchanged and independent of $q_x$.

\begin{figure}[ht!]
    \centering
    \hspace{-1cm}
        \subfigure[]{\includegraphics[width=0.35\textwidth]{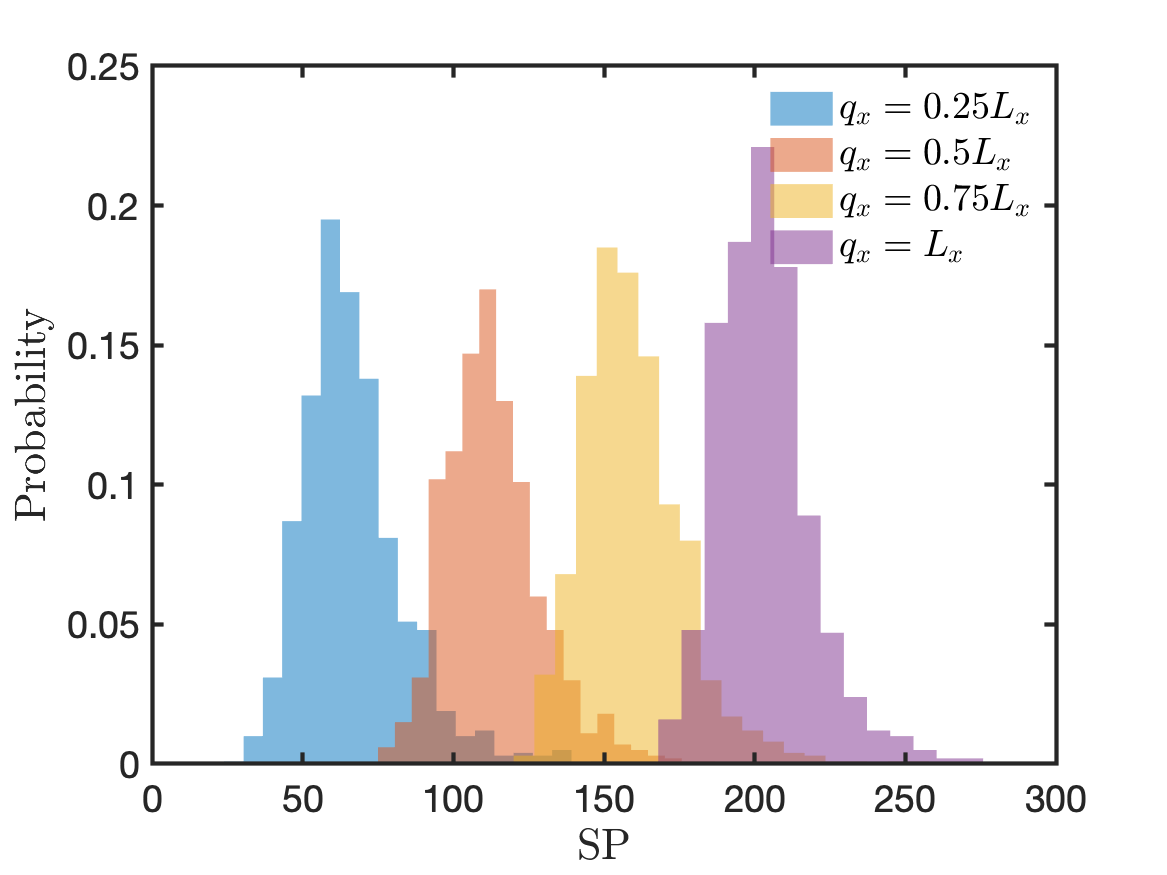}\label{fig:corr_sig_valid_a}}
        \hspace{-0.6cm}
        \subfigure[]{\includegraphics[width=0.35\textwidth]{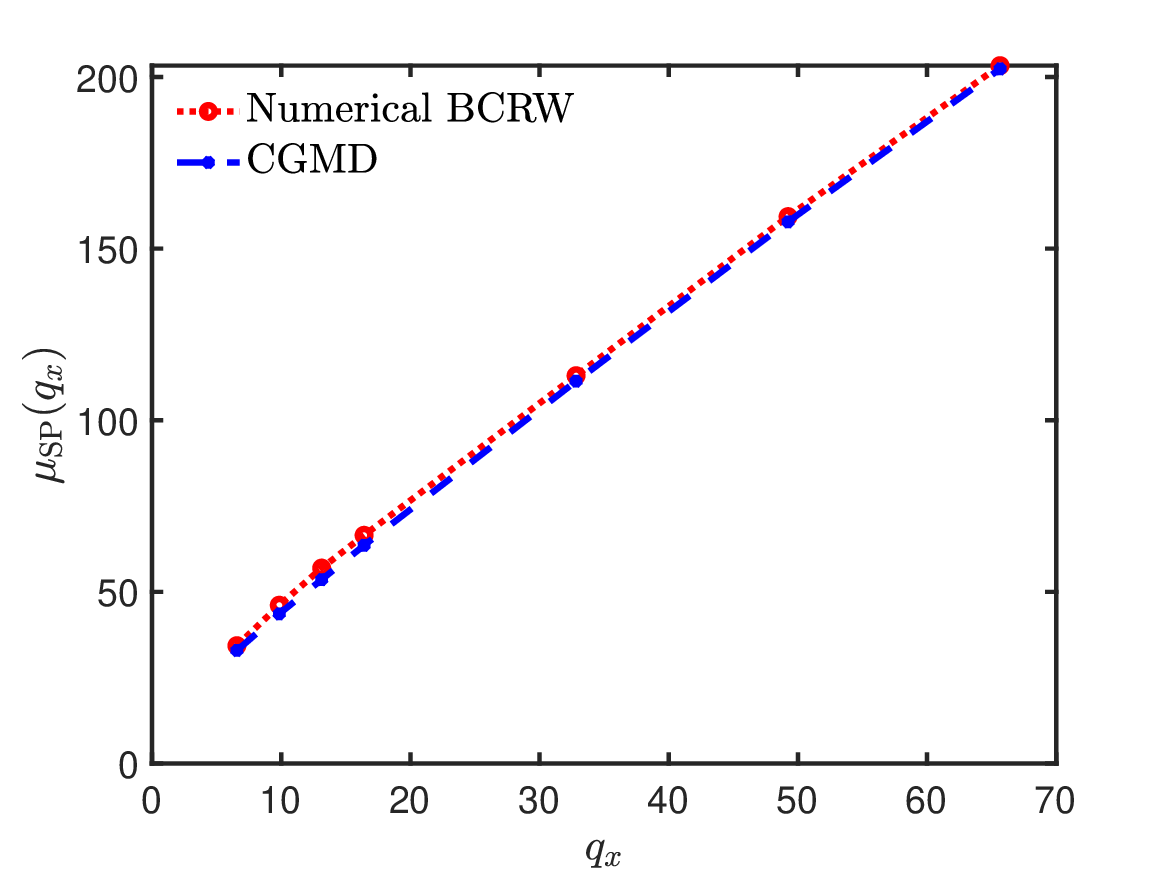}\label{fig:cgmd linear}}
         \hspace{-0.6cm}
        \subfigure[]{\includegraphics[width=0.35\textwidth]{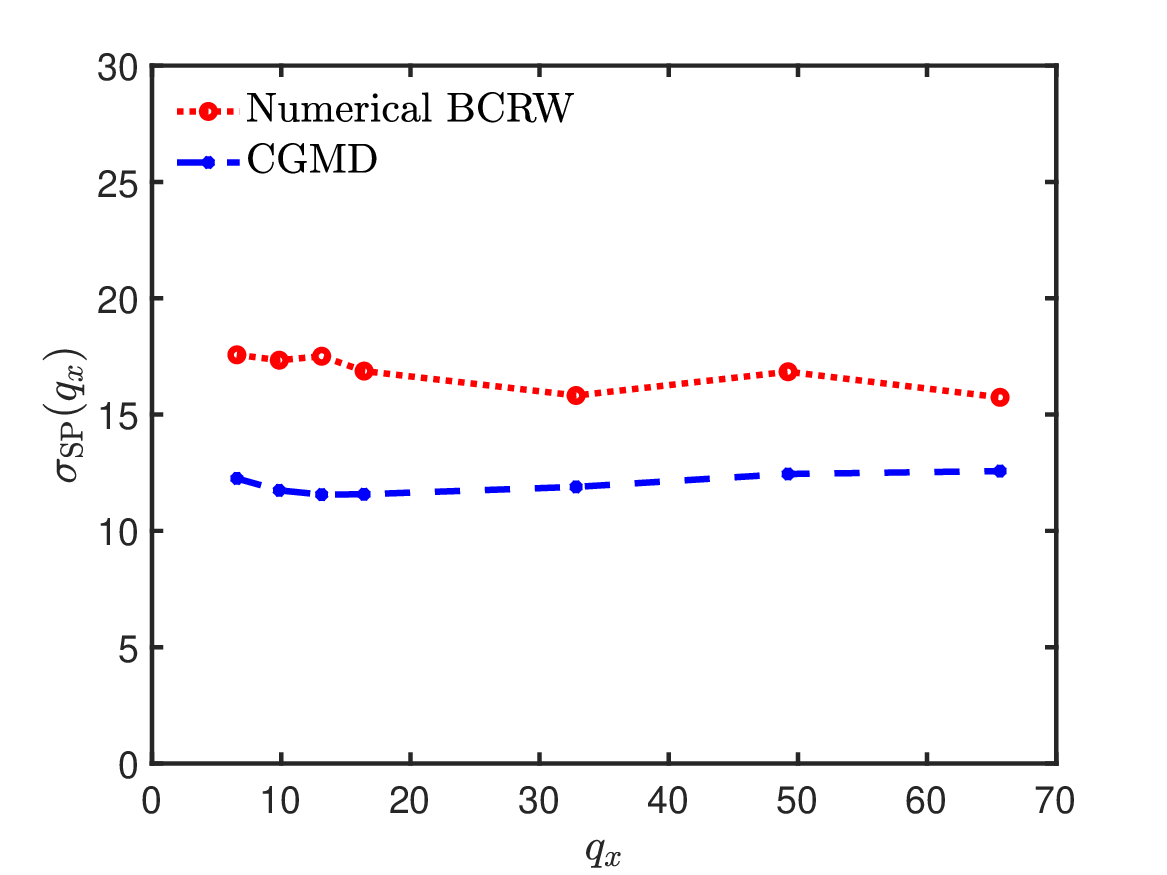}\label{fig:corr_sig_valid_b}}
       \hspace{-1cm}
    \caption{(a) The SP distributions at different values of $q_x$ from the  numerical BCRW. The (b) mean and (c) standard deviation of the SP as a function of $q_x$. The data points correspond to $q_x = 0.1L_x,\,0.15L_x,\,0.2L_x,\,0.25L_x,\,0.5L_x,\,0.75L_x,\,L_x$. All network analysis corresponds to a single CGMD configuration with a cross-link density that corresponds to $\tl=0.0856$ in the BCRW model.
    }
    \label{fig:corr_sig_valid}
\end{figure}

\paragraph{Metric for linear dependence of SP on $q_x$}
Figure~\ref{fig:cgmd linear} suggests that in the range of $q_x$ considered here, the SP mean can be well approximated by a linear relation,
\begin{equation}\label{eq:SP linearity}
  \mu_{\rm SP}(q_x) = \frac{q_x}{\overline{c}_1} + b
\end{equation}
where $\overline{c}_1$ is the inverse of the slope and $b$ is the intercept. 
We expect this relation to be well obeyed as long as $q_x$ is not too small. The $\overline{c}_1$ and $b$ parameters can be extracted from the $\mu_{\rm SP}(q_x)$ data, as shown in Figure~\ref{fig:cgmd linear} by linear regression.\footnote{For the linear regression fit, we compute the $\mu_{\rm SP}$ for $q_x = 0.25L_x,\,0.5L_x,\,0.75L_x,\,L_x$. We will later see in \eqref{eq:realtauxasymp} that an extra logarithm correction term is expected. However, since the logarithmic function grows very slowly, this extra term can be well described by the linear and constant terms in the linear regression estimate for the range of $q_x$ we use in our presented analysis.} 
We note that the intercept $b$ is relatively small (compared against $q_x/\overline{c}_1$) and is on the order of $\sigma_{\rm SP}$.
Hence the parameter $\overline{c}_1$ allows us to have a quick estimate of the SP mean, which is fairly accurate at large $q_x$.

The parameter $\overline{c}_1$ is a key property; larger $\overline{c}_1$ means straighter shortest path.
Within the BRW models, $\overline{c}_1$ can be interpreted as the `speed' of the shortest paths; a larger $\overline{c}_1$ corresponds to greater distances that the shortest paths can reach per unit time. 
If we ignore $b$, which is small, then $\overline{c}_1$ is a measure of the straightness of the shortest paths, and the inverse of $\overline{c}_1$ is the \textit{tortuosity}.
As we shall discuss in Section~\ref{sec:8chain}, the inverse of $\overline{c}_1$ corresponds to an important physical parameter that measures the maximum stretch that can be applied to the polymer network before significant bond-breaking events must occur.

\begin{figure}[ht!]
    \centering
        \subfigure[]
        {\includegraphics[width=0.41\textwidth]{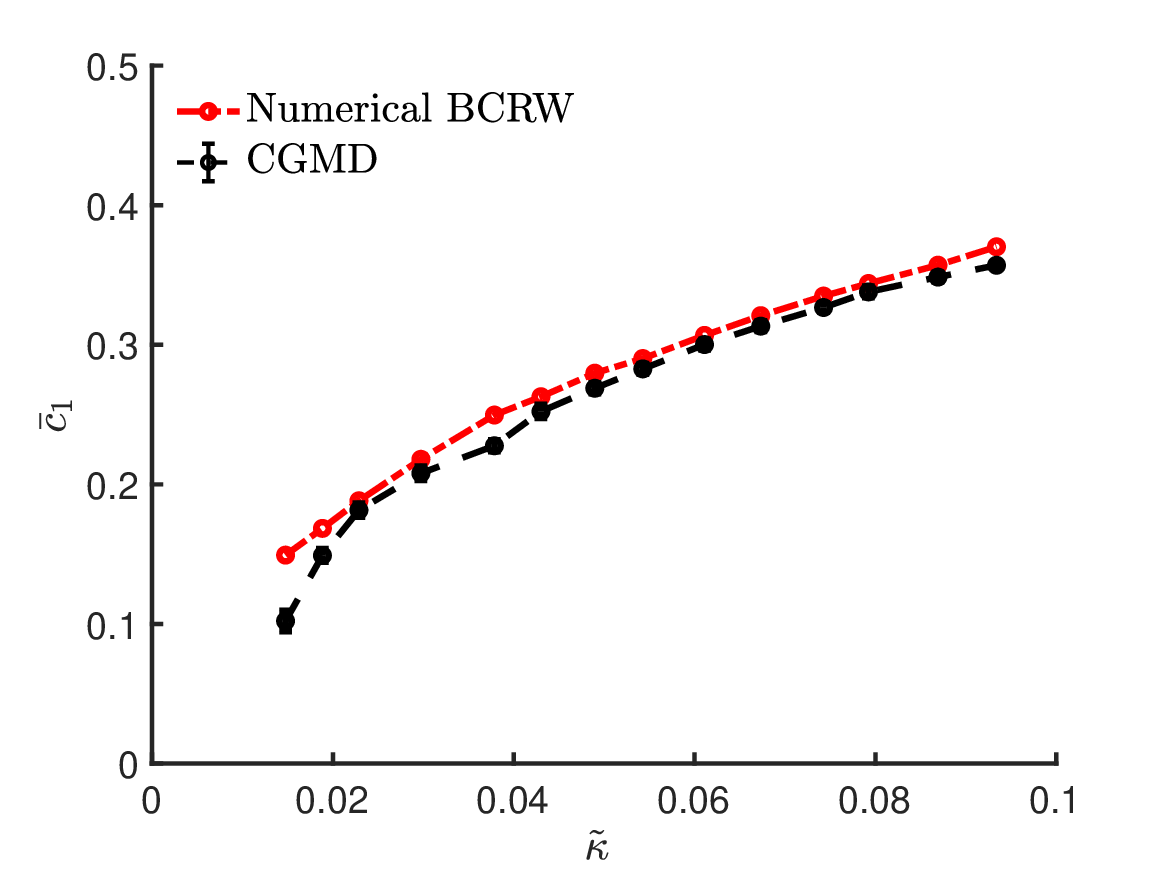}\label{fig:corr_c1_rho_cgmd_b}}
        \subfigure[]{\includegraphics[width=0.41\textwidth]{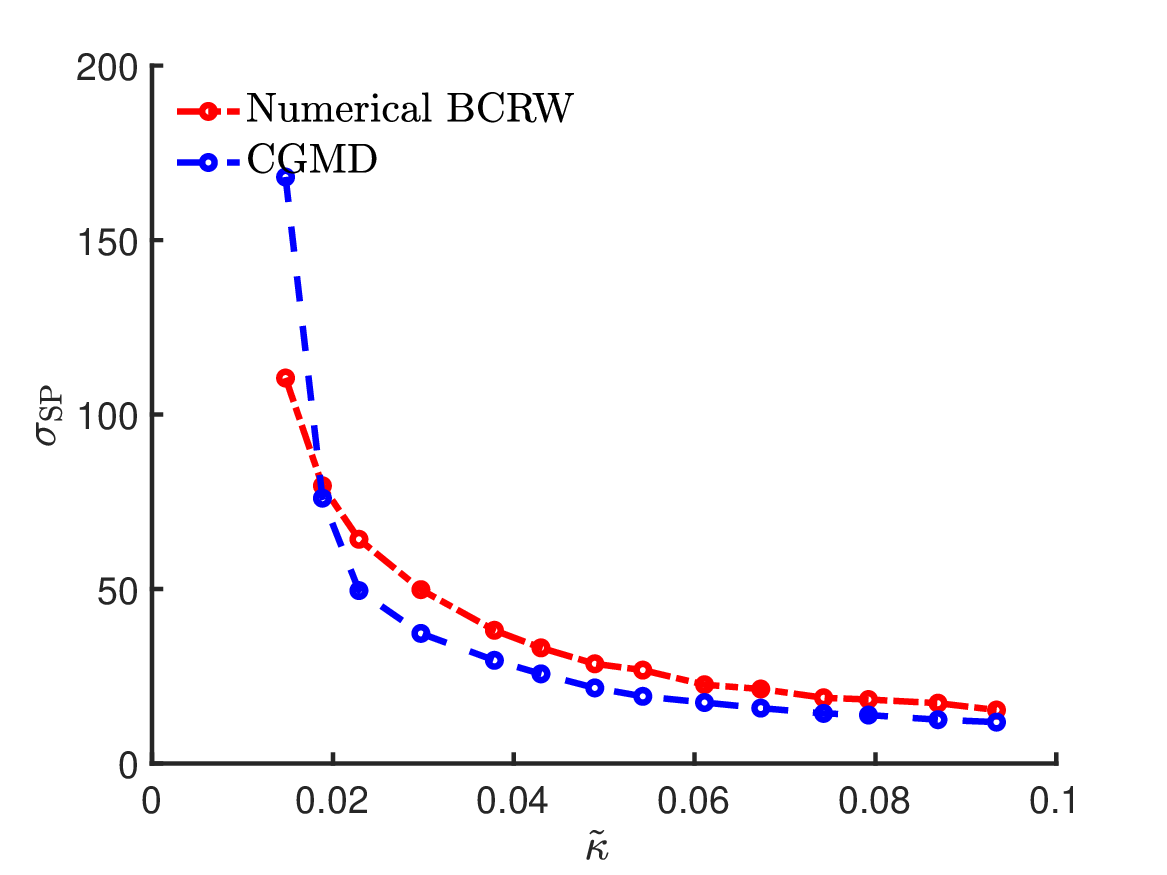}\label{fig:corr_c1_rho_cgmd_c}}
    \caption{
    (a) $\overline{c}_1$ averaged over 10 independent CGMD configurations compared against the BCRW at different cross-link densities (corresponding to branching rate $\tl \in [0.0148,0.0934]$). (b) Comparison of the standard deviations of the SP distribution from the CGMD simulation and of the FPT distribution from the numerical BCRW.
    }
    \label{fig:corr_c1_rho_cgmd}
\end{figure}

Figure~\ref{fig:corr_c1_rho_cgmd_b} plots $\overline{c}_1$ as a function of the branching rate $\tl$ from both the BCRW and the CGMD models. (The cross-link density of typical elastomers corresponds to $\tl$ values lower than the maximum value considered here.)  The CGMD results are obtained by averaging over 10 configurations at each cross-link density.
The $\overline{c}_1$ prediction of the BCRW model agrees well with the CGMD model, especially at a high branching rate $\tl$.
In general, $\overline{c}_1$ increases with $\tl$, indicating that the shortest paths become straighter at a higher cross-link density.
The overall shape of the $\overline{c}_1$-$\tl$ relation resembles that of a square root function.  This is not a coincidence; in Section~\ref{sec: analytic} we will show that if certain simplifications are introduced, the relationship becomes exactly a square root.
Figure~\ref{fig:corr_c1_rho_cgmd_c} plots the standard deviation $\sigma_{\rm SP}$ of the shortest paths as a function of $\tl$ from both the BCRW and the CGMD models for $q_x=L_x$.
Here, the prediction from the BCRW model also shows qualitative agreement with the CGMD model.
In general, $\sigma_{\rm SP}$ decreases with $\tl$, indicating that the shortest path distribution becomes more concentrated at a higher cross-link density.
The intercept $b$ also decreases with increasing $\tl$, as shown in \ref{sec:intercept}.

\subsection{SP statistics predicted by scaled BRW}
\label{sec:cgmd as brw}

While the BCRW model in Section \ref{sec:cgmd as bmrw} can successfully capture the SP statistics in the polymer network, the correlation between jumps is difficult to account for in the theoretical analysis.
To circumvent these difficulties, we introduced simplified models such as scaled BRW, where individual increments are independent but the jump distance is scaled to match the MSID at $n = 1/\tl$ (scaling the jump as $\sigma_s=\sqrt{\mathrm{MSID}(1/\tl)}\,\sigma$). 
This simplification makes theoretical predictions on SP possible (see Section \ref{sec:analytic:BRW}).

\paragraph{Numerical results}
\begin{figure}[ht!]
    \centering
        \subfigure[]{\includegraphics[width=0.41\textwidth]{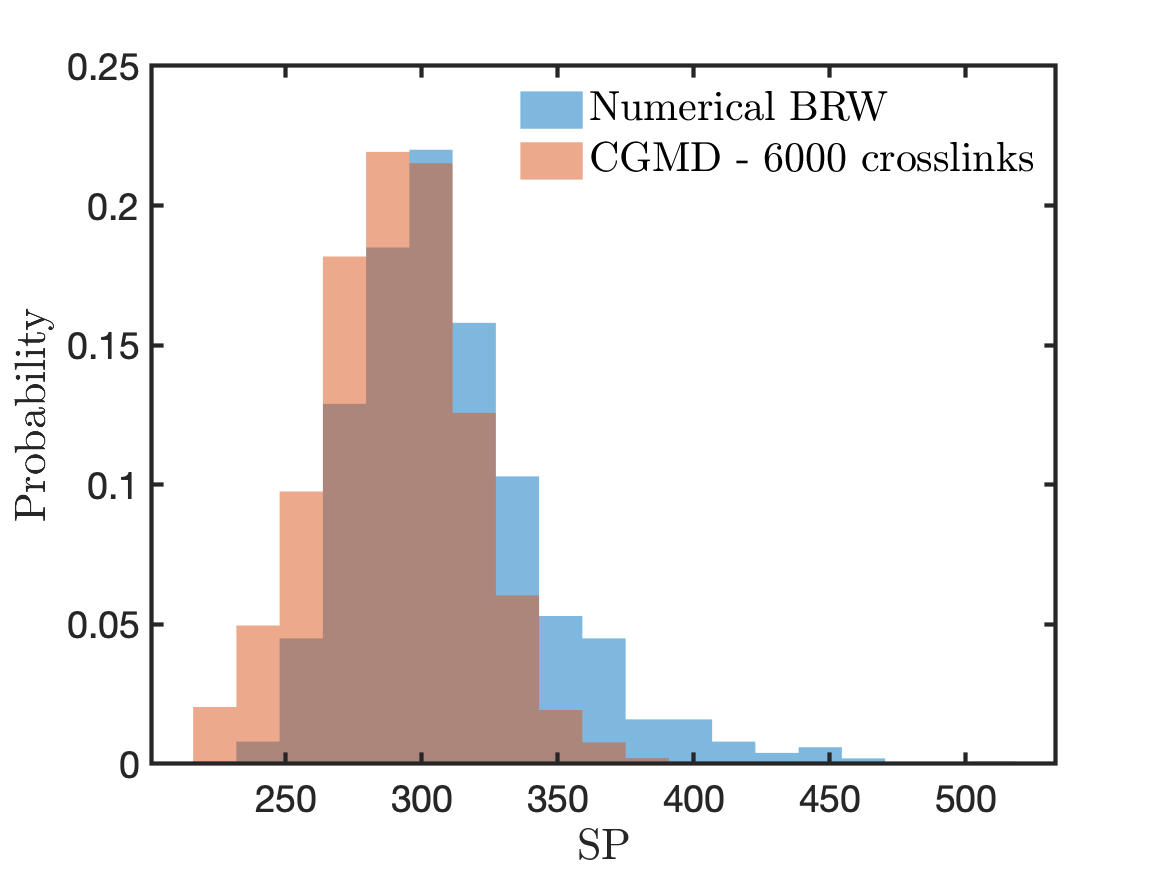}\label{fig:scaled_cgmd_dist_lower_c}}
        \subfigure[]{\includegraphics[width=0.41\textwidth]{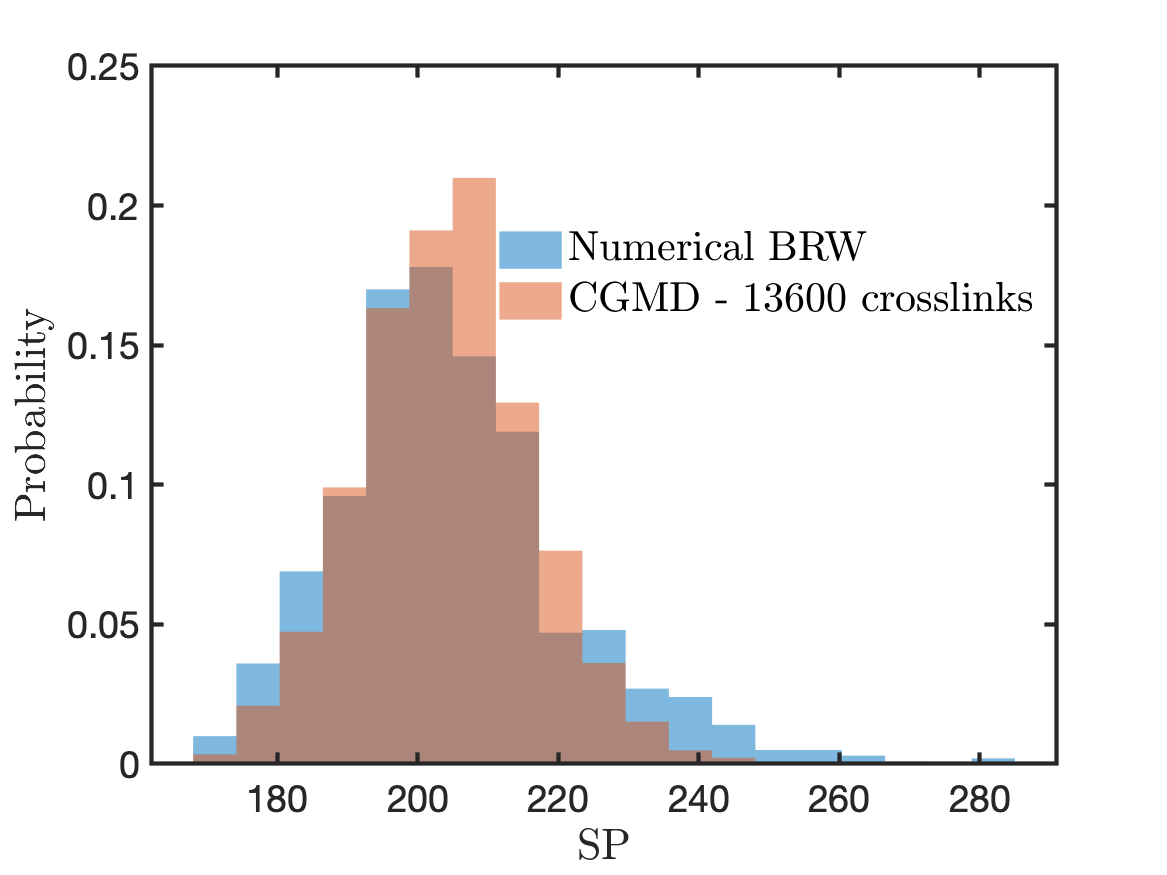}\label{fig:scaled_cgmd_dist_ref_c}}
    \caption{The SP distributions compared against a single CGMD configuration from the scaled BRW at different cross-link densities: (a) 6000 cross-links ($\tl \approx 0.0398$), and (b) 13600 cross-links  ($\tl \approx 0.0856$) for $q_x = 65.5\,\sigma$ ($\approx 98.2$ nm). 
    }
    \label{fig:scaled_cgmd_dist_lower}
\end{figure}
Figure~\ref{fig:scaled_cgmd_dist_lower} shows the
FPT distribution predicted by the scaled $(\tl,\tnu)$-BRW model, which is in good agreement with the SP distribution from the CGMD configuration with 6000 cross-links ($\tl=0.0398$) and 13600 cross-links ($\tl=0.0856$). 
The scaled BRW performs much better than the unscaled BRW (not shown in this paper) and does comparably as well as the BCRW from Section \ref{sec:cgmd as bmrw}. This is because by scaling jump steps to match the MSID at a $n = 1/\tl$, the expansion of the BRW tree in space is just as fast as the BCRW tree, thereby capturing a similar FPT or SP behavior. 

\begin{figure}[ht!]
    \centering
       \hspace{-1.0cm}
        \subfigure[]{\includegraphics[width=0.35\textwidth]{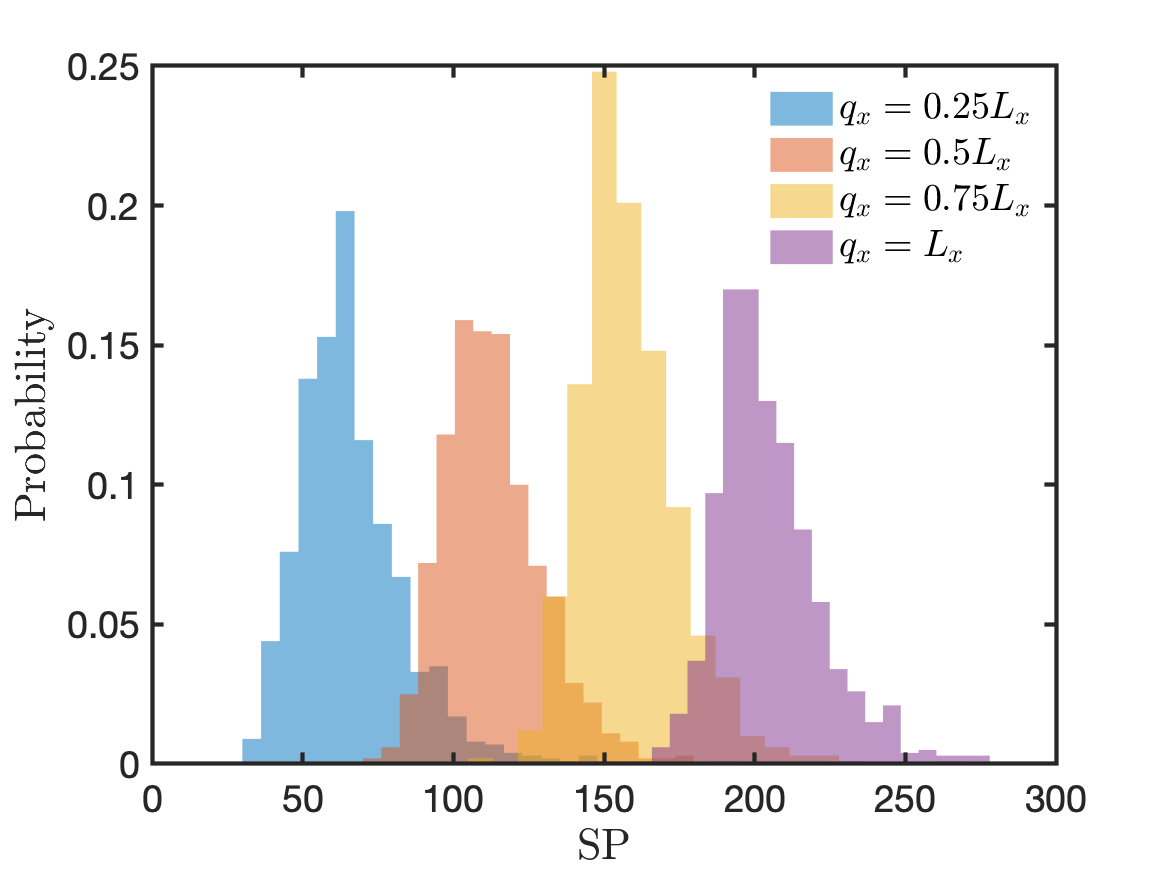}\label{fig:scaled_sig_valida}}
        \hspace{-0.6cm}
        \subfigure[]{\includegraphics[width=0.35\textwidth]{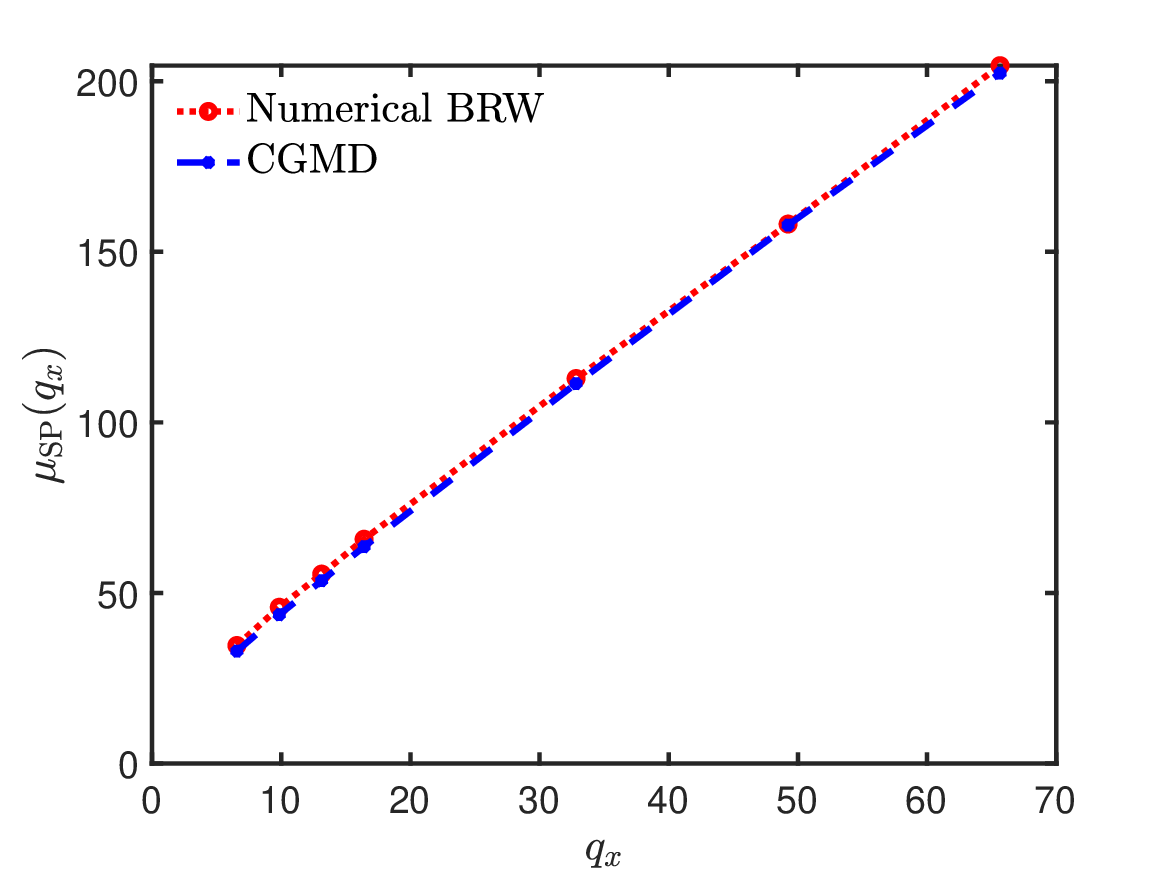}\label{fig:scaled_sig_validb}}
        \hspace{-0.6cm}
        \subfigure[]{\includegraphics[width=0.35\textwidth]{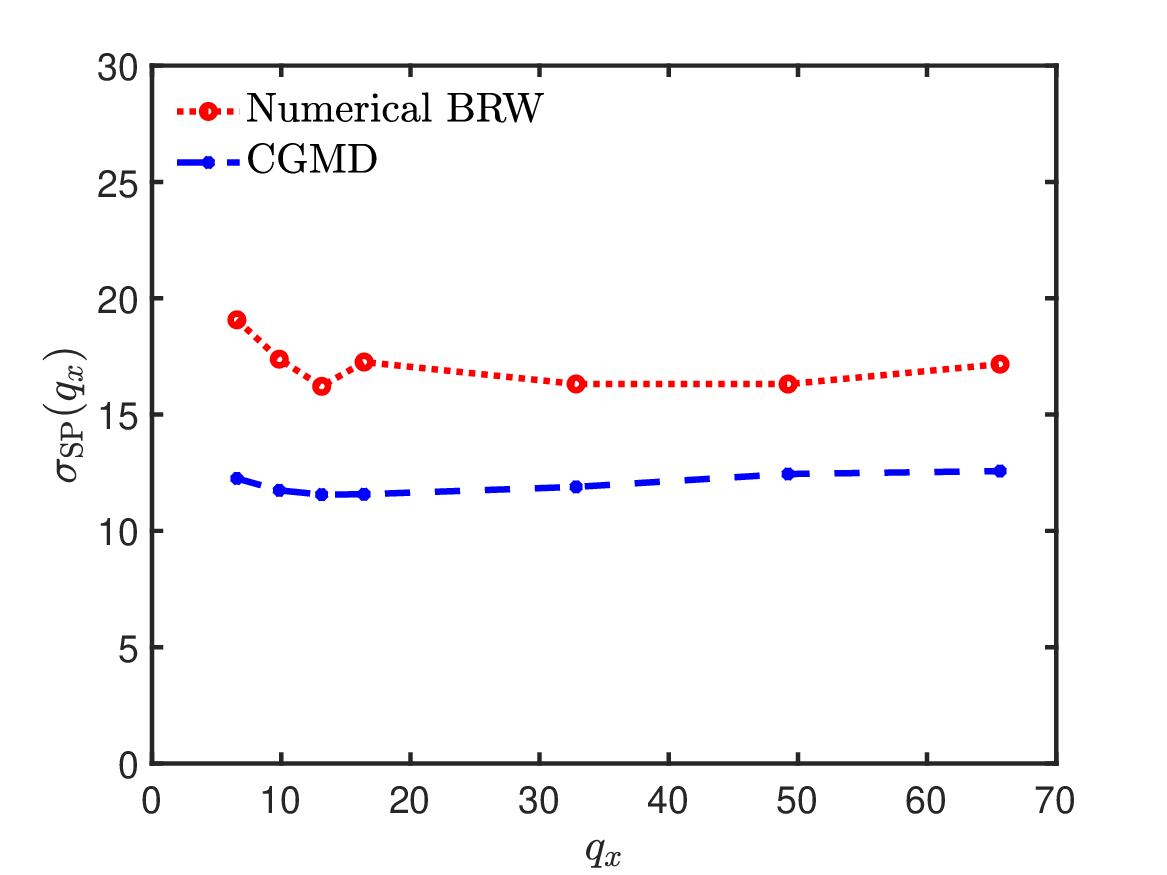}\label{fig:scaled_sig_validc}}
        \hspace{-1.0cm}
    \caption{ (a) The SP distributions at different values of $q_x$ from the numerical  BRW. The (b) mean and (c) standard deviation of the SP as a function of $q_x$. The data points correspond to $q_x = 0.1L_x,\,0.15L_x,\,0.2L_x,\,0.25L_x,\,0.5L_x,\,0.75L_x,\,L_x$. All network analysis corresponds to a single CGMD configuration with a cross-link density corresponding to $\tl=0.0856$ in the BRW model.
    }   \label{fig:scaled_sig_valid}
\end{figure}
Figure~\ref{fig:scaled_sig_valida} shows that the scaled BRW successfully captures the change in the SP distribution for different offset distances $q_x$ (at branching rate $\tl=0.0856$).
Figure~\ref{fig:scaled_sig_validb} shows that its prediction of mean SP as a function of $q_x$ agrees very well with the CGMD result.
Figure~\ref{fig:scaled_sig_validc} shows that the predicted width of the SP distribution stays nearly constant in agreement with CGMD.
\begin{figure}[ht!]
    \centering
        \subfigure[]{\includegraphics[width=0.41\textwidth]{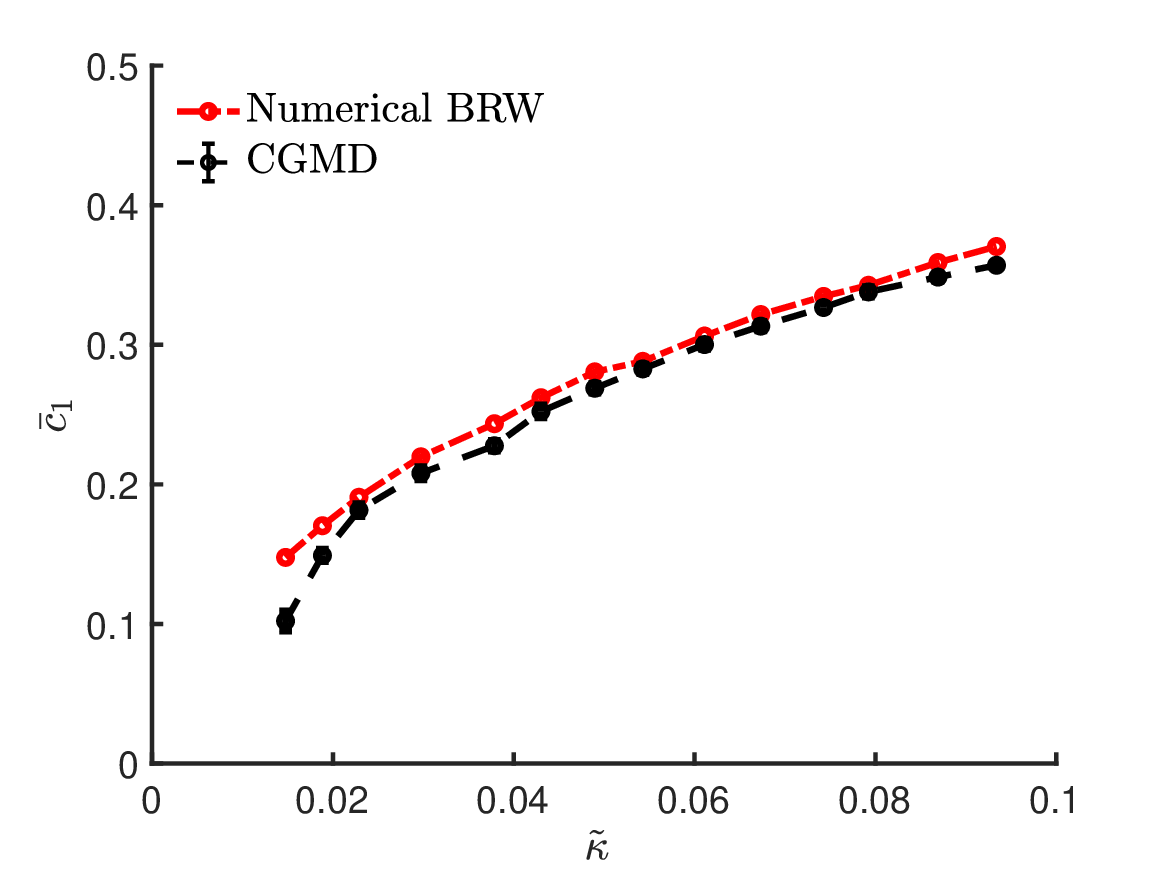}\label{fig:scaled_c1_rho_cgmd_b}}
        \subfigure[]{\includegraphics[width=0.41\textwidth]{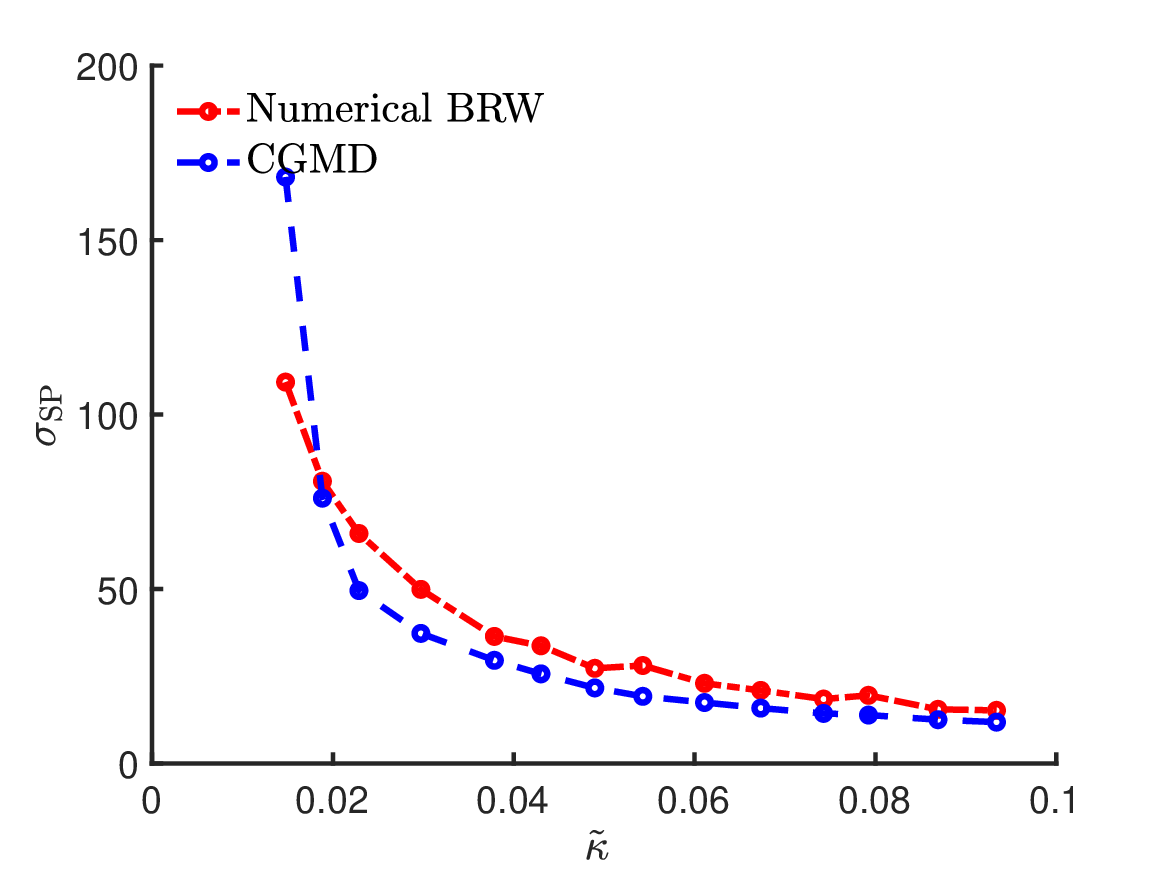}\label{fig:scaled_c1_rho_cgmd_c}}
    \caption{
    (a) $\overline{c}_1$ averaged over 10 independent CGMD configurations compared against the scaled BRW at different branching rates $\tl$. (b) Comparison of the standard deviations of the SP distribution from the CGMD simulation and of the FPT distribution from the numerical BRW.
    }
    \label{fig:scaled_c1_rho_cgmd}
\end{figure}
From the linear dependence of mean SP on $q_x$, 
we compute the $\overline{c}_1$ at different $\tl$ from the scaled BRW model, as shown in Figure~\ref{fig:scaled_c1_rho_cgmd_b}. Figure~\ref{fig:scaled_c1_rho_cgmd_c} shows the $\sigma_{\rm SP}$ at $q_x=L_x$ as a function of different $\tl$. These results are in good agreement with the CGMD results and demonstrate that the independence of the jumps in the BRW model does not affect the statistics as long as the jumps are scaled to reproduce the MSID at $n=1/\tl$.

\subsection{SP statistics predicted by GBRW}
\label{sec:cgmd as gbrw}

In this section, we present the numerical results of the GBRW model, which makes the further approximation that the random walk steps satisfy Gaussian distribution, in addition to being independent.
In order to match the MSID of the CGMD model at $n = 1/\tl$, the incremental distribution of the GBRW model is given by the 3-dimensional centered Gaussian distribution with covariance matrix $\sqrt{\mathrm{MSID}(1/\tl)/3}\, I_3$, where $I_3$ is the $3\times 3$ identity matrix.

\paragraph{Numerical results}
\begin{figure}[ht!]
    \centering
        \subfigure[]{\includegraphics[width=0.41\textwidth]{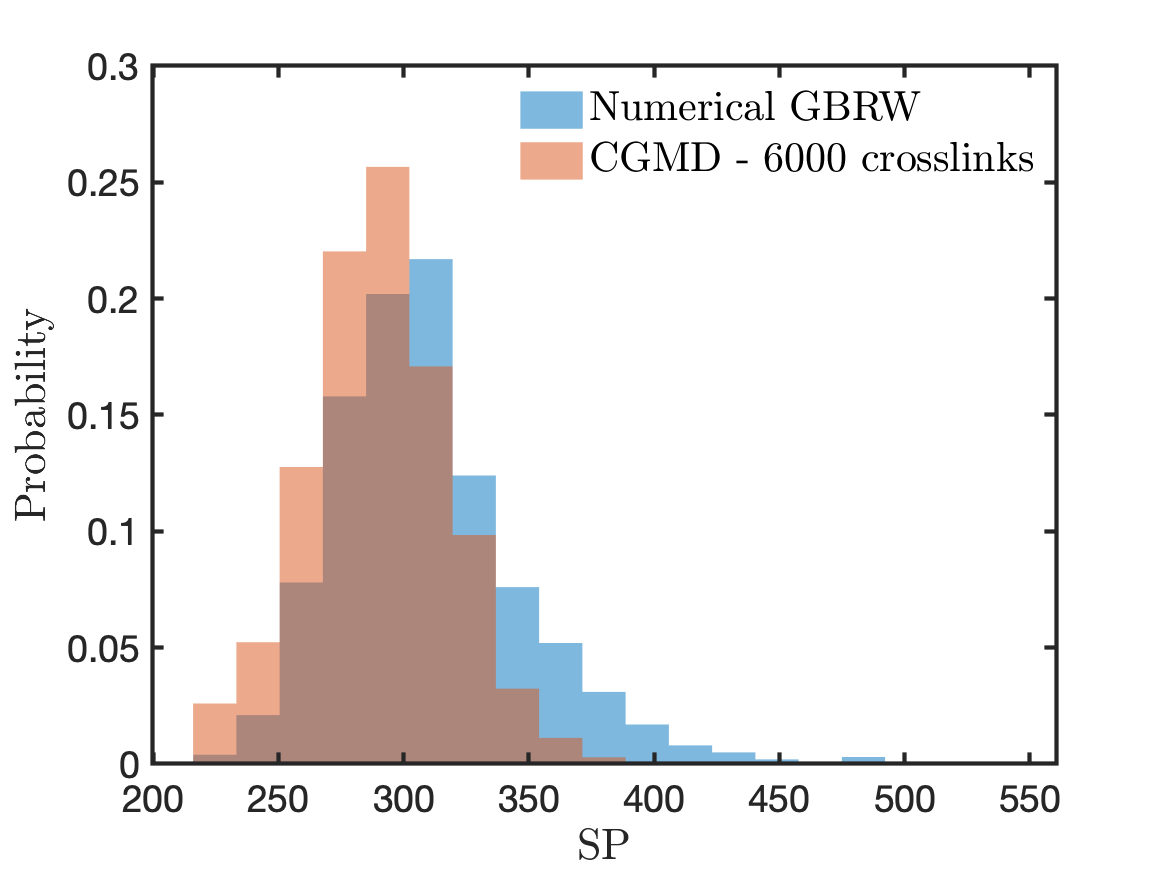}\label{fig:bbm_cgmd_dist_lower_c}}
        \subfigure[]{\includegraphics[width=0.41\textwidth]{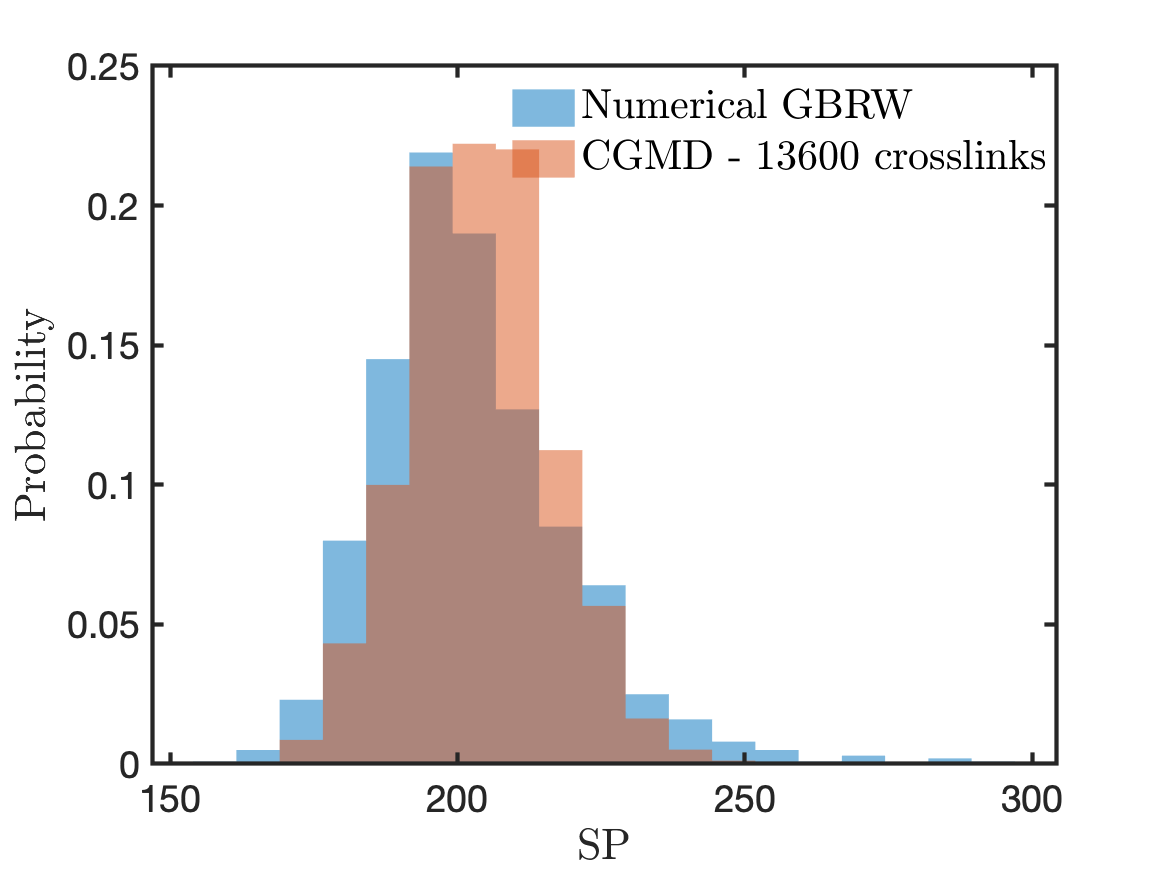}\label{fig:bbm_cgmd_dist_ref_c}}
    \caption{The SP distributions compared against a single CGMD configuration from the numerical GBRW at different cross-link densities: (a) 6000 cross-links ($\tl \approx 0.0398$), and (b) 13600 cross-links  ($\tl \approx 0.0856$) for $q_x = 65.5\,\sigma$ ($\approx 98.2$ nm). 
    }
    \label{fig:bbm_cgmd_dist_lower}
\end{figure}
Figure~\ref{fig:bbm_cgmd_dist_lower} shows the
FPT distribution predicted by the scaled $(\tl,\tnu)$-GBRW model, which is in good agreement with the SP distribution from the CGMD configuration with 6000 cross-links ($\tl=0.0398$) and 13600 cross-links ($\tl=0.0856$). 
\begin{figure}[ht!]
    \centering
       \hspace{-1.0cm}
        \subfigure[]{\includegraphics[width=0.35\textwidth]{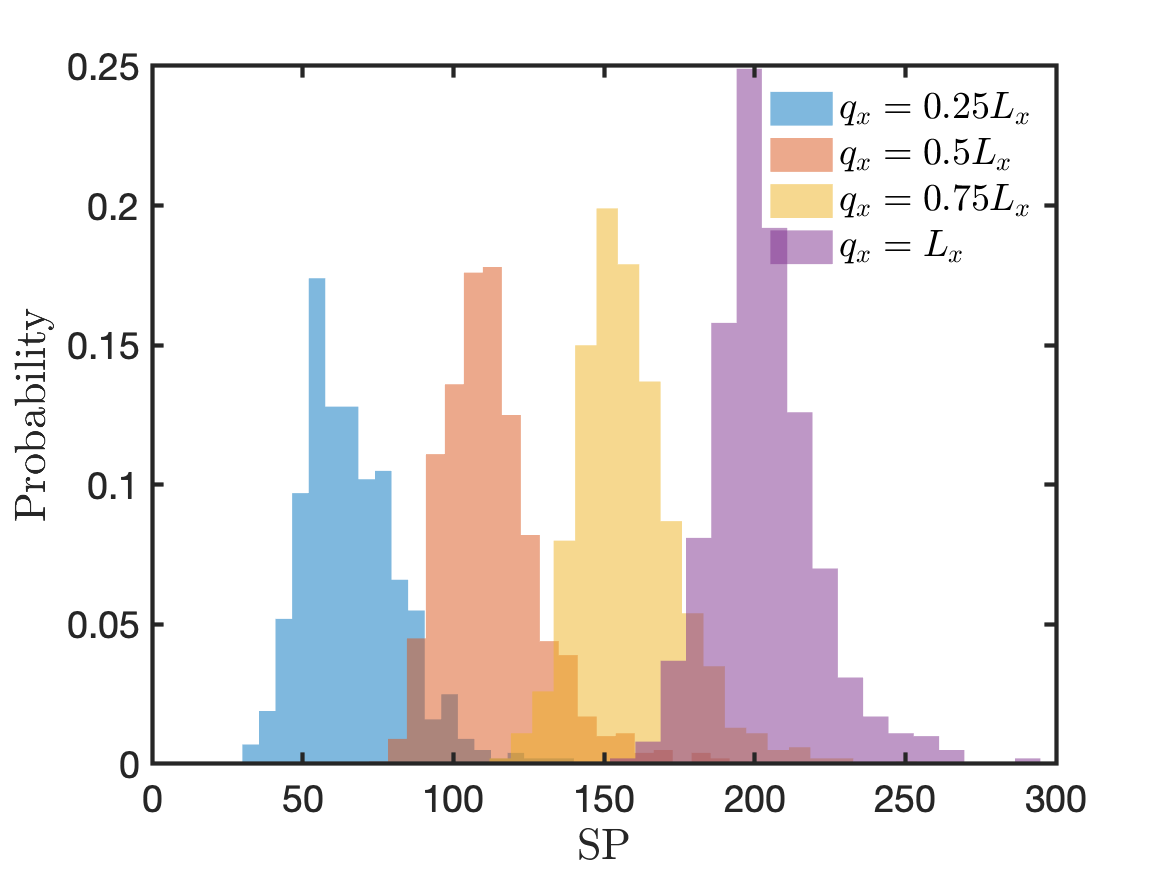}\label{fig:bbm_sig_valida}}
        \hspace{-0.6cm}
        \subfigure[]{\includegraphics[width=0.35\textwidth]{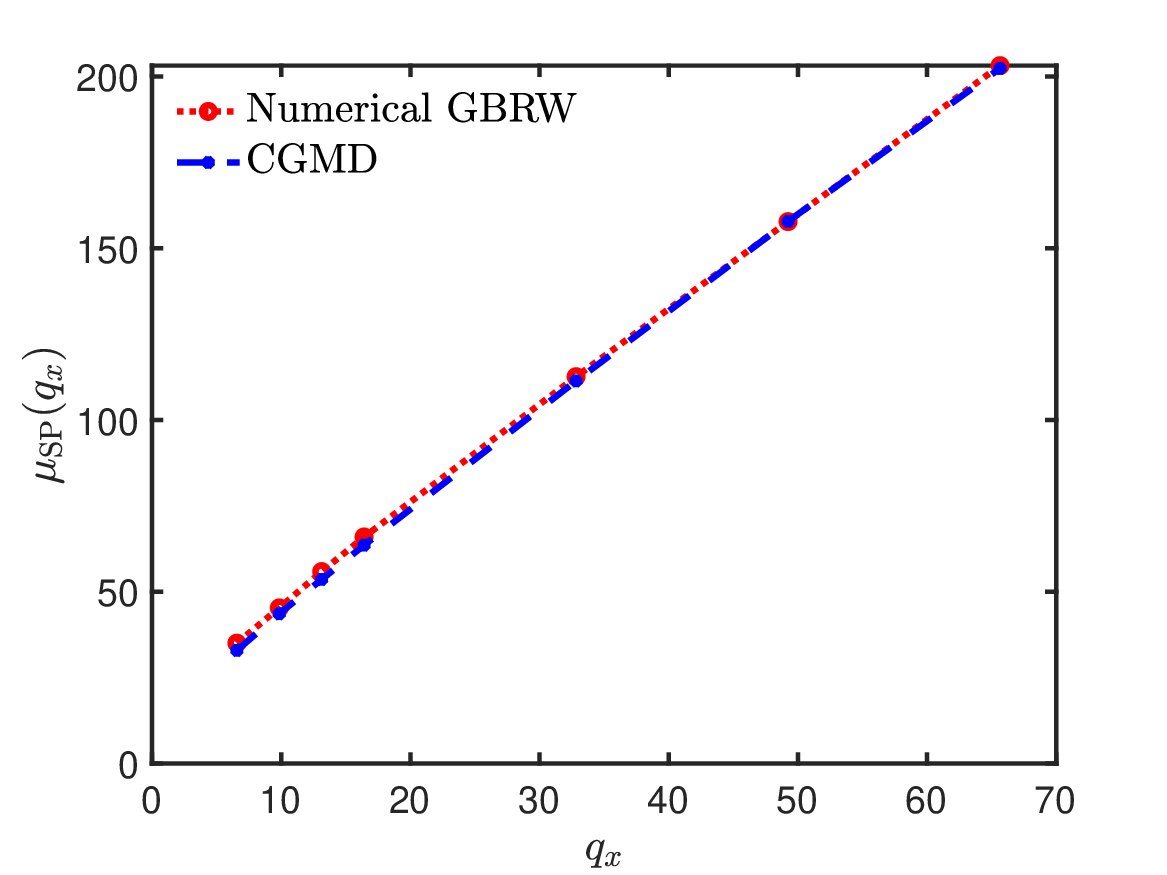}\label{fig:bbm_sig_validb}}
        \hspace{-0.6cm}
        \subfigure[]{\includegraphics[width=0.35\textwidth]{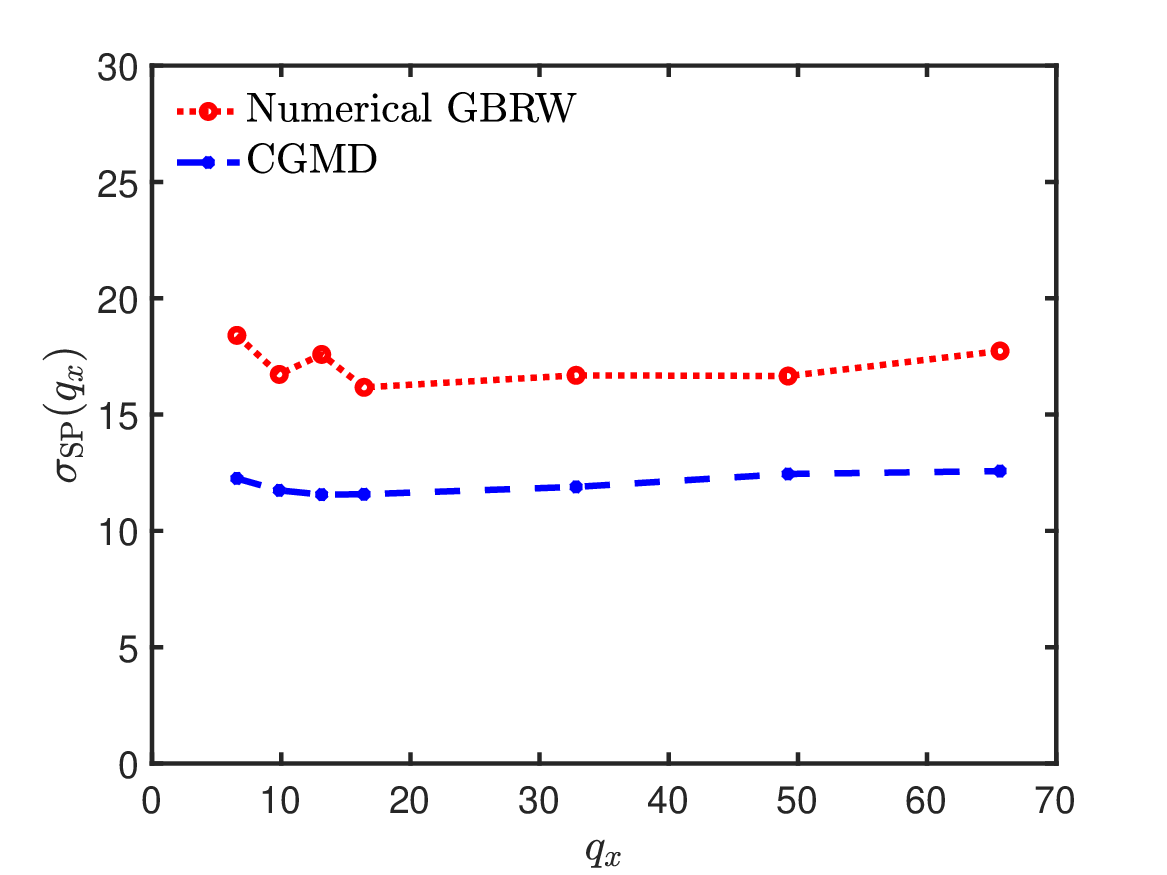}\label{fig:bbm_sig_validc}}
        \hspace{-1.0cm}       
    \caption{ (a) The SP distributions at different values of $q_x$ from the numerical  GBRW. The (b) mean and (c) standard deviation of the SP as a function of $q_x$. The data points correspond to $q_x = 0.1L_x,\,0.15L_x,\,0.2L_x,\,0.25L_x,\,0.5L_x,\,0.75L_x,\,L_x$. All network analysis corresponds to a single CGMD configuration with a cross-link density corresponding to $\tl=0.0856$ in the GBRW model.
    }
    \label{fig:bbm_sig_valid}
\end{figure}

Figure~\ref{fig:bbm_sig_valida} shows that the scaled GBRW successfully captures the change in the SP distribution for different offset distance $q_x$ (at branching rate $\tl=0.0856$).
Figure~\ref{fig:bbm_sig_validb} shows that its prediction of mean SP as a function of $q_x$ agrees very well with the CGMD result.
Figure~\ref{fig:bbm_sig_validc} shows that the predicted width of the SP distribution stays nearly constant in agreement with CGMD.
\begin{figure}[ht!]
    \centering
        \subfigure[]{\includegraphics[width=0.41\textwidth]{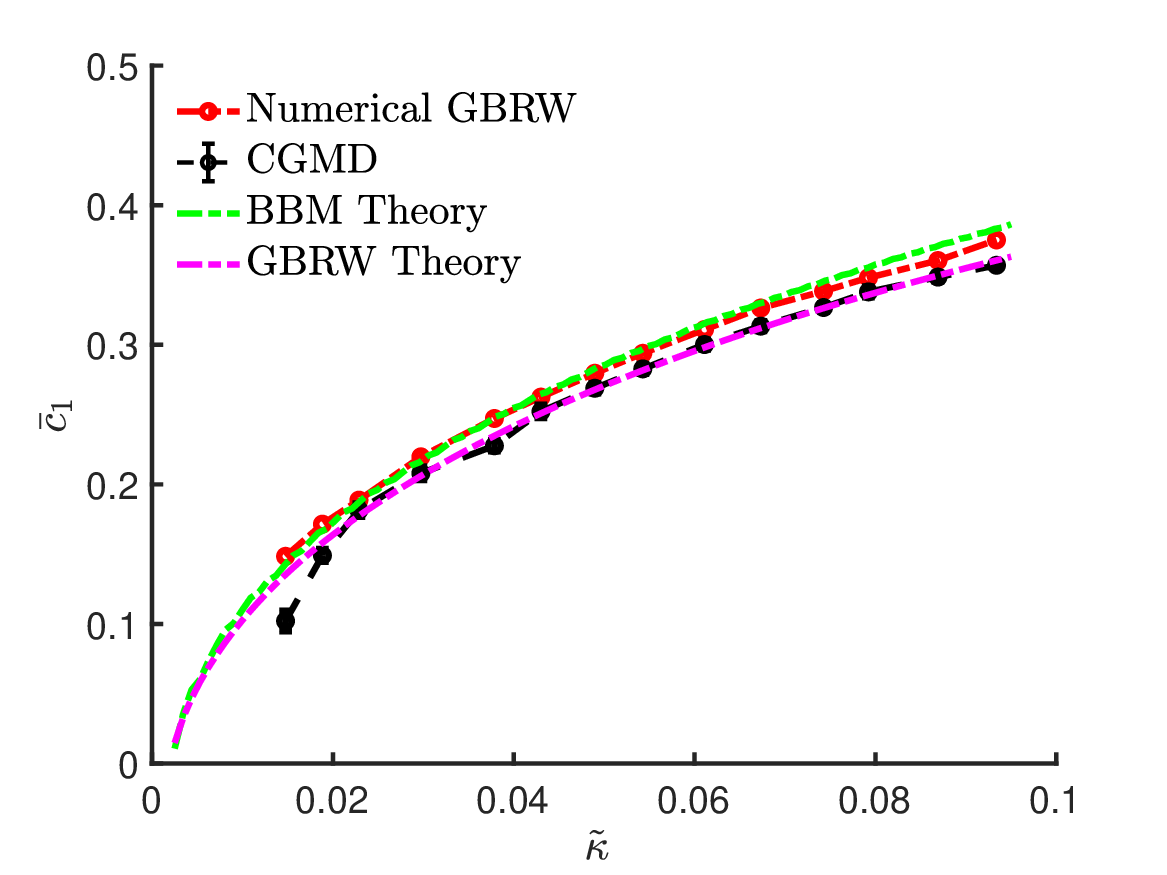}\label{fig:bbm_c1_rho_cgmd_b}}
        \subfigure[]{\includegraphics[width=0.41\textwidth]{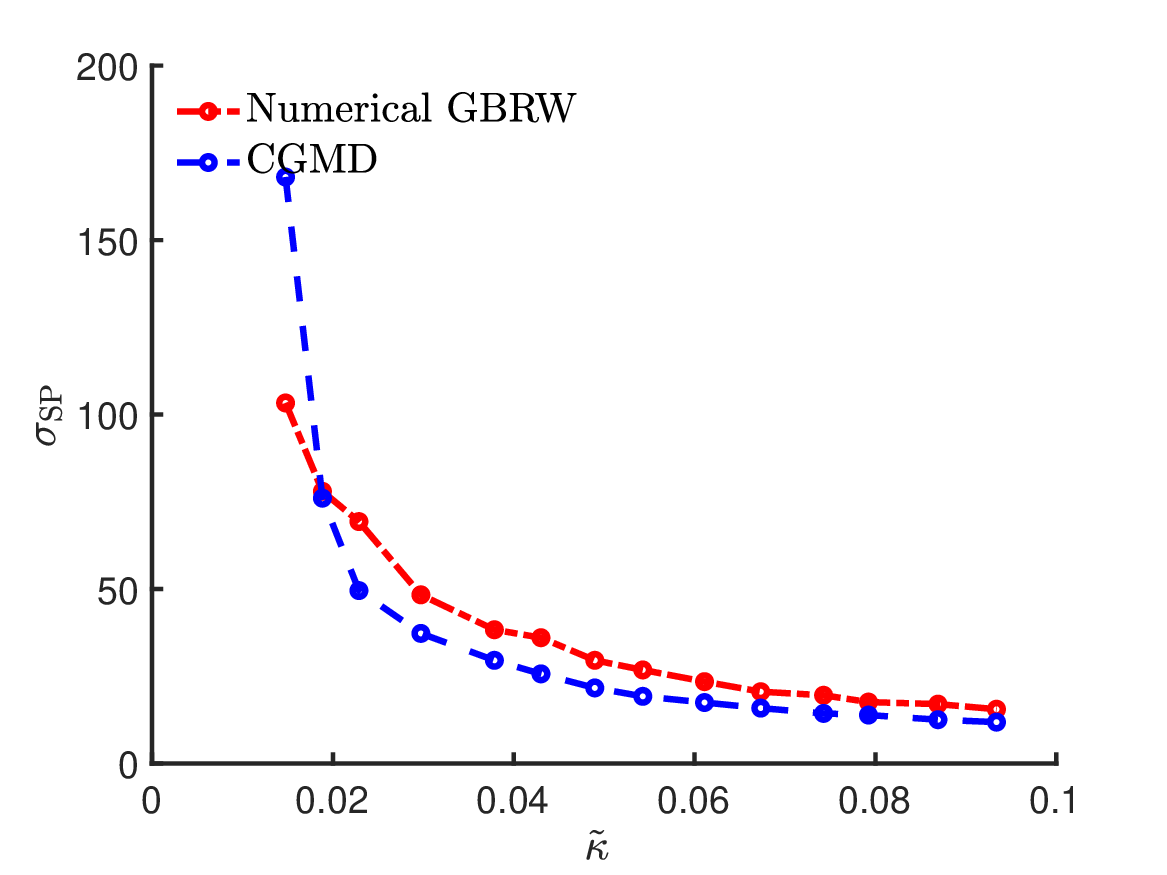}\label{fig:bbm_c1_rho_cgmd_c}}
    \caption{
    (a) $\overline{c}_1$ averaged over 10 independent CGMD configurations compared against the numerical GBRW at different branching rates $\tl$. The curves labeled \emph{Theory} refer to the analytic predictions of the FPT established in Section \ref{sec: analytic}.  Here, the \emph{BBM Theory} curve represents a proxy for the GBRW using BBM with a modified branching rate based on the termination and delayed branching regimes. \ref{sec:GBRW as BBM}. (b) Comparison of the standard deviations of the SP distribution from the CGMD simulation and of the FPT distribution from the numerical GBRW.   }
    \label{fig:bbm_c1_rho_cgmd}
\end{figure}
From the linear dependence of mean SP on $q_x$, 
we compute the $\overline{c}_1$ at different $\tl$ from the scaled GBRW model, as shown in Figure~\ref{fig:bbm_c1_rho_cgmd_b}. Figure~\ref{fig:bbm_c1_rho_cgmd_c} shows the $\sigma_{\rm SP}$ at $q_x=L_x$ as a function of different $\tl$. These results are in good agreement with the CGMD results and demonstrate that the independence and Gaussian distribution of the jumps in the GBRW model do not affect the statistics as long as the jumps are scaled to reproduce the value of MSID at $n=1/\tl$.

\section{Analytical predictions and conjectures on first passage times}\label{sec: analytic}

The numerical analysis of the BRW models in Section \ref{sec:results} demonstrated that they can capture the essential statistics of shortest paths in a polymer network.
Here we will show that some of these models (BRW and GBRW) are analytically treatable, i.e.~we can obtain analytic expressions on their first passage time (FPT) statistics.
These analytic results would provide not only a deeper understanding of the statistical behavior of shortest paths in polymer networks but also a convenient way to estimate the SP lengths given the cross-link density.
Our analysis establishes a connection between polymer physics and the extensive literature on the extremal behavior of spatial branching processes.
In particular, the inverse of parameter $\overline{c}_1$ determines the critical stretch that can be applied to the polymer before extensive bond-breaking events occur. \revision{We note that in linking $1/\overline{c}_1$ with the critical stretch, we are not implying that all chains in the polymer network deforms affinely when the material is subjected to a macroscopic stretch.  The non-affine deformation of individual subchains has been well recognized~\citep{panyukov1992microscopic,panyukov1994solid,smith1969effect,smith1974modulus} and is observed in our CGMD simulations.  However, for two nodes in the polymer network separated by a distance $q_x$, in the limit of $q_x \to \infty$ (i.e.~greatly exceeding the subchain length scale), the distance vector connecting these two nodes must deform affinely, which has also been well-recognized before~\citep{panyukov1994solid,svaneborg2005disorder}.}
%
Our analysis shows that classical polymer models that assume a periodic network topology (such as the 8-chain model) predict a much longer mean SP for a given offset distance $q_x$, hence a much larger critical stretch, compared to the more realistic models of polymer networks where cross-links are introduced randomly.

\subsection{First passage times of branching Brownian motion}\label{sec:analytic BBM}

Starting from the GBRW model considered in Section \ref{sec:cgmd as gbrw}, if we introduce a further modification where the time becomes a continuous variable (instead of being integers), then we arrive at the branching Brownian motion (BBM) model, which is even more convenient for theoretical analysis.
In the BBM model, the trajectory of each particle follows the Wiener process in $\R^d$ and they can produce new particles at any time with a fixed branching rate.
Here we define the standard BBM where the branching is binary (each branching event turns one particle into two) and the branching rate is $\tl = 1$ (and zero termination rate, $\tnu = 0$).
Denote by $M_t$ the maximal displacement of the one-dimensional standard BBM at time $t$. A classical result of \citep{bramson1978maximal} shows that
\begin{align}
    M_t=\sqrt{2}\,t-\frac{3}{2\sqrt{2}}\log t+O_\p(1).\label{eq:BBM maximum formula}
\end{align}
In the following, we present our theoretical results on FPT asymptotics for BBM and compare them against  \eqref{eq:BBM maximum formula}. We denote by $B_x$  the $d$-dimensional ball of radius one centered at $(x,0,\dots,0)\in\R^d$, which is our termination criterion for the BBM.

\begin{theorem}[FPT for standard BBM]\label{thm:BBM}
The first passage time $\tau_x$ for the standard BBM in dimension $d\geq 1$ to $B_x$ is given by 
\begin{align}
    \tau_x=\frac{x}{\sqrt{2}}+\frac{d+2}{4}\log x+O_\bP(1),\label{eq:BBM FPT formula}
\end{align}where by definition, the $O_\bP(1)$ term is tight in $x$.\footnote{A family of random variables $\{Z_x\}_{x\geq 0}$ is \emph{tight} if for every $\epsilon>0$ there are constants $b, x_0>0$ such that $\sup_{x>x_0}\bP(|Z_x|>b)<\epsilon$.}  
\end{theorem}

The proof of this theorem is given in \ref{sec:proofBBM}.
To achieve greater mathematical generality, we may use a standard scaling argument to obtain the following asymptotes for BBM with a generic branching rate $\kappa$ (different from the branching rate $\tl$ for the BRW) and a generic diffusivity constant $s$ (different from the step size $\sigma$ for the BRW).\footnote{By definition, a Brownian motion of diffusivity $s$ is equivalent in distribution to a standard Brownian motion scaled by $s$. In other words, the variance at time $t$ is $s^2t$.}

\begin{corollary}[FPT for BBM]\label{coro:bbm general}
  Consider a BBM in $\R^d$ with branching rate $\kappa>0$ and diffusivity $s>0$, then its first passage time $\tau_{\kappa,s}(x)$ to $B_x$ is given by
\begin{align}
    \tau_{\kappa,s}(x)=\frac{x}{s\sqrt{2\,\kappa}}+\frac{d+2}{4\,\kappa}\log\left(\frac{x}{s}\right)+O_{\bP}(1),\label{eq:bbm asymp}
\end{align}
where by definition, the $O_{\bP}(1)$ is tight in $x$ for each fixed $\kappa$ and $s$.
\end{corollary}
Corollary \ref{coro:bbm general} has a number of consequences that lead to a better understanding of the BRW models. As a simple example, we showcase how the delayed branching BRW with uniform jumps on $\bS^2$ can be approximated using a BBM. The quantity $x$ in \eqref{eq:bbm asymp} corresponds to the $q_x$ in the SP and FPT analysis in the previous sections, and the first two terms on the right-hand side of  \eqref{eq:bbm asymp} correspond to $\mu_{\rm SP}(q_x)$. Since BBM is isotropic in 3 dimensions, the central limit theorem yields that each step of the BRW can be approximated by a Gaussian vector $\mathrm{N}(\z,\sigma^2 I_3/3)$. In other words, we apply Corollary \ref{coro:bbm general} with diffusivity $s=\sigma/\sqrt{3}$. Next, we will derive in \ref{sec:GBRW as BBM} that the BBM branching rate $\kappa$ well approximates a BRW branching rate $\tl$ if $\kappa$ solves $e^\kappa(\kappa+\tnu)=2\,\tl$. Note that in the limit of $\tnu,\tl\to 0$, we have the approximation $\kappa\approx 2\,\tl$ (this is because, in our BRW models, every branching event produces two new particles instead of one new particle in the standard BBM model).

The asymptotic \eqref{eq:bbm asymp} will be numerically validated in \ref{sec:bbm validation} for diffusivity $s=1$, which corresponds to a random walk with jump length of $\sqrt{3}$. We leave it as a mathematical challenge for future studies to analyze the structure of the $O_\bP(1)$ term: whether it converges in law, and even whether it is of the form $c+o_\p(1)$ for some constant $c$ (as is the case for the one-dimensional maximum $M_t$).
We note that  \eqref{eq:bbm asymp} provides a theoretical justification for our numerical finding that the standard deviation of SP in the polymer network is much less than its mean at large $q_x$.

The right-hand side of   \eqref{eq:bbm asymp} for the FPT $\tau_{\kappa,s}(x)$ does not appear to be linear in $x$, in contrast to \eqref{eq:SP linearity}. In the limit as $x\to\infty$,  the quantity $\overline{c}_1$ approximately equals $s\sqrt{2\,\kappa}$ where $s=\sqrt{\mathrm{MSID}(1/\tl)/3}$, whereas in the linear fitting executed in our numerical analysis, the effect from the logarithm term is not negligible at $x=q_x$.
Nonetheless, we can still use \eqref{eq:bbm asymp} to provide an estimate of the $\overline{c}_1$ parameter obtained from the linear fit,
\begin{equation}
 {\overline{c}_1} \approx \left( \frac{1}{s\,\sqrt{2\,\kappa}} 
    +\frac{d+2}{4\kappa\, \hat{q}_x} \right)^{-1}. \label{eq:bbm_lc}
\end{equation}
where $\hat{q}_x$ lies somewhere inside the range of $q_x$ where the linear fit is performed.

\subsection{Conjectural asymptotics for branching random walk models}
\label{sec:analytic:BRW}

In Section~\ref{sec:analytic BBM}, we established the asymptotes for the maximum \eqref{eq:BBM maximum formula} and the FPT \eqref{eq:BBM FPT formula} for the BBM. In particular, the asymptote \eqref{eq:BBM FPT formula} with $d=1$ is precisely the inversion of \eqref{eq:BBM maximum formula}, in the sense that $M_{\tau_x}\approx x$ and $\tau_{M_t}\approx t$. This crucial \emph{inversion relation} motivates several conjectures for the FPT of the BRW models of interest. Before stating these conjectures, we need to understand the asymptotics of the maximum of the BRW models we introduced. 

The key difference between BRW and BBM models is that in BRW the time is discrete, which makes the analysis more difficult.
We will use integer $n$ to represent the time in BRW models.
In the following, we will state a theorem on the maximal displacement for the BRW models, which should be compared against  \eqref{eq:BBM maximum formula}.
We will work in a general dimension $d\geq 1$ and only impose mild assumptions on the jump distribution, while the termination and delayed branching schemes remain.
First, we introduce a large deviation rate function to characterize the jump vector $\bxi$ at each step of the BRW models.
We assume that the distribution of $\bxi$ is rotationally invariant.
Denote by $\xi$ the first coordinate of $\bxi$, which is a real-valued random variable. The large deviation rate function is defined as
\begin{align}
    I(x):=\sup_{\lambda>0}\Big(\lambda x-\log\phi_\xi(\lambda)\Big),\label{eq:ratef}
\end{align} 
where $\phi_\xi(\lambda):=\E[e^{\lambda\xi}]$ is the moment generating function for $\xi$ (we assume implicitly that this is well-defined for $\lambda\in\R$).
For example, for the (unscaled) GBRW model we have $I(x)=x^2/2$.

\begin{theorem}[maxima for one-dimensional delayed branching BRW]\label{thm:brw maximum}In the above setting, suppose that $(\tl,\tnu)$ satisfies $\tl+\tnu\leq 1$ and $2\tl(1-\tnu)>\tnu$. Let $M_n$ denote the maximum of the first coordinate of the $(\tl,\tnu)$-BRW. Conditioned upon survival,
    \begin{align}
        M_n=c_1n-\frac{3}{2\,c_2}\log n+O_\p(1) \, , 
        \label{eq:BRW maximum formula}
    \end{align}
where
\begin{align}
   \rho:= {\rho}(\tl,\tnu)=\frac{1-\tnu }{2}+\sqrt{\frac{(1-\tnu)^2}{4}+2\tl(1-\tnu )}\label{eq:newrho}
\end{align}
and $c_1$ and $c_2$ are constants satisfying the following equations:
\begin{align}
    I(c_1) &= \log  \rho \, , \label{eq:c1}\\
    c_2 &= I'(c_1) \, .\nonumber
\end{align}
\end{theorem}
\noindent For example, for the (unscaled) GBRW model we have $c_1 = c_2 = \sqrt{2 \log  \rho}$. The form of \eqref{eq:newrho} is slightly more involved due to the delayed branching property. Indeed, for BRW without delayed branching, \eqref{eq:newrho} is replaced by $\rho=1+2\tl-\tnu$, a well-known result in the literature \citep{addario2009minima}. A derivation of the formula \eqref{eq:newrho} can be found in the proof of Lemma \ref{lemma:number1} in \ref{sec:proofother}.

Intuitively, $\rho$ is the parameter that indicates the rate of growth of the number of particles: at time $n$, we expect that the number of particles grows like $\rho^n$ conditioned upon survival. 
The existence and uniqueness of $c_1$ in \eqref{eq:c1} is a consequence of the assumption  $2\tl(1-\tnu)>\tnu$. Indeed, this implies $\rho>1$, and we recall that $I$ is strictly increasing, concave, and continuous on $[0,\infty)$, and $I(0)=0$.

For a large $n$, the linear coefficient $c_1$ in \eqref{eq:BRW maximum formula} describes the effective velocity of the maximum of the BRW. Let us briefly explain why intuitively we expect that the effective velocity $c_1$ satisfies 
$I(c_1)=\log\rho$. Suppose that the locations of the particles at time $n$ are independent.\footnote{Of course, this is a wrong hypothesis, since two paths have the same displacements before they branch. This partly explains the logarithm correction term.} By Cram\'{e}r's theorem (see \citep{dembo2009large}), the probability of finding a certain particle located around $c_1n$ at time $n$ is roughly $e^{-(I(c_1)+o(1))n}$. Since we expect around $\rho^n$ particles at time $n$, the total number of particles near $c_1n$ at time $n$ can be estimated by $\approx \rho^ne^{-I(c_1)n}=1$, meaning that the maximum reach of the particles is close to $c_1n$ at time $n$.

It is instructive to compare \eqref{eq:BRW maximum formula} with \eqref{eq:BBM maximum formula}.
For example, \eqref{eq:BRW maximum formula} reduces to 
\eqref{eq:BBM maximum formula}
if $c_1=c_2=\sqrt{2}$.
%
Intuitively, we may consider BBM as a generalization of the GBRW model (i.e.~with Gaussian increments) with $I(x)=x^2/2$ to continuous time.
For a standard BBM model, we expect the number of particles to grow as $e^t$, and hence $\rho=e$. This amounts to $c_1=c_2=\sqrt{2\log\rho}=\sqrt{2}$.

In view of the inversion relation in the BBM model between the FPT in \eqref{eq:BBM FPT formula} and the maximal displacement in \eqref{eq:BBM maximum formula}, we pose the following conjecture for the BRW model.

\begin{conjecture}[FPT for delayed branching BRW]\label{thm:main}
In the above setting, suppose that $(\tl,\tnu)$ satisfies $\tl+\tnu\leq 1$ and $\rho(\tl,\tnu)>1$. Conditioned upon survival, the first passage time $\tau_x$ to $B_x$ satisfies the asymptotic
\begin{align}
    \tau_x=\frac{x}{c_1}+\frac{d+2}{2\bl c_1}\log x+O_\bP(1).\label{eq:tauxasymp}
\end{align}

\end{conjecture}

In a companion paper \citep{futurepaper}, we prove a slightly weaker version of \eqref{eq:tauxasymp}, with the $O_\p(1)$ term replaced by $O_\p(\log\log x)$ under certain mild assumptions on $\bxi$ (that applies for uniform distribution on $\S^2$ and Gaussian distribution on $\R^3$). That is, we prove that conditioned upon survival, 
 \begin{align}
    \tau_x=\frac{x}{c_1}+\frac{d+2}{2\bl c_1}\log x+O_\bP(\log\log x).\label{eq:realtauxasymp}
 \end{align}
The asymptotic relation \eqref{eq:realtauxasymp} is confirmed numerically in \ref{sec:brw validation}.

The formula \eqref{eq:realtauxasymp} provides the analytic prediction of the quantity $\overline{c}_1$ that can be compared against the CGMD results, as shown in Figure \ref{fig:bbm_c1_rho_cgmd_b}.
As remarked above, $c_1$ becomes $\overline{c}_1$ only in the limit of $x = q_x\to\infty$.
At a finite offset distance $q_x$, the logarithm term in \eqref{eq:realtauxasymp} is not negligible,
\begin{equation}
 {\overline{c}_1} \approx \left( \frac{1}{c_1} 
    +\frac{d+2}{2 c_1 c_2\, \hat{q}_x} \right)^{-1}, \label{eq:gbrw_lc}
\end{equation}
where $\hat{q}_x$ lies somewhere inside the range of $q_x$ where the linear fit is performed.
Figure~\ref{fig:bbm_c1_rho_cgmd_b} shows the analytic prediction of $\overline{c}_1$ (for $\hat{q}_x = L_x/2$) as a function of $\tl$ for the BBM model (see \eqref{eq:bbm_lc}) and the GBRW model (see \eqref{eq:gbrw_lc}), which is in good agreement with both the numerical implementation of GBRW model and the CGMD results.

\subsection{Critical stretch and comparison with classical periodic network model for polymers}\label{sec:8chain}
The theoretical analysis for the BBM, GBRW, and BRW models above provides an accurate estimate of the FPT distribution that explains the SP distribution in a polymer network given the cross-link density from the CGMD simulations. 
The mean SP, in particular the parameter $\overline{c}_1$, determines the stretchability of the elastomer before the onset of significant bond-breaking events.
To see why this is the case, consider a shortest path of contour length $L_{\rm SP}$ connecting two nodes at distance $q_x$ apart.
If a stretch $\lambda$ is applied to the network but no bonds break, then the two nodes are separated by $\lambda \, q_x$ and the shortest path contour length stays at $L_{\rm SP}$.
Because the end-to-end distance of a shortest path can \revision{not} be stretched \revision{much} longer than its contour length~\footnote{\revision{Due to enthalpic stretching, we found that in our CGMD simulations, the SP paths can be stretched upto 10\% longer than its original contour length before breaking.  But this is a small effect that will be ignored in this discussion.}}, we have $\lambda \, q_x \le L_{\rm SP}$, i.e.~$\lambda \le L_{\rm SP} / q_x \approx 1/c_1$.
Therefore, we can define $\lambda_c = 1/c_1$ as the critical stretch that can be applied to the elastomer before extensive bond-breaking occurs.

We now discuss how the critical stretch $\lambda_c$ depends on the cross-link density, characterized by parameter $\tl$.
For simplicity, here we shall ignore the effect of correlation between consecutive jumps in the random walk, i.e., modeling the polymer chain as an uncorrelated random walk ($\sigma = 1$).
We pick the BBM model (at $d=3$) with unit jump length (corresponding to $s=\sqrt{{1}/{3}}$), where the analytic expression is the simplest.
In the limit of $\tnu,\tl\to 0$, $\kappa \approx 2\tl$.
In the limit of large $q_x$, 
we have $c_1 \to s \sqrt{2 \kappa} \approx \sqrt{{4\,\tl}/{3}}$ and hence
 \begin{equation}
     \lambda_c^{\rm BBM}\approx\sqrt{\frac{3}{4\,\tl}} \, . 
 \end{equation}
On the other hand, if one assumes the periodic network structure of the 8-chain (a.k.a.~Arruda-Boyce) model~\citep{arruda1993three}, it can be easily shown that $c_1 = \sqrt{\tl/3}$ and hence\footnote{Strictly speaking, within the 8-chain model, the contour length of the shortest path can never be exactly straight because it must pass through the body-center of the unit cell.  Accounting for this constraint leads to an implicit equation for the critical stretch, ${\lambda_c^2+{2}/{\lambda_c}}= {{3}/{\tl}}$, the solution of which is close to (\ref{eq:lc_bbm}) in the limit of large $\lambda_c$.}
 \begin{equation}
     \lambda_c^{8\text{-chain}}\approx\sqrt{\frac{3}{\tl}}. \label{eq:lc_bbm}
 \end{equation}

Note that both models predict the same $\tl^{-1/2}$ scaling with the cross-link density.  However, $\lambda_c^{8\text{-chain}} / \lambda_c^{\rm BBM} \approx 2$.  This means that relative to the more realistic model that accounts for the random distribution of cross-links, the 8-chain model overestimates the critical stretch for extensive bond-breaking by about a factor of 2 (in the limit of $q_x \to \infty$, $\tnu, \tl \to 0$, see \ref{app:abm_comp} for comparison against other spatial branching models).

\begin{figure}[ht!]
    \centering
        \subfigure[]{\includegraphics[trim={1.5cm 0 1cm 0.4cm},clip,width=\textwidth]{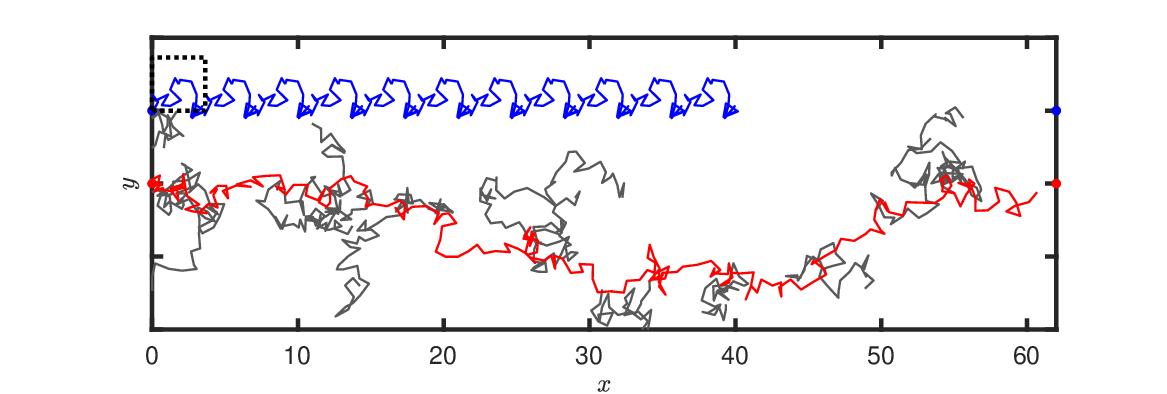}\label{fig:small_sp_traj}}
        \vspace{-1em}
        
        \subfigure[]{\includegraphics[trim={1.5cm 0 1cm 0.4cm},clip,width=\textwidth]{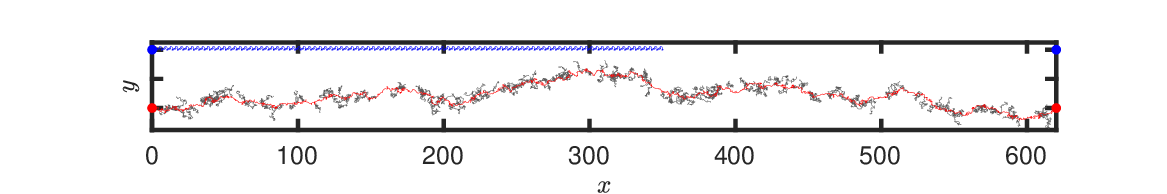}\label{fig:large_sp_traj}}
    \caption{
    Shortest path trajectory for the BRW model (red path with selected branches shown in black) and 8-chain model (blue path) at $\tl=0.1$ for the same amount of time it takes for the BCRW path to reach its destination for (a) $q_x = 62$ and (b) $q_x = 620$.  }
    \label{fig:sp_traj}
\end{figure}
To provide an intuitive understanding of the difference between BRW (approximated by the BBM estimates in \eqref{eq:bbm asymp}) and the periodic model of the polymer network, we look at the trajectory of the shortest path from the BRW model and the 8-chain model for $\tl=0.1$ and $\tnu=2/500$ for which $\kappa\approx1.65\,\tl$. Here we expect from \eqref{eq:bbm asymp} that $\lambda_c^{8\text{-chain}} / \lambda_c^{\rm BBM}\approx 1.82$ (in the limit of $q_x \to \infty$). Figure~\ref{fig:small_sp_traj} shows the trajectory of the shortest path generated by the BRW model (shown in red). A few side branches connected to the shortest path are shown in black.
In this case, the mean shortest path is able to reach a distance of $q_x = 62$ 
in $220$ steps.
The blue path in Figure~\ref{fig:small_sp_traj} shows the shortest path in a network consisting of periodic repeating cells (dashed line) each containing eight chains (only two chains are shown for clarity).
In this case, the shortest path only reaches a distance of $q_x^{8\text{-chain}} \approx 40$ in $220$ steps. 
In other words, the shortest path in the 8-chain model of the polymer network is much more tortuous than that predicted by the more realistic BRW models.
As a result, the 8-chain model would overestimate the critical stretch by which extensive bond-breaking events occur. Here, the ratio of the distance covered in the same duration as the FPT of the BRW ($q_x/q_x^{8\text{-chain}}$) is $\approx 62/40 \approx 1.54$. On increasing the $q_x$ to 620, this ratio increases to $\approx 1.77$, as shown by the paths in Figure~\ref{fig:large_sp_traj}.
This is smaller than the theoretical estimate of $\lambda_c^{8\text{-chain}} / \lambda_c^{\rm BBM}\approx 1.82$  because in this illustration $q_x$ is still relatively small.
We have numerically verified that for a very large offset distance $q_x$ the limit $\lambda_c^{8\text{-chain}} / \lambda_c^{\rm BBM}\approx 1.82$ is indeed recovered.

The 8-chain model has been very successful in predicting the elastic (i.e.~reversible) responses of elastomers, which is the result of the aggregated response of all polymer chains in the network. However, at the critical stretch, $\lambda_c$, we are entering the regime of irreversible strain-induced damage.
In this regime, the polymer chains that lie on the shortest paths play a decisive role because they constrain the maximum stretch that can be applied unless bond breaking occurs.
The spatially branching processes provide a more realistic description of the length distribution of shortest paths and hence a more physical description of the polymer as it enters the irreversible strain-induced damage regime.

\section{Conclusions and outlook} 
\label{sec:conclusion}
In this paper, we present a class of branching random walk (BRW) models whose first passage time (FPT) statistics are used to generate shortest path length (SP) statistics to understand the structure of polymeric networks as modeled by CGMD simulations. The effective branching rate (for a given cross-link density) in our BRW models is obtained from the inter-cross-link chain length distribution from the CGMD simulation cell. 
We analyze the FPT of the BRW and relevant models from both numerical and theoretical aspects. 
The numerical simulations show that multiple BRW models, with various levels of idealizations, are able to reproduce the SP distribution (in terms of both mean and standard deviation) of the polymer model as modeled by CGMD.
Our theory yields an explicit relation between the mean SP (FPT) and the offset distance $q_x$ as a function of the cross-link density. The theoretical estimate validates the results from our numerical approach and is in good agreement with the CGMD calculations.
The FPT or SP dependency on the offset distance at equilibrium serves as an indicator of the stretch limit or stretchability of the polymer. The theoretical estimate of the FPT from the spatial branching processes shows a much lower stretch limit as compared to idealized representations of polymer networks using periodic repeating structures. This shows that treating the polymer network as a BRW captures a more realistic response of the material response in the regime of extensive bond-breaking.
\revision{We expect our work to complement existing approaches~\citep{higgs1988polydisperse} to connect the microscopic statistical properties of the polymer network with the macroscopic stress-strain response of the elastomeric material.}

In this paper, we have analyzed polymer networks prepared at equilibrium, before any deformation and strain-induced damage, and the SP distribution is expected to be isotropic.
As a result, the BRW has been modeled with isotropic jumps in space. 
The next step is to see whether the BRW models can be used to understand the evolution of SP statistics as bonds break by deformation.
In elastomers with irreversible cross-links, we hypothesize that the reduction in the number of bonds (cross-link density) with loading may be modeled by a modification in the branching rate of the BRW or by introducing anisotropy in the jumps as governed by the directional dependence of the SP. Reversibly cross-linked systems present a challenge since there is no reduction in the cross-link density (broken cross-links readily reform at ultrafast timescales). It appears that the BRW would have to model anisotropy in each jump to account for the evolution in the SP distribution. 
Modeling the SP evolution with strain in polymers with different types of cross-links would be the next challenge in representing polymer network evolution using spatial branching processes.

\section*{Conflict of Interest}
The authors declare no competing interests.

\section*{Acknowledgements}
We thank Haotian Gu and Lenya Ryzhik for their helpful discussions.  The material in this paper is based upon work supported by the Air Force Office of Scientific Research under award number FA9550-20-1-0397. Additional support is gratefully acknowledged from NSF 1915967, 2118199, 2229012, 2312204.

\section*{Data Availability}
All details are available in the main text and the Supplementary appendices. The open source code for the branching random walk calculations along with the reference CGMD and theoretical predictions can be accessed from \href{https://gitlab.com/micronano_public/PolyBranchX}{PolyBranchX}.

\appendix

\theoremstyle{plain}

\newpage

\section{Proofs of results in Section \texorpdfstring{\ref{sec: analytic}}{}}\label{sec:proofs}

\subsection{Proof of Theorem \texorpdfstring{\ref{thm:BBM}}{} and related discussions}\label{sec:proofBBM}

 In this appendix, we prove Theorem \ref{thm:BBM}. Let us first summarize a few preliminary results from the literature. In the following, we work in a general dimension $d\geq 1$ and assume the BBM is standard. 
Denote by $M_t$ the maximal displacement of the BBM at time $t$. The first precise asymptotic in dimension $d=1$ for $M_t$ was given by \citep{bramson1978maximal}, whose proof was later considerably simplified by \citep{roberts2013simple}. More precisely, we have
\begin{align}
    M_t=\sqrt{2}\,t-\frac{3}{2\sqrt{2}}\log t+O_\p(1).
\end{align}Finer behavior near the frontier was analyzed by \citep{aidekon2013branching} and \citep{arguin2013extremal}. In dimension $d>1$, we mention the very recent works of \citep{kim2023maximum} and \citep{berestycki2021extremal} on characterizing the behavior of the maximal norm of the BBM in $\R^d$, where the precise asymptotic was first given by \citep{mallein2015maximal}.

The BBM is intimately connected to the \emph{Fisher-KPP equation} introduced by  \citep{fisher1937wave,kolmogorov1937etude}---it is shown by \citep{mckean1975application} that $v(t,x)=\p(M_t>x)$ solves the Fisher-KPP initial value problem $v_t=v_{xx}/2+v-v^2$, with initial condition $v(0,x)=\bone_{\{x\leq 0\}}$.\footnote{Here and later, we follow the probabilists' notation with the factor of $1/2$ in front of the Laplacian.} Analogously, as pointed out by  \citep{mallein2015maximal}, in dimension $d\geq 1$, $v(t,x)=\p(\exists u\in \mathcal N_t:\n{X_t(u)-\bx}\leq 1)$ solves the multi-dimensional Fisher-KPP equation 
\begin{align}
    \begin{cases}
        v_t=\frac{1}{2}\Delta v+v-v^2,~t>0,~\bx\in \R^d,\\
        v(0,\bx)=\bone_{\{\bx\in B_{\z}(1)\}}.
    \end{cases}\label{eq:fkpp2}
\end{align} 
Here and later, we let $\mathcal N_t$ be the collection of particles at time $t$ and $X_t(u)$ the location of a particle $u\in\mathcal N_t$ at time $t$. The multi-dimensional   Fisher-KPP equation has been studied by \citep{gartner1982location} in a probabilistic framework and later by \citep{ducrot2015large} and \citep{roquejoffre2019sharp} using a PDE approach, where it is shown that the level set $v=1/2$ appears at $x=\sqrt{2}t-((d+2)\log t)/(2\sqrt{2})+O_\p(1)$.\footnote{Inverting this relation gives $t=x/\sqrt{2}+((d+2)\log x)/4+O_\p(1)$, giving precisely \eqref{eq:BBM FPT formula}.}  Consequently, \citep{mallein2015maximal} stated that ``the probability to find an individual  within distance 1 of a given point $\bx$ is small if $\n{\bx}\gg \sqrt{2}t-((d+2)\log t)/(2\sqrt{2})$ and large
if  $\n{\bx}\ll \sqrt{2}t-((d+2)\log t)/(2\sqrt{2})$." Formally, with $A(x)=x/\sqrt{2}+((d+2)\log x)/4$, the following result holds.  
\begin{theorem}\label{thm:BBMball}Fix $\ee>0$. There exists $C=C(\ee)>0$ such that for every $\bx\in\R^d$ with $\n{\bx}$ large enough and every $t<A(\n{\bx})-C$,
    $$\p(\exists u\in \mathcal N_{t}:\n{X_t(u)-\bx}\leq 1)<\ee,$$and for every $t>A(\n{\bx})+C$,
    $$\p(\exists u\in \mathcal N_{t}:\n{X_t(u)-\bx}\leq 1)>1-\ee.$$
\end{theorem}

\begin{remark}
In a similar manner,  finding the FPT for a domain in $\R^d$ is equivalent to solving the multi-dimensional Fisher-KPP equation with a  Dirichlet boundary condition; see e.g., the derivation in \citep{mckean1975application} and the Appendix of \citep{bramson1978maximal}. More precisely, consider the boundary value problem
\begin{align}
    \begin{cases}v_t=\frac{1}{2}\Delta v+v-v^2,~t>0,~\bx\in \R^d,\\
v(0,\bx)=\bone_{\{\bx\in B_\z(1)\}},\\
v(t,\bx)=1,~t>0,~\bx\in\bS^{d-1}.
\end{cases}\label{eq:multi fkpp}
\end{align}
It holds that $v(t,\bx)=\p(\tau_{\bx}<t)$, the probability that  the BBM starting from $\z\in\R^d$ has already reached $B_{\bx}(1)$ by time $t$. Nevertheless, we are not aware of studies of \eqref{eq:multi fkpp}, given that it is a boundary value problem instead of an initial value problem. Our Theorem \ref{thm:BBM} shows that asymptotically, the solutions to the two problems \eqref{eq:fkpp2} and \eqref{eq:multi fkpp} are reasonably close.

\end{remark}

\begin{proof}[Proof of Theorem \ref{thm:BBM}]
The upper bound of $\tau_x$ follows directly from Theorem \ref{thm:BBMball}. For the lower bound, we fix $\ee>0$ and an increasing sequence $a(x)\to\infty$, and it suffices to show for $x$ large enough,
\begin{align}
    \p\left(\tau_x\leq \frac{x}{\sqrt{2}}+\frac{d+2}{4}\log x-a(x)\right)<\ee.\label{eq:lbtoprove}
\end{align}
Observe that with $\bx=(x,0\dots,0)\in\R^d$ and $t_0=x/\sqrt{2}+((d+2)\log x)/4-a(x)/2$,
\begin{align*}
    &\hspace{0.5cm}\p\left(\tau_x\leq \frac{x}{\sqrt{2}}+\frac{d+2}{4}\log x-a(x)\right)\\
    &\hspace{0cm}\leq \p(\exists u\in \mathcal N_{t_0}:\n{X_{t_0}(u)-\bx}\leq 3)+\p(\forall u\in \mathcal N_{t_0-\tau_x}:\n{X_{t_0-\tau_x}(u)}\geq 2;\tau_x\leq t_0-\frac{a(x)}{2}).
\end{align*}
   Since a ball of radius $3$ can be covered by finitely many balls of radius $1$ in $\R^d$, the first probability is bounded by $\ee/2$ for $x$ large. Splitting the second probability on the events $\tau_x\in(j-1,j]$ we obtain
   \begin{align*}
       &\p\left(\forall u\in \mathcal N_{t_0-\tau_x}:\n{X_{t_0-\tau_x}(u)}\geq 2;\tau_x\leq t_0-\frac{a(x)}{2}\right)\\
       &\hspace{3cm}\leq \sum_{j=1}^\infty \p\left(\exists t\in \big(\frac{a(x)}{2}+j-1,\frac{a(x)}{2}+j\big),~\forall u\in \mathcal N_t,~\n{X_t(u)}\geq 2\right).
   \end{align*}
   Divide equally the interval $(a(x)/2+j-1,a(x)/2+j)$ into $q_x(j):=\exp((a(x)/2+j)/3)$ many intervals $\{I_\ell\}$ with endpoints $a(x)/2+j-1=t_0<\dots<t_{q_x(j)}=a(x)/2+j$. Denote by $\bf Z$ a  $d$-dimensional standard Gaussian random variable and write $u\mapsto v$ if $v$ is a descendant of $u$.  We have 
   \begin{align*}
       &\hspace{0.5cm}\p\left(\exists t\in \big(\frac{a(x)}{2}+j-1,\frac{a(x)}{2}+j\big),~\forall u\in \mathcal N_t,~\n{X_t(u)}\geq 2\right)\\
       & \leq \sum_{\ell=1}^{q_x(j)}\p(\forall u\in \mathcal N_{t_\ell},~\n{X_{t_\ell}(u)}\geq 1)+\p\left(\sup_\ell \sup_{\substack{s,t\in I_\ell\\ s>t}} \sup_{\substack{u\in \mathcal N_{t},v\in \mathcal N_s\\ u\mapsto v}} \n{X_t(u)-X_s(v)}\geq 1\right)\\
       &\leq \sum_{\ell=1}^{q_x(j)}\p(\forall u\in \mathcal N_{t_\ell},~\n{X_{t_\ell}(u)}\geq 1)+{q_x(j)}e^{2(a(x)/2+j)}\p\left(\n{\bf Z}\geq \sqrt{q_x(j)}\right)+o(1)\\
       &=\sum_{\ell=1}^{q_x(j)}\p(\forall u\in \mathcal N_{t_\ell},~\n{X_{t_\ell}(u)}\geq 1)+o(1),
   \end{align*}where in the second inequality we used Markov's inequality on the number of particles present at time $a(x)/2+j$. 
   By Theorem 2.1 of \citep{oz2023} applied with $a=k=\theta=0$ and $r_0=1$, for $x$ large,
   $$\p(\forall u\in \mathcal N_{t_\ell},~\n{X_{t_\ell}(u)}\geq 1)\leq e^{-(a(x)/2+j-1)/2}.$$
   Altogether, we obtain the bound \eqref{eq:lbtoprove}.
\end{proof}

\subsection{Proof of the remaining results from Section \texorpdfstring{\ref{sec: analytic}}{}}\label{sec:proofother}

\begin{proof}[Proof of Corollary \ref{coro:bbm general}]
  Fix a dimension $d\geq 1$, and let us denote by $\tau_{\kappa,s}^r(x)$ the FPT of a BBM with branching rate $\kappa$ and diffusivity $s$ to a ball of radius $r$ centered at $(x,0,\dots,0)\in\R^d$. In particular, $\tau_{\kappa,s}^1(x)=\tau_{\kappa,s}(x)$.  By self-similarity of the Brownian motion, we have the relations
    \begin{align*}
        \tau_{\kappa,s}^1(x)&\dd \tau_{\kappa,1}^{s^{-1}}\left(\frac{x}{s}\right)\qquad \text{ and }\qquad 
        \tau_{\kappa,s}
    ^1(x)\dd \frac{1}{\kappa}\tau_{1,s}^{\sqrt{\kappa}}\left(\sqrt{\kappa}x\right)    
    \end{align*}for every $\kappa,s>0$. 
    The proof then follows immediately from Theorem \ref{thm:BBM}, where we note that the same proof works if we replace the target $B_x$ by a ball centered at $(x,0,\dots,0)\in\R^d$ of a fixed radius $r>0$.
\end{proof}

\begin{proof}[Proof of Theorem \ref{thm:brw maximum}]
Note that the only difference between the classical BRW and our delayed branching BRW models is the branching structure, i.e., the underlying tree that describes the genealogy of the particles. The proof of the version of Theorem \ref{thm:brw maximum} for classical BRW depends on the branching structure only through the first moment and second moment estimates; see (28) and (29) of \citep{bramson2016convergence}. In our case, this can be adapted using Lemma \ref{lemma:number1} below.
\end{proof}

\begin{lemma}\label{lemma:number1}
   Conditioned upon survival, the expected number of particles $N_n$ at time $n$ for the delayed branching BRW model satisfies  $N_n\asymp{\rho}^n$,\footnote{For two sequences $\{A_n\}$ and $\{B_n\}$, we write $A_n\asymp B_n$ if there is a constant $C>0$ independent of $n$ such that $A_n/C\leq B_n\leq CA_n$ for all $n$.} where $\rho$ is given by \eqref{eq:newrho}.
\end{lemma}

\begin{proof}Let us first prove that $\widetilde{N}_n\asymp\rho^n$ where $\widetilde{N}_n$ is the expected value of $Z_n$, the number of particles at time $n$ (without conditioning upon survival). 
Recall from Section \ref{sec:newbrw} that in the delayed branching regime, each branching event consists of two branching sub-events at consecutive times, both into two branches. We call the first of the two sub-events the \emph{branching of type I}, and the second \emph{branching of type II}. 
     For $n\in\N$, let $\alpha_n$ be the expected number of branching sub-events of type I at time $n-1$. By construction, this is the same expected number of branching sub-events of type II at time $n$.  
Our definition of the delayed branching regime then leads to the following recursive equations of $(\widetilde{N}_n,\alpha_n)$:
\begin{itemize}
\item $\widetilde{N}_1=1,~\alpha_1=0$;
\item each particle at time $n$ that does not initiate a branching of type II independently has probability $\tl$ to create a   branching of type I, thus $\alpha_{n+1}=\tl(\widetilde{N}_n-\alpha_n)$;
    \item the increment of particles at time $n+1$ comes from contributions from branchings of types both I and II, meaning that $\widetilde{N}_{n+1}=2(1-\tnu)\alpha_n+(1+\tl-\tnu)(\widetilde{N}_n-\alpha_n)$.
\end{itemize}
In matrix form, we write
\begin{align*}
    \begin{bmatrix}
\widetilde{N}_{n+1} \\
\alpha_{n+1} 
\end{bmatrix}&=\begin{bmatrix}
1+\tl-\widetilde{\nu}  & 1-\tl-\widetilde{\nu} \\
\tl & -\tl 
\end{bmatrix}\begin{bmatrix}
\widetilde{N}_{n} \\
\alpha_{n} 
\end{bmatrix}\\
&=\begin{bmatrix}
1+\tl-\widetilde{\nu}  & 1-\tl-\widetilde{\nu} \\
\tl & -\tl 
\end{bmatrix}^n\begin{bmatrix}
1 \\
0 
\end{bmatrix}=:M(\tl,\widetilde{\nu} )^n\begin{bmatrix}
1 \\
0 
\end{bmatrix}.
\end{align*}
The largest eigenvalue of the matrix $M(\tl,\widetilde{\nu} )$ is precisely ${\rho}(\tl,\widetilde{\nu})$ defined in \eqref{eq:newrho}, and we conclude that $\widetilde{N}_n\asymp\rho^n$.

Denote by $p=1-q=\p(S)$. We first prove that $q<1$ if $\rho>1$. Let $q_n=\p(Z_n=0)$, so that $q_n\uparrow q$. By construction of the delayed branching regime, we have  $q_1=\widetilde{\nu}$, $q_2=\widetilde{\nu}+(1-\tl-\widetilde{\nu})q_1+\tl\widetilde{\nu}^2q_1$, and for $n\geq 1$,
\begin{align}
    q_{n+2}=\widetilde{\nu}+(1-\tl-\widetilde{\nu})q_{n+1}+\tl(1-\widetilde{\nu})^2q_{n+1}q_n^2+\tl\widetilde{\nu}^2q_{n+1}+2\tl\widetilde{\nu}(1-\widetilde{\nu})q_{n+1}q_n.\label{eq:recursionqn}
\end{align}
Consider the equation
$$q=h(q):=\widetilde{\nu}+(1-\tl-\widetilde{\nu})q+\tl(1-\widetilde{\nu})^2q^3+\tl\widetilde{\nu}^2q+2\tl\widetilde{\nu}(1-\widetilde{\nu})q^2.$$
It is elementary to check that $h(0)=\widetilde{\nu}\geq 0$, $h(1)=1$, and that if $\rho>1$, then $h'(1)<1$. In particular, there exists a solution  $\hat{q}\in(\widetilde{\nu},1)$ to the equation $q=h(q)$. It follows from \eqref{eq:recursionqn} that if $q_n\leq q_{n+1}\leq \hat{q}$, then $q_{n+2}\leq\hat{q}$. By induction, we know that $q_n\leq\hat{q}$ for each $n$. Therefore, $q=\lim q_n\leq \hat{q}<1$.

Observe that $\widetilde{N}_n=pN_n+q\E[Z_n\mid S^c]$. Since we have the bound 
$$\p(Z_n>y\mid S^c)=\frac{\p(Z_n>y,\,S^c)}{q}\leq \frac{\p(S^c\mid Z_n>y)}{q}\leq q^{y-1},$$
it holds $\E[Z_n\mid S^c]\leq C(q)$. Hence, we conclude that $N_n\asymp \rho^n$.
\end{proof}

\section{Numerical validation of the  first passage time  asymptotics}\label{sec:validation}

In this appendix, we provide numerical verification of the FPT asymptotics for BRW \eqref{eq:realtauxasymp} and BBM \eqref{eq:bbm asymp}. In particular, we show that the path purging in our numerical algorithm affects very little the FPT.

\subsection{Validation of the  asymptotic \texorpdfstring{\eqref{eq:realtauxasymp}}{}}\label{sec:brw validation}
Recall that the coefficient $1/c_1$ of the linear term for the estimation \eqref{eq:realtauxasymp} of $\tau_x$ can be computed through the relation $I(c_1)=\log\rho$. Recall also that  $\widetilde{\nu}=2/l_c$ is the termination rate and $${\rho}={\rho}(\tl,\widetilde{\nu})=\frac{1-\widetilde{\nu} }{2}+\sqrt{\frac{(1-\widetilde{\nu} )^2}{4}+2\tl(1-\widetilde{\nu} )}.$$
The implicit relation between $c_1$ and $\tl$ can thus be computed for the BRW and compared against the numerical BRW implementation. 

Next, we compute numerically the relation between $c_1$ and $\tl$. For any given $\tl$,   we calculate the SP distribution for 1000 paths at different offset distances for $q_x \in [20, 60]$. The mean SP (same as the FPT for a unit length jump per unit time), $\tau(q_x)$, is then fit to a function of the form,
\begin{equation}
    \tau(q_x) =  \frac{q_x}{c_1} + B \log(q_x) + C. \label{eq:fit_tau_approx}
\end{equation}
The parameter $\tl$ is obtained by uniformly sampling 19 points from  $ [0.05,0.95]$.

\begin{figure}[ht!]
    \centering
    \subfigure[]{
      \includegraphics[width=0.41\textwidth]{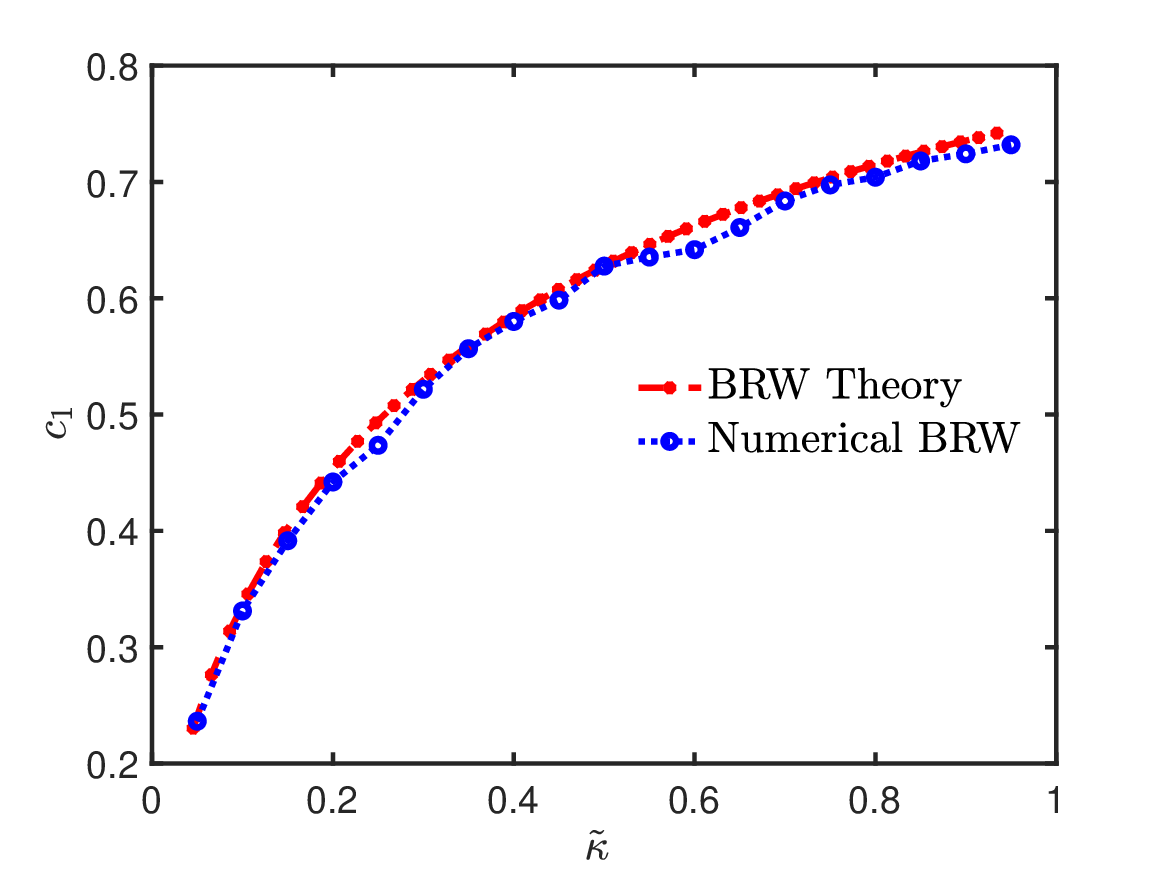}}
      \subfigure[]{\includegraphics[width=0.41\textwidth]{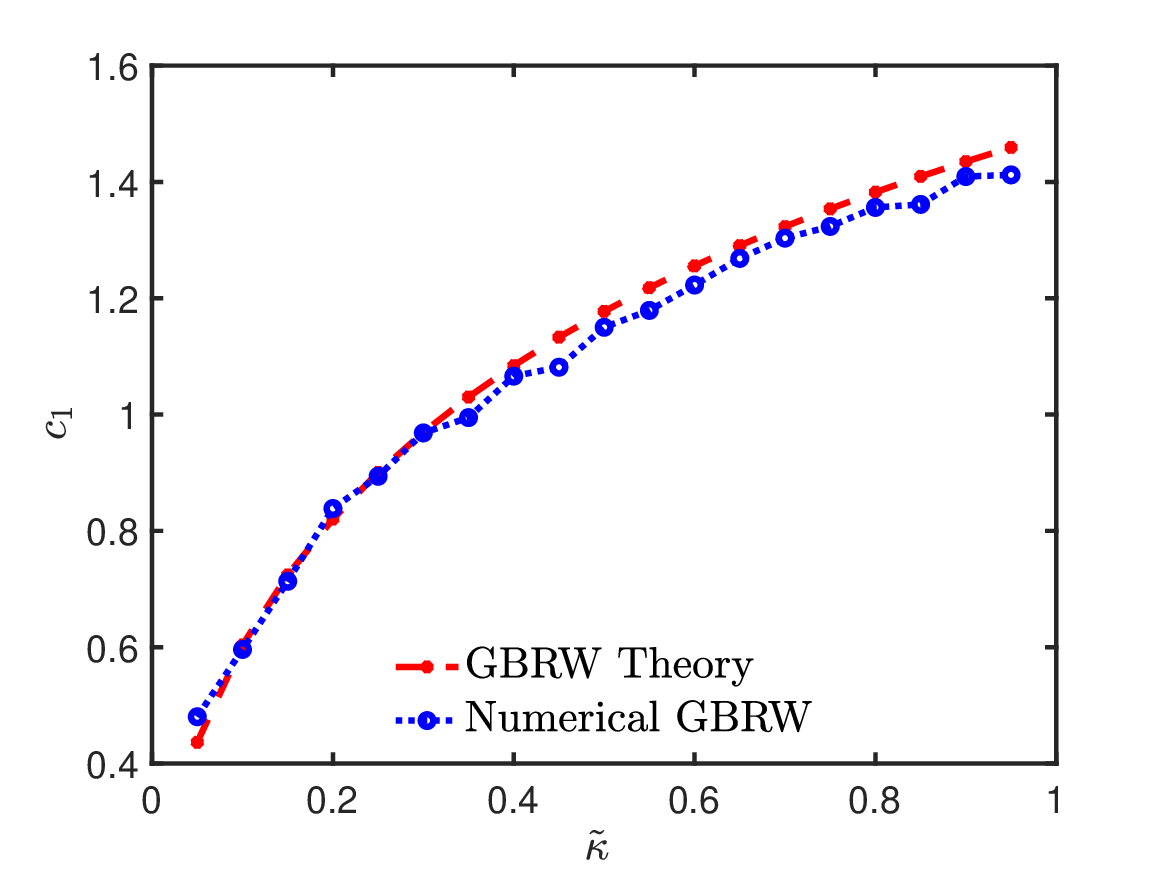}}
    \caption{The numerical BRW obtained $c_1$ compared against the reference calculation \eqref{eq:realtauxasymp} at different branching rate $\tl $ for the  BRW model with increments (a) uniformly distributed on $\S^2$ (BRW); (b) Gaussian distributed on $\R^3$ (GBRW). 
    }
    \label{fig:c1_rho_val}
\end{figure}
We see that the numerically obtained coefficient $c_1$ of the linear term is in decent agreement with the reference theory estimates for the BRW, as shown in Figure~\ref{fig:c1_rho_val}. This indicates that the numerical approximation of {path purging} does not affect the linear coefficient of the scaling behavior of the SP or FPT distribution of the numerical BRW, and can be used as a suitable approximation for carrying out the numerical BRW to represent a CGMD network.

\subsection{Validation of the  asymptotic \texorpdfstring{\eqref{eq:bbm asymp}}{}}
\label{sec:bbm validation}

We use a similar approach to verify the theoretical prediction \eqref{eq:bbm asymp} for the FPT of the standard BBM. Note that here we do not include the termination or delayed branching regimes. With the choice of $\sigma=1$, the coefficient $1/c_1$ for \eqref{eq:bbm asymp} is equal to $1/\sqrt{2\tl}$ for $\tl>0$.

For $\tl$ sampled uniformly in $[0.05,0.95]$, we compute numerically the mean first passage times $\tau(q_x)$ at different offset distances $q_x\in[20,60]$ and fit them to a function of the form \eqref{eq:fit_tau_approx}. The linear coefficient demonstrates a decent agreement with the theoretical prediction, as shown in Figure~\ref{fig:bbm_c1_rho_val}. 
\begin{figure}[ht!]
    \centering
     {\includegraphics[width=0.41\textwidth]{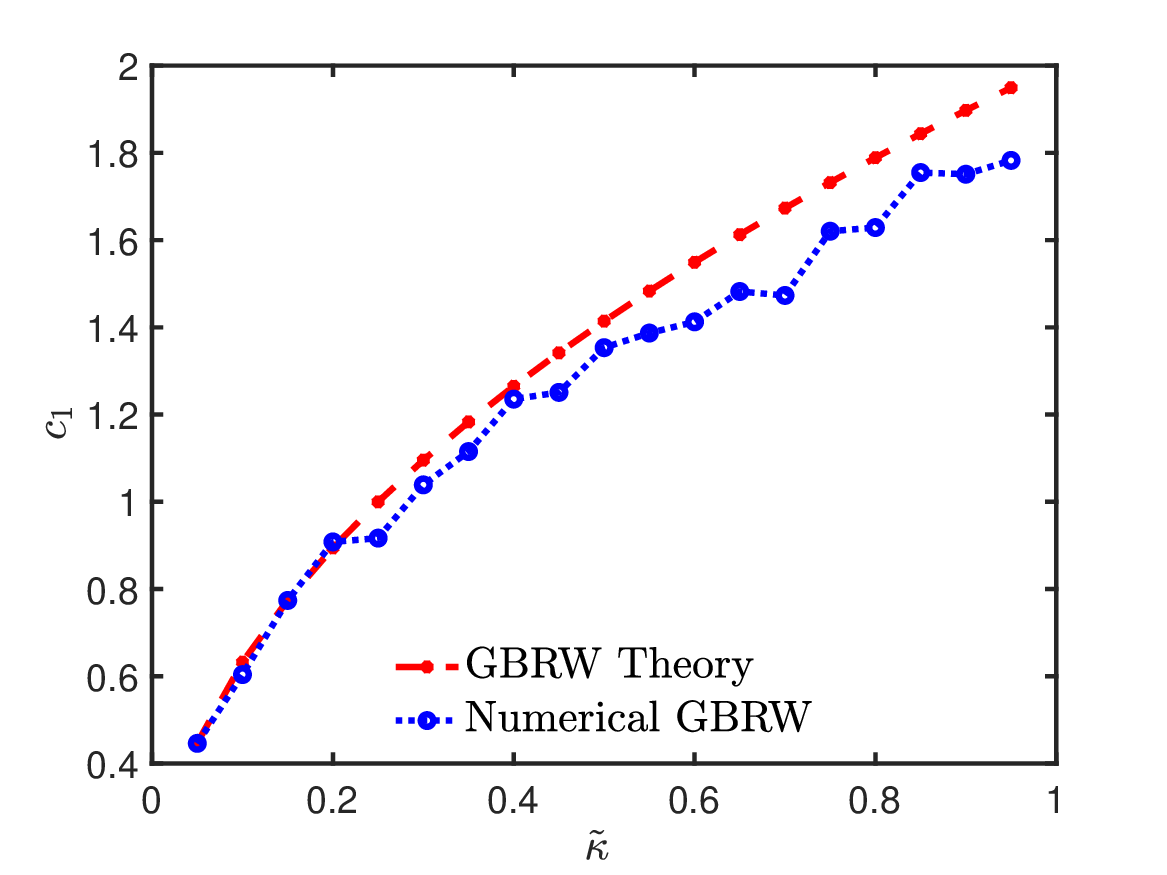}}
       
    \caption{The $c_1$ obtained numerically from the   BBM  compared against the theoretical estimate at different branching rates $\tl $ for the  BRW model. 
    }
    \label{fig:bbm_c1_rho_val}
\end{figure}
Again, this supports that path purging has a negligible effect on the distribution of the FPT. The divergence of the behavior of the numerical implementation from the theoretical estimate occurs due to the pseudo-continuous implementation of time (discrete time steps of 0.1). The effect of discretization is negligible at lower branching rates (even with a time step of 0.25 for $\tl\leq0.1$), but becomes more apparent for $\tl>0.3$ which is well above the cross-link density in realistic polymeric systems that are simulated using the CGMD method.

\section{Intercept of the linear dependency}\label{sec:intercept}
We have used \textit{linear dependency}, $\overline{c}_1$, to classify the quality of agreement of the numerical model to the CGMD calculations, as shown in Figures~\ref{fig:corr_c1_rho_cgmd_b},~\ref{fig:scaled_c1_rho_cgmd_b} and~\ref{fig:bbm_c1_rho_cgmd_b}. However, the same linear dependency $\overline{c}_1$ may represent a family of straight lines with a slope of $1/\overline{c}_1$ but distinct intercepts. As a result, to fully characterize the SP mean we look at the \emph{intercept}, SP$(q_x=0)$. We see that the three presented numerical models are in agreement with the CGMD results, barring at lower cross-link densities as shown in Figure~\ref{fig:intercepts}.
\begin{figure}[ht!]
    \centering
    \hspace{-1cm}
       \subfigure[]{\includegraphics[width=0.35\textwidth]{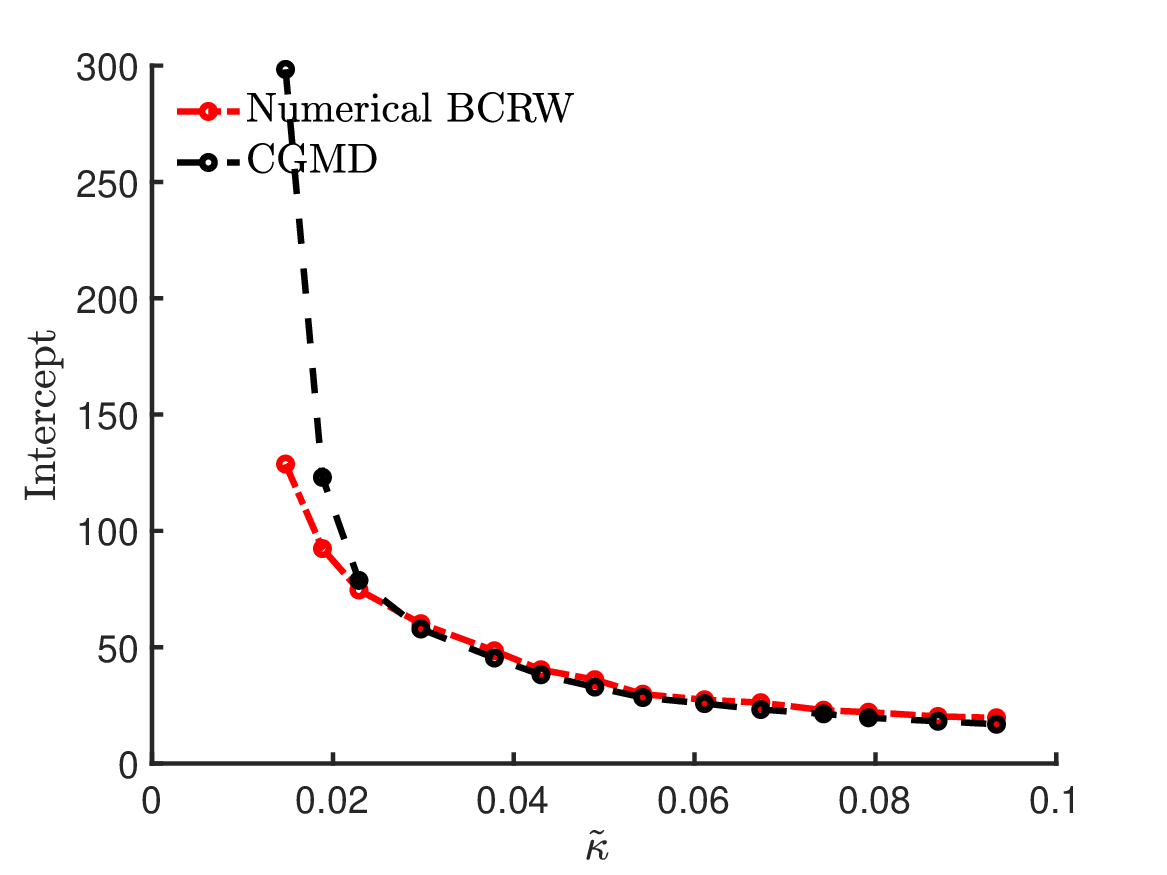}}
       \hspace{-0.6cm}
        \subfigure[]{\includegraphics[width=0.35\textwidth]{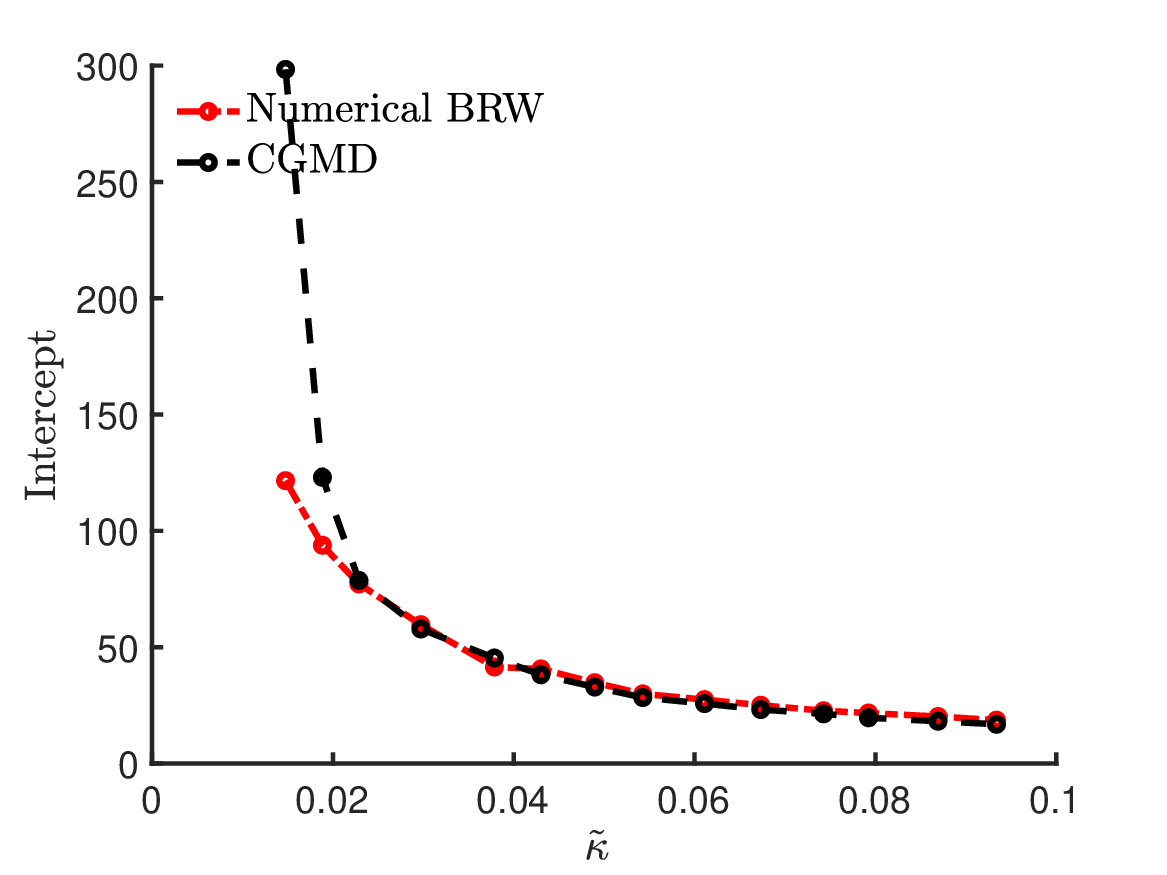}}
         \hspace{-0.6cm}
        \subfigure[]{\includegraphics[width=0.35\textwidth]{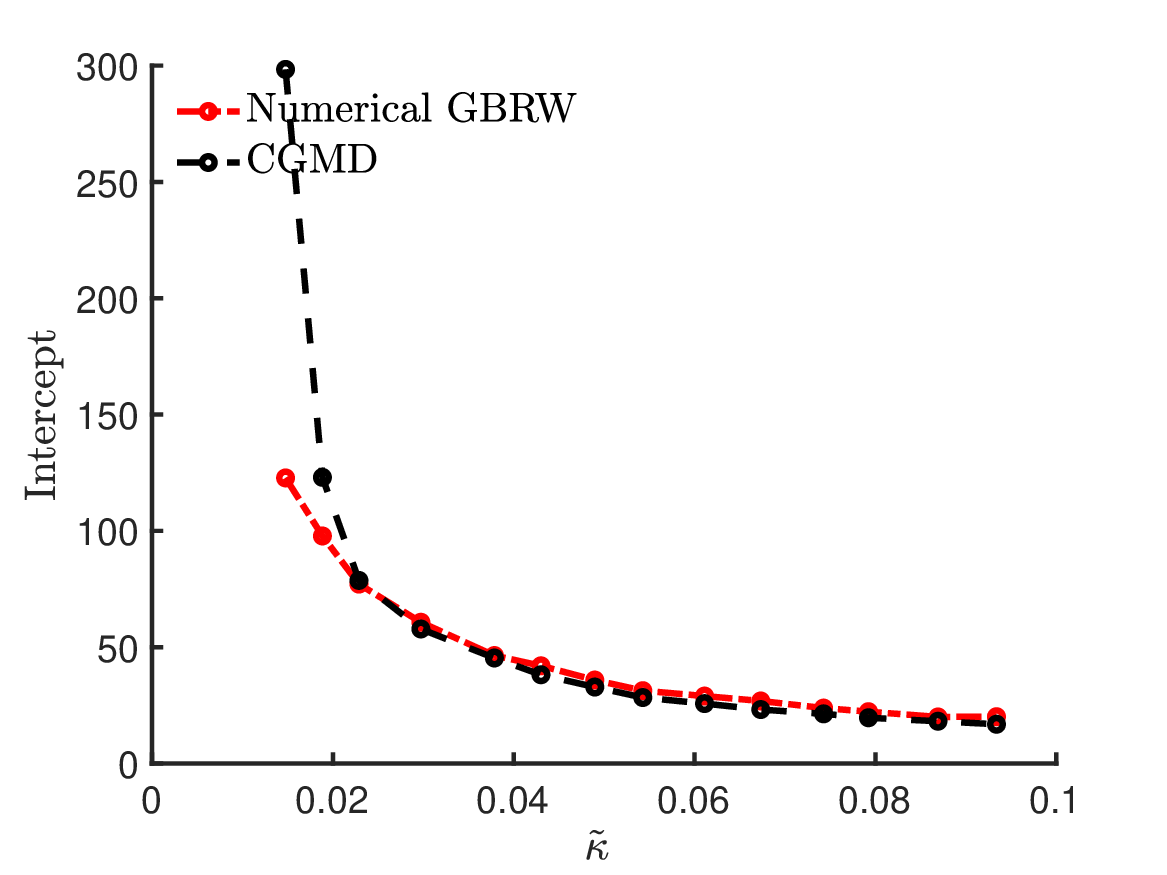}}
         \hspace{-1cm}
    \caption{ The intercepts of the SP$(q_x)$ for the (a) BCRW, (b) scaled  BRW, and (c) scaled GBRW methods. 
    }
    \label{fig:intercepts}
\end{figure}

\section{Shortest path in the 8-chain model} \label{app:abm_comp}
In this appendix, we present the comparison of the theoretical estimates of $\overline{c}_1$ obtained from the spatial branching processes (scaled BRW, GBRW, and BBM). The intuition is that the idealized and simplified 8-chain (a.k.a.~Arruda-Boyce) model places all chains along the shortest path. This makes the shortest path between distance nodes much longer than the spatial branching models, for the realistic range of branching rates encountered in the CGMD simulations ($\tl < 0.1$). 
\begin{figure}[ht!]
    \centering
     {\includegraphics[width=0.41\textwidth]{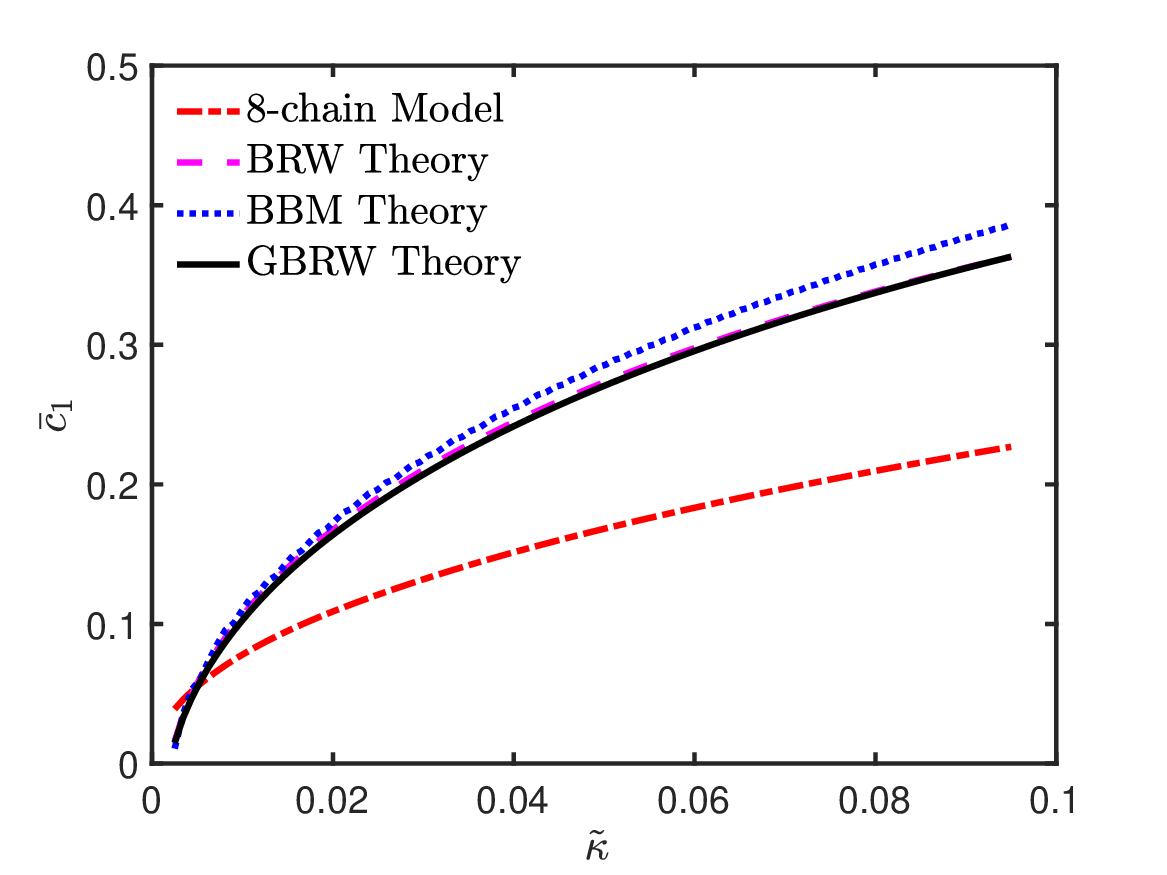}}
       
    \caption{The $\overline{c}_1$ computed analytically from the theoretical estimate at different branching rates $\tl $ for the 8-chain, scaled BRW, GBRW, and BBM models. The scaling of the jumps as a function of the MSID has been accounted for resulting in an additional ${1}/{\sigma}$ factor multiplying the estimates in \eqref{eq:bbm_lc} and~\eqref{eq:gbrw_lc}.
    }
    \label{fig:comp_arruda}
\end{figure}
The resultant $\overline{c}_1$ is expected to be much lower in the 8-chain model as a consequence and this is confirmed by the theoretical estimate of $\overline{c}_1$ from the 8-chain model and the other spatial branching processes, as shown in Figure~\ref{fig:comp_arruda}.

\section{Approximating GBRW with BBM}\label{sec:GBRW as BBM}

In this appendix, we discuss how to approximate our GBRW model (see Section \ref{sec:cgmd as gbrw} for details) by a standard BBM with an implied branching rate. 
 Recall from Section \ref{sec:newbrw} that we introduced the extra features of termination and delayed branching property. The analogies for the BBM can be summarized as follows:
\begin{itemize}
    \item termination: each existent particle carries an independent exponential clock with rate $\tnu=1/250$, representing the termination of the particle;
    \item delayed branching property: the branching events occur according to exponential clocks with parameter $\tl$. Each branching event consists of two sub-events: the particle first branches into two and then one of the two descendants branches into two after a unit of time (within this unit of time, termination could happen but no extra branching event will occur).
\end{itemize}

The resulting model will be called the $(\tl,\tnu)$-BBM. 
Unfortunately, both the termination and the delayed branching properties are not handy to deal with when analyzing the FPT, due to the reminiscence of the connection to Fisher-KPP equations. Nevertheless, we still expect that a correspondent asymptotic result of the form \eqref{eq:realtauxasymp}  with an $O_\bP(\log\log x)$ error term should at least hold true.

For clarity of our discussions, we take the following detour. It is not unreasonable to approximate the $(\tl,\tnu)$-BBM model (conditioned upon non-extinction) with the classical BBM model that carries the same asymptote for the expected number of particles. Consider the $(\tl,\tnu)$-BBM where $2\tl>\tnu$,  and denote by $n(t):=\E[\#\mathcal N_t]$, the expected number of particles at time $t$.

\begin{proposition}\label{prop:bbmnumberof particles}
    Denote by $\lambda_0>0$ the unique real root to $\lambda_0=-\tnu+2\tl e^{-\lambda_0}$.  It holds that $n(t)\asymp e^{\lambda_0t}$. 
\end{proposition}

\begin{proof}
    By our construction, and conditioning on the first branching/termination event, there is
    \begin{align*}
        n(t)=e^{-t(\tl+\tnu)}&+\int_0^t\tl \,e^{-(\tl+\tnu)s}n(t-s)\,\d s\\
        &+2\int_0^t\tl \,e^{-(\tl+\tnu)s}\left(n(t-s-1)\bone_{\{s\leq t-1\}}+\bone_{\{t-1<s\leq t\}}\right)\d s.
    \end{align*}
Duhamel's principle then yields the delay differential equation
\begin{align}
    n'(t)=-\tnu \, n(t)+2\tl  \left(n(t-1)\bone_{\{t\geq 1\}}+\bone_{\{0\leq t<1\}}\right).\label{eq:dde}
\end{align}
Let us introduce the solution $\phi(t)=\E[\#\mathcal N_t]$ for $0\leq t\leq 1$, where $1\leq\phi(t)\leq C(\tl)$ for some $C(\tl)>0$. The solution to \eqref{eq:dde} then satisfies the linear autonomous impulsive delay differential equation
\begin{align}
    n'(t)=-\tnu \, n(t)+2\tl \, n(t-1),\ n|_{[0,1]}=\phi|_{[0,1]}.\label{eq:dde2}
\end{align}

Recall our definition of $\lambda_0$, which is the root of the characteristic equation associated with \eqref{eq:dde2}. That $\lambda_0>0$ follows from $2\tl>\tnu$. 
 Theorem 1 of \citep{yenicceriouglu2020asymptotic} then leads to $\lim_{t\to\infty}n(t)e^{-\lambda_0t}=C(\phi)$ for some constant $C(\phi)$ depending on $\phi$ that is uniformly bounded when $\tl\leq 1$.      In other words, $n(t)\asymp e^{\lambda_0t}$.
\end{proof}

Upon solving for $\lambda_0$, the $(\tl,\tnu)$-BBM can be approximated by a scaled BBM with binary branching and implied branching rate $\lambda_0$. The diffusivity of the BBM is now set to $s=\sqrt{\mathrm{MSID}(1/\tl)/3}$ in view of the central limit theorem. This explains the regime behind the \emph{BBM Theory} curve in Figure \ref{fig:bbm_c1_rho_cgmd_b}.



\end{document}